\documentclass[a4paper]{article}
\usepackage{geometry}
\usepackage{amssymb}
\usepackage{hyperref}
\usepackage{amsmath}
\usepackage{amsthm}
\usepackage{bbold}
\usepackage[pagewise]{lineno}
\usepackage[lgrmath]{mathgreeks}
\usepackage{array}
\usepackage{caption}
\usepackage{lipsum}
\usepackage[colorinlistoftodos]{todonotes}

\usepackage[ruled]{algorithm}
\usepackage{algorithmic}

\usepackage{tikz}
\usetikzlibrary{datavisualization}
\usetikzlibrary{plotmarks}

\newcommand{\bZ}{{\mathbb{Z}}}
\newcommand{\bQ}{{\mathbb{Q}}}

\newcommand{\bC}{{\mathbb{C}}}
\newcommand{\bN}{{\mathbb{N}}}

\newcommand{\cR}{{\mathcal{R}}}

\newcommand{\cP}{{\mathcal{P}}}
\newcommand{\cT}{{\mathcal{T}}}
\newcommand{\cL}{{\mathcal{L}}}

\newcommand{\cS}{{\mathcal{S}}}

\newcommand{\bT}{{\mathbb{T}}}

\newcommand{\bP}{{\mathbb{P}}}
\newcommand{\bH}{{\mathbb{H}}}
\newcommand{\bI}{{\mathbb{I}}}
\newcommand{\bfxi}{{\mathgreeksfont{enc=LGR,fam=Alegreya-LF,shape=n}{\xi}}}
\newcommand{\bfeta}{{\mathgreeksfont{enc=LGR,fam=Alegreya-LF,shape=n}{\eta}}}
\newcommand{\bfomega}{{\mathgreeksfont{enc=LGR,fam=Alegreya-LF,shape=n}{\omega}}}
\newcommand{\bfOmega}{{\mathgreeksfont{enc=LGR,fam=Alegreya-LF,shape=n}{\Omega}}}
\newcommand{\cH}{{\mathcal{H}}}

\newcommand{\vw}{{w}}
\newcommand{\vu}{{u}}
\newcommand{\vv}{{v}}
\newcommand{\vx}{{x}}
\newcommand{\vy}{{y}}

\newcommand{\lc}{\operatorname{lc}}

\newcommand{\len}{\operatorname{len}}

\newcommand{\hc}{\operatorname{hc}}

\newcommand{\im}{\operatorname{im}}

\newcommand{\spa}{\operatorname{span}}

\newcommand{\den}{\operatorname{den}}
\newcommand{\Deg}{{\operatorname{Deg}}}
\newcommand{\num}{\operatorname{num}}

\newcommand{\tdeg}{\operatorname{tdeg}}
\newcommand{\hdeg}{\operatorname{hdeg}}

\newcommand{\Li}{\operatorname{Li}}

\newcommand{\tc}{\operatorname{tc}}

\newcommand{\residue}{\operatorname{residue}}

\newtheorem{thm}{Theorem}[section]
\newtheorem{prop}[thm]{Proposition}
\newtheorem{cor}[thm]{Corollary}
\newtheorem{lemma}[thm]{Lemma}
\newtheorem{convention}[thm]{Convention}
\newtheorem{remark}[thm]{Remark}
\newtheorem{define}[thm]{Definition}
\newtheorem{example}[thm]{Example}

\newtheorem{alg}[thm]{Algorithm}
\newtheorem{hyp}[thm]{Hypothesis}
\newtheorem{out}[thm]{Outline}

\begin{document}

\title{Complete Reduction for Derivatives  in a Transcendental Liouvillian Extension\footnote{ This research was funded in part by the National Key R\&D Programs of China (No. 2023YFA1009401), 
the Natural Science Foundation of China (No.~12271511, No.~12201065, No.~12101105),  the Austrian Science Fund (FWF) 10.55776/PAT1332123 and by the Natural Science Foundation of Fujian Province of China (No.\ 2024J01271).
This work was also supported by the International Partnership Program of Chinese Academy of Sciences (Grant No. 167GJHZ2023001FN).
}}

\author{Shaoshi Chen\textsuperscript{a,b},  Hao Du\textsuperscript{c},
Yiman Gao\textsuperscript{d,}\footnote{Corresponding author; Present address: Research Institute for Symbolic Computation (RISC),
Johannes Kepler University, Linz, Austria},  Hui Huang\textsuperscript{e}, \\ Wenqiao Li\textsuperscript{a,b}, Ziming Li\textsuperscript{a,b}  }
\hypersetup{pdfauthor={Shaoshi Chen, Hao Du, Yiman Gao, Hui Huang, Wenqiao Li, Ziming Li}}
\date{}

\maketitle

\begin{center}
\textsuperscript{a} MMRC, Academy of Mathematics and Systems Science, Chinese Academy of Sciences, Beijing (100190), China\\
\textsuperscript{b} School of Mathematical Sciences,  University of Chinese Academy of Sciences, Chinese Academy of Sciences, Beijing (100049), China\\
\textsuperscript{c}  School of Mathematical Sciences,  Beijing University of Posts and Telecommunications, Beijing (102206), China \\
\textsuperscript{d}  Johannes Kepler University Linz, Research Institute for Symbolic Computation (RISC),
  Altenberger Stra\ss e 69, 4040, Linz, Austria \\
\textsuperscript{e} School of Mathematics and Statistics, Fuzhou University, Fuzhou (350108), China \\[\medskipamount]
Email: \href{mailto:schen@amss.ac.cn}{\tt schen@amss.ac.cn}, \href{mailto:haodu@bupt.edu.cn}{\tt haodu@bupt.edu.cn},
\href{mailto: ymgao@risc.uni-linz.ac.at}{\tt ymgao@risc.jku.at},
\href{mailto: huanghui@fzu.edu.cn}{\tt huanghui@fzu.edu.cn},
\href{mailto: liwenqiao@amss.ac.cn}{\tt liwenqiao@amss.ac.cn},
\href{mailto: zmli@mmrc.iss.ac.cn}{\tt zmli@mmrc.iss.ac.cn}
\end{center}

\begin{abstract}
Transcendental Liouvillian extensions are differential fields, in which one can model
poly-logarithmic, hyperexponential, and trigonometric functions, logarithmic integrals, and their (nested) rational expressions.
For such an extension, we construct, over the subfield of constants, a complement of the subspace of derivatives, and develop an algorithm that decomposes any element of the field into the sum of a derivative and a component lying in the complement. Consequently, an element is a derivative if and only if its complementary component vanishes. Moreover, the algorithm enables us to determine elementary integrability over the extension by computing parametric logarithmic parts, and leads to a reduction-based approach to constructing telescopers for 
elements in the extension, provided that  an a priori order bound is given.
\end{abstract}


\section{Introduction} \label{SECT:intro}

 Symbolic integration as a classical topic of computer algebra  is to express indefinite integrals in closed forms.
Historical developments and fundamental results on this topic are surveyed in~\cite{RaabSingerBook}.
From a computational perspective, the most systematic  advances are due to 
Risch \cite{Risch1969,Risch1970} and numerous researchers who  have refined, enhanced, extended, and implemented Risch’s algorithm \cite{RothsteinThesis, CherryThesis, TragerThesis, Davenport1986, Bronstein1990a, Bronstein1990b, BronsteinBook, RaabThesis}.

%
For an elementary function,  Risch's algorithm
determines whether it has an elementary integral, and computes such an integral if there exists one. 
\begin{example} \label{EX:log1}
Let
$ \displaystyle f(x)= \frac{x\log(x)^3+1}{x\log(x)}$ and  $ \displaystyle g(x) = \dfrac{x\log(x)^3+1}{(x+3)\log(x)}.$
By Risch's algorithm,
$$ \int f \,  dx = x \log(x)^2 -2x \log(x) +2x+ \log(\log(x))
\quad \text{and} \quad
\int g \,  dx = \int \dfrac{x\log(x)^3+1}{(x+3)\log(x)} \,  dx.$$
The algorithm yields an elementary integral of $f$, and reveals that $g$ has no elementary integral.
\end{example}
The most basic routine in Risch's algorithm is the Hermite-Ostrogradsky reduction for rational functions \cite{Hermite1872,Ostrogradsky1845}.
We refer to  \cite[Chapter 2]{BronsteinBook} and \cite[Chapter 11]{BlueBook} for its modern presentations.

In the rest of this section, let $C$ be a field of characteristic zero.
\begin{example} \label{EX:rational}
For every element $f \in C(x)$, the Hermite-Ostrogradsky reduction computes two elements $g, r \in C(x)$ such that
$f = g' + r,$
where $^\prime$ stands for the usual derivation $d/dx$, and
$r$ is  an $x$-proper fraction with a squarefree denominator.
Moreover, $r$ is unique, and 
$f \in C(x)^\prime$ if and only if $r=0$, where $C(x)^\prime$ stands for the $C$-subspace consisting of  derivatives in $C(x)$.

Set $S_x := \{ s \in C(x) \mid \text{$s$ is  $x$-proper and has a squarefree denominator} \}$.
Then $S_x$ is a $C$-subspace and $C(x) = C(x)^\prime \oplus S_x$ by the Hermite-Ostrogradsky reduction. 
The projection from~$C(x)$ to $S_x$ with respect to the direct sum
is a complete reduction for $C(x)^\prime$ 
in the sense of \cite[Definition 5.67]{KauersBook}, which is recalled below.
\end{example}

\begin{define} \label{DEF:crK}
Let $V$ be a $C$-linear space and $U$ be a subspace. A linear operator $\Phi: V \rightarrow V$ 
is called a {\em complete reduction} for $U$ if $v - \Phi(v) \in U$ for all $v \in V$ and $\ker(\Phi) = U$. 
\end{define}

A straightforward linear algebra argument reveals that a $C$-linear operator $\Phi$ on $V$ is a complete reduction for~$U$ if and only if 
$V=U \oplus \im(\Phi)$ and $\Phi$ is  the projection from $V$ onto $\im(\Phi)$ with respect to this direct sum.
In particular, all complete reductions are idempotent.

In the present paper, the subspace $U$ in the above definition is typically the image of some linear operator on $V$.
We therefore adopt the following definition for convenience.
\begin{define} \label{DEF:cr}
Let $V$ be a $C$-linear space and $\cL$ be a linear operator on $V$.
Another linear operator $\Phi$ on $V$ is called a 
{\em complete reduction for the pair $(V, \cL)$} if $\Phi$ is a complete reduction for $\im(\cL)$.
In other words, 
$V = \im(\cL) \oplus \im(\Phi)$ and $\Phi$ is the projection 
from $V$ onto $\im(\Phi)$ with respect to this direct sum.
\end{define}
The next definition describes how a complete reduction decomposes an element additively.
\begin{define} \label{DEF:Rpair}
Let $\Phi$ be a complete reduction for the pair $(V,\cL)$ and let $u,v \in V$ be such that $v = \cL(u) + \Phi(v)$.
We call $\Phi(v)$ the {\em remainder of $v$ with respect to $\Phi$}, and $(u, \Phi(v))$ a 
{\em reduction pair of $v$ with respect to $\Phi$} (or an {\em R-pair with respect to $\Phi$} for brevity).
\end{define}

\begin{remark} \label{RE:algcr}
Given a linear operator $\cL$ on a linear space $V$, we say that {\em there is an algorithm for constructing a complete reduction
for $(V, \cL)$} if we can fix a complement $W$ of $\im(\cL)$ in $V$ and develop an algorithm that, for every $v \in V$, computes 
$u \in V$ and $w \in W$ such that $v = \cL(u) +w$. Such an algorithm induces a complete reduction 
for $(V, \cL)$ by mapping  $v$ to $w$.
\end{remark}

The Hermite-Ostrogradsky reduction  in Example \ref{EX:rational}
not only induces a complete reduction for $(C(x), \, ^\prime)$, but also yields an algorithm
for computing
the corresponding R-pairs.
Based on this reduction, an algorithm is developed to construct telescopers for bivariate rational functions (see \cite{BCCL2010}).
The algorithm separates the computation of telescopers from the computation of certificates. This is desirable in the typical situation
where we are only interested in telescopers. Such an advantage 
has motivated a rapid development of complete reductions in various settings.
\begin{example} \label{EX:hyperexp}
Let $h$ be a nonzero rational function in $C(x)$, and
$V$ be the space spanned by the hyperexponential function $\exp(\int h \, dx)$ over $C(x)$.
The Hermite reduction for hyperexponential functions developed in \cite{BCCLX2013} gives a complete reduction 
for $(V, \, d/dx)$.
\end{example}

Recall that a {\em derivation} $^\prime: R \rightarrow R$ on a commutative ring $R$ is an additive map satisfying the usual Leibniz's rule:
$(ab)^\prime = a^\prime b  + a b^\prime$ for all $a, b\in R$.
The pair $(R, \, ^\prime)$  is called a {\em differential ring} and a {\em differential field} if $R$ is a field.
The set $\{c\in R\mid  c^\prime=0\}$ forms a subring whose elements are called {\em constants}.
This subring of constants becomes a subfield if $R$ itself is a field.

Let $(R, \, ^\prime)$ be a differential ring whose  subring $C$ of constants is a field. Then $R$ is a $C$-linear space.
Complete reductions have been established for such differential rings with instances of~$R$ including
the field of differential fractions \cite{BLLRR2016}, the field of algebraic functions \cite{CKK2016},
(sub)rings of D-finite functions \cite{CHKK2018,vdHJ2021,CDK2023}, exponential
 towers \cite{GaoThesis} and primitive towers \cite{DGLL2025}.
A complete reduction for the pair $(C(x), \cL)$  is developed in \cite{BCPS2018},
where $\cL$ denotes a linear differential operator. 
The reader is  referred  to \cite{BS2024,CGHS2025,CDKW2025} for complete reductions related to symbolic summation.

A complete reduction for $(R, \, ^\prime)$ decomposes $f\in R$ as $f=g^\prime  + r$ for some $g, r \in R$, where $r$ is minimal
with respect to a partial ordering defined in \cite[Section 1]{DGLL2025}.
In particular, $f$
is a derivative in $R$ if and only if $r=0$. 
So a complete reduction for $(R,  \, ^\prime)$  is an additive decomposition 
in the sense of \cite{CDL2018,DGLW2020}.
This offers an alternative, and potentially more informative, way to express integrals.
\begin{example} \label{EX:log2}
By viewing the functions $f$ and $g$ in Example \ref{EX:log1} as elements of
the differential field $\bQ(x, \log(x))$ with derivation $^\prime=d/dx$,
the complete reduction in \cite{DGLL2025} yields
\[ f(x) = h(x)^\prime + \frac{1}{x\log(x)} \quad \text{and} \quad
g(x) = h(x)^\prime + \frac{1}{(x+3) \log(x)} - \frac{3 \log(x)^2}{x+3}, \]
where $h(x) =  x \log(x)^2 -2 x \log(x) +2x$.

Note that $f(x)$ is the sum of a derivative in $\bQ(x, \log(x))$ and
a remainder $1/(x\log(x))$, which has an elementary integral $\log(\log(x))$.
By contrast, $g(x)$ is the sum of a derivative in $\bQ(x, \log(x))$ and a remainder, which is not elementarily
integrable. The elementary integrability of such remainders in logarithmic extensions can be determined
by computing parametric logarithmic parts (see \cite[Theorem 5.3]{DGLL2025}).
Unlike Risch's algorithm taken in Example \ref{EX:log1},
the complete reduction  renders us an integrable part $h(x)$ of $g(x)$ in~$\bQ(x, \log(x))$
and a minimal non-integrable part with respect to the partial ordering defined in \cite[Section 1]{DGLL2025}.

Both decompositions can also be obtained from
the algorithms in \cite{CDL2018,DGLW2020}.
\end{example}

We are interested in developing complete reductions on transcendental Liouvillian extensions. 
Such extensions are more general than transcendental elementary extensions
(see \cite[Definition 5.1.4]{BronsteinBook}), and do not contain any complicated reduced elements (see \cite[Definition 3.5.2]{BronsteinBook}).
\begin{define} \label{DEF:liouvillian}
A differential field $F_n=C(t_1, \ldots, t_n)$ with derivation $^\prime$
is called a {\em transcendental Liouvillian extension} of $C$ if
\begin{itemize}
\item[(i)] $t_1, \ldots, t_n$ are algebraically independent over $C$,
\item[(ii)] for every  $i \in \{1, \ldots, n\}$, either $t_i^\prime \in F_{i-1}$ or $t_i^\prime/t_i \in F_{i-1}$, 
where $F_{i-1}= C(t_1, \ldots, t_{i-1})$,
\item[(iii)] $C$ is the subfield of constants in $F_n$.
\end{itemize}
We call $t_i$ a {\em primitive} (resp.\ {\em hyperexponential}) {\em generator} if 
$t_i^\prime \in F_{i-1}$ (resp.\ $t_i^\prime/t_i \in F_{i-1}$). Moreover, $F_n$ is called a {\em primitive} (resp.\ {\em hyperexponential}) {\em tower}
if all of the $t_i$'s are primitive (resp.\ hyperexponential). 
\end{define}

 Due to the presence of hyperexponential generators, a more general inductive framework is indispensable (see Example \ref{EX:risch}).
A similar phenomenon also occurs in the inductive proof of the main theorem in \cite{Risch1969}.
In fact, the theorem has two assertions. The first is of interest in elementary integration, 
and the second  one serves to facilitate  the proof.    
Within the second assertion, the (parametric) Risch differential equation
\begin{equation} \label{EQ:risch}
   y^\prime + h y = \sum_{i=1}^m c_i g_i
\end{equation}
appears in the literature for the first time. We write $h$ here instead of the symbol $f$ used in \cite{Risch1969}, because 
$f$ often stands for an arbitrary element of a differential field in the sequel. 
In \eqref{EQ:risch}, $h$ and the $g_i$'s belong to a differential field $F$, and the $c_i$'s are constants to be determined.

To develop complete reductions for symbolic integration, we define Risch operators using the homogeneous part of \eqref{EQ:risch}.
\begin{define} \label{DEF:risch}
Let $(F, \, ^\prime)$ be a differential field and $h \in F$. The map
\[ \begin{array}{cccc}
\cR_h: & F & \rightarrow & F \\
       & y & \mapsto     & y^\prime + h y
\end{array} \]
is called the {\em Risch operator associated to $h$}.
\end{define}
Risch operators are linear over the subfield of constants in $F$. In particular, $\cR_0$ is the derivation operator on $F$.

Assume that there is a complete reduction $\Phi_h$  for $(F, \, \cR_h)$, where $h$ 
is given on the left-hand side of \eqref{EQ:risch}. Then
each of the $g_i$'s on the right-hand side of \eqref{EQ:risch} has an R-pair $(p_i, \Phi_h(g_i))$ with respect to $\Phi_h$.  
Therefore, \eqref{EQ:risch} can be rewritten as
$\cR_h (y) = \sum_{i=1}^m c_i \left(\cR_h(p_i)+ \Phi_h(g_i)\right)$, which is equivalent to 
\[ \cR_h\left( y - \sum_{i=1}^m c_i p_i \right) = \sum_{i=1}^m c_i \Phi_h(g_i). \]
Since $\im(\cR_h)\cap\im(\Phi_h)=\{0\}$,
there exists an element $y \in F$ such that \eqref{EQ:risch} holds for some 
constants $c_1, \ldots, c_m$ 
if and only if $\sum_{i=1}^m c_i \Phi_h(g_i) = 0$ and $\cR_h\left( y - \sum_{i=1}^m c_i p_i \right) = 0$. The first constraint leads 
to a linear system over the subfield  of constants in $F$, and the second amounts to determining the kernel of~$\cR_h$.
In other words, the complete reduction finds all solutions of \eqref{EQ:risch} in $F$
by solving a linear algebraic system and $y^\prime + h y =0$ in $F$.

Let $F_n$ 
be a transcendental Liouvillian extension of $C$,
and $h$ be any element of~$F_n$.
The main result of this paper is Theorem \ref{TH:main}, in which an algorithm is presented 
for constructing a complete reduction for $(F_n, \cR_h)$ for every $h \in F_n$.
In doing so, we
generalize
the complete reduction in \cite{DGLL2025}  published in the Proceedings of ISSAC 2025 from primitive towers to general transcendental Liouvillian extensions.
The generalization also includes the complete reduction in Example \ref{EX:hyperexp} and that on exponential towers in \cite{GaoThesis}
as special cases.  In particular, setting $h:=0$ yields a complete reduction $\Phi_0$ for derivatives in $F_n$.

The new complete reduction enables us to straightforwardly determine in-field integrability of functions 
that can be expressed as elements of a transcendental Liouvillian extension.


\begin{example} \label{EX:exp}
Let
	\[
	f=\dfrac{x}{1+\exp(x)} \cdot \exp\left(\dfrac{x}{1+\exp(x)} \right).
	\]
It can be regarded as an element in the transcendental Liouvillian  extension $\bC(x,t,y)$ of $\bC$, where
$t=\exp(x)$ and $y=\exp\left(\dfrac{x}{1+\exp(x)} \right)$.
Indeed,  $f=\dfrac{xy}{1+t} \in \bC(x,t,y)$.
We determine whether $f$ has an integral in $\bC(x,t,y)$.

 The complete reduction given by the algorithm outlined in Section \ref{SECT:main}
 finds that 
 $$f = - \left(y+ \frac{y}{t} \right)'$$ and thus
\[ \int f \, dx=-\bigl(1+\exp(-x)\bigr)\cdot \exp\left(\dfrac{x}{1+\exp(x)}\right). \]
The {\sc Axiom}-based computer algebra system {\sc FriCAS 1.3.13} returns the same integral  (see \cite{FriCAS}).
But the {\tt int()}  command (with  option {\tt method=\_RETURNVERBOSE}) in {\sc maple} 2026
leaves the integral unevaluated, and so does the {\tt Integrate[]} command in {\sc mathematica} 14.3.

Computing this integral amounts to carefully handling a special case in Risch's algorithm,
which corresponds to Proposition
\ref{PROP:hbasis} (ii). Details are given in Example \ref{EX:exptower}.
\end{example}

In \cite[Chapters 3 and 4]{RaabThesis},
Raab provides several building blocks that were either missing or incomplete
within the description of Risch's algorithm in \cite{BronsteinBook}, 
and
presents a method for computing elementary integrals over admissible differential fields subject to some mild restrictions.
Such fields are more general than transcendental Liouvillian extensions.
Raab's
method needs all the ingredients from Risch's algorithm for the transcendental case.
The complete reduction presented in this paper leads to 
an alternative algorithm for elementary integration over transcendental Liouvillian extensions.
It  is based on some special structure of remainders described in Proposition~\ref{PROP:decomp},
and techniques in \cite[Section 4.3.1]{RaabThesis} for dealing  with  the logarithmic derivative
recognition problem in \cite[Section 5.12]{BronsteinBook} and parametric logarithmic derivative
problem in \cite[Section 7.3]{BronsteinBook}.

Besides elementary integration, our complete reduction may also be used for 
reduction-based creative telescoping in transcendental Liouvillian extensions,  as illustrated below. 
\begin{example} \label{EX:tele0}
Let 
$$ R(k) = \int_{0}^{\frac{\pi}{2}}\cos(2kx)\log(\sin(x))~dx \quad \text{and} \quad  I(k) = \int_{0}^{\frac{\pi}{2}}\sin(2kx)\log(\sin(x))~dx.$$
By the complete reduction,
we find that $(k+1)\cS_k - k$ is a minimal telescoper for both $R(k)$ and~$I(k)$, where $\cS_k$ stands for the shift operator
$k \mapsto k+1$. Then the telescoper yields two linear recurrences:
\[ R(k+1) - \frac{k}{k+1} R(k) = 0
\quad \text{and} \quad
I(k+1) -  \frac{k}{k+1} I(k) = \displaystyle  \frac{(-1)^k-2k-1}{4k(k+1)^2},
\]
which help us express $R(k)$ and $I(k)$ in closed forms.
Details are given in Example \ref{EX:tele}.
\end{example}

Preliminary experiments in {\sc maple} 2026 illustrated that our complete reduction outperformed the {\tt int()}  command 
in {\sc maple}  2026 for computing in-field integrals of derivatives of elements in certain transcendental elementary 
extensions. In particular, a significant speed-up was observed when integrands were derivatives
of polynomials in the generators of such an extension (see Tables \ref{tab:dLEE} and \ref{tab:dLEE2} in Section \ref{SECT:expr}).

Let $F_n = C(t_1, \ldots, t_{n-1}, t_n)$ be a transcendental Liouvillian extension of $C$.
Our construction of a complete reduction proceeds by induction on $n$: the base case $n=0$ corresponds to the subfield of constants, 
and for the inductive step we set $F := C(t_1,\dots,t_{n-1})$ and $t = t_n.$
Then $t$ is either primitive or hyperexponential over $F$.

Assume that there is an algorithm for constructing complete reductions for $(F, \cR_\alpha)$ for all $\alpha \in F$. 
For an element $h \in F(t),$ the overall strategy for constructing a complete reduction for $(F(t), \, \cR_h)$ is outlined in Figure \ref{FIG:cr},
in which $t$-normalized elements, 
the companion operator $\cP_h$, and echelon sequences are defined in Definitions \ref{DEF:normalized}, \ref{DEF:comp} and
\ref{DEF:echelon}, respectively.
The differential subring $F\langle t \rangle$ of $F(t)$ is defined in \eqref{EQ:reduced}.
\usetikzlibrary{shapes.geometric,positioning,arrows.meta}
\tikzset{
process/.style = {rectangle, draw, minimum width = 4cm, text centered, rounded corners, minimum height = 1.2cm},
mprocess/.style = {rectangle, draw, text width = 4.2cm, text centered, rounded corners, minimum height = 1.2cm},
arrow/.style  = {->,>=stealth}
}
\begin{figure}[htbp] 
\centering
\begin{tikzpicture}[thick,node distance=2cm]
\node [process] (start) {Complete reduction for $(F(t),\cR_h)$ with $h\in F(t)$};
\node [process, below of=start, yshift = -.5cm] (norm) 
{Complete reduction for $(F(t),\cR_h)$ with $t$-normalized $h\in F(t)$};
\node [process, below of=norm, yshift = -.5cm] (cr) 
{Complete reduction for $(F\langle t\rangle,\cP_h)$ with $t$-normalized $h\in F(t)$};
\node [mprocess, below left of=cr, yshift = -1cm,xshift=-1.7cm] (prim) 
{Primitive case\quad $F\langle t\rangle = F[t]$};
\node [mprocess, below right of=cr, yshift = -1cm,xshift=1.7cm] (hyper) 
{Hyperexponential case $F\langle t\rangle = F[t,t^{-1}]$};
\node [process, below of=cr, yshift = -2.9cm] (auxi) 
{Auxiliary subspace $A_h$ s.t.\ $F\langle t\rangle = \im(\cP_h) + A_h$ with $t$-normalized $h\in F(t)$};
\node [process, below of=auxi, yshift = -.5cm] (comp) 
{Complementary subspace $W$ s.t.\ $F\langle t\rangle = \im(\cP_h) \oplus W$ with $t$-normalized $h\in F(t)$};
\draw [arrow] (start) -- (norm) node[above right = 1cm and .1cm] {Algorithm GKS in \cite{CDGL2025}};
\draw [arrow] (norm) -- (cr) node[above right = 1cm and .1cm] {Algorithm GKSR in \cite{CDGL2025}};
\draw [arrow] (cr) -- (0,-6.2) -| (prim); 
\draw [arrow] (cr) -- (0,-6.2) -| (hyper); 
\draw [arrow] (prim) |- (0,-8.6) -- (auxi);
\draw [arrow] (hyper) |- (0,-8.6) node [below right = .1cm and .1 cm] {Complete reductions on $F$} -- (auxi);
\draw [arrow] (auxi) -- (comp) node[above right = 1cm and .1cm] {Echelon sequence of $\im(\cP_h)\cap A_h$};
\end{tikzpicture}
\caption{General outline of the complete reduction construction} \label{FIG:cr}
\end{figure}

We transform the problem of constructing a general complete reduction
to that of constructing a complete reduction for $(F\langle t \rangle, \, \cP_h)$ with $h$ being $t$-normalized. 
The idea for such a transformation 
can be traced back to \cite[Condition b) on page 905]{Davenport1986} and \cite[Section 6.1]{BronsteinBook},
in which the problem of solving Risch differential equations in $F(t)$ is reduced to the corresponding problem in the differential subring~$F\langle t \rangle$.
Similar transformations are also given in \cite{GLL2004,CDGL2025}, which can be viewed, respectively, as a differential analogue and a
generalization of the minimal decomposition for hypergeometric terms in \cite{AP2001, AP2002}.
This step is based on Algorithms {\sc GKS} and  {\sc GKSR} in \cite{CDGL2025}.
Algorithm {\sc GKS} decomposes an element of $F(t)$  as the sum of a $t$-normalized element and a logarithmic derivative,
and Algorithm {\sc GKSR} decomposes an element of $F(t)$ as the sum of an element of $\im(\cR_h)$, an element of $F\langle t \rangle$ and a 
$t$-simple element (see Definition \ref{DEF:tsimple}) under the assumption that $h\in F(t)$ is $t$-normalized.

Despite the substantial differences between primitive and hyperexponential generators, we develop a unified
approach to constructing a complete reduction for $(F\langle t \rangle, \, \cP_h)$. Our approach constructs an auxiliary subspace $A_h$ such that $F\langle t \rangle = \im(\cP_h)+A_h$,
and computes  an echelon sequence of~$\im(\cP_h) \cap A_h$.
This sequence then enables us to find a complement of $\im(\cP_h)$ in a dual manner.
While related ideas have appeared in 
\cite[Section 4.1]{BCCLX2013}, \cite[Chapter 4]{GaoThesis} and \cite[Section 3]{DGLL2025} in one way or another,
a comprehensive generalization is needed in our setting, because $F(t)$ is  more involved than
the differential fields considered in those works. 
This step makes essential use of  two algorithms in \cite[Section 4.3.1]{RaabThesis}. One is
for recognizing logarithmic derivatives in $F$, and the other is for solving 
the parametric logarithmic derivative problem in $F$.

The remainder of this paper is organized as follows. Section~\ref{SECT:pre} introduces basic terminology from symbolic integration and provides
a dual description of complementary subspaces. 
In Section \ref{SECT:risch}, we recall the notion of $t$-normalized elements, define companion operators,
and transform the problem of constructing complete reductions on a transcendental Liouvillian extension
to that on a differential subring.
Our induction hypothesis is given in Section \ref{SECT:ind}.
The most technical inductive steps are treated for primitive generators in Section \ref{SECT:prim}, 
and for hyperexponential generators in Section \ref{SECT:hyperexp}.
The induction is finalized in Section \ref{SECT:main}. 
Applications and computational benchmarks are presented in Sections \ref{SECT:app} and \ref{SECT:expr}, respectively.
Section \ref{SECT:conc} contains concluding remarks.
In the appendix, we describe algorithms derived from the constructive proofs given in Sections \ref{SECT:prim} and~\ref{SECT:hyperexp}.

\section{Preliminaries} \label{SECT:pre}

This section consists of three parts. First, we introduce notation to be used in the sequel, along with
fundamental terminology for symbolic integration. Next, we review several useful properties of residues 
relevant to elementary integration. Finally, we characterize certain complementary subspaces via
kernels of linear functions. This characterization will enable us to construct complementary subspaces 
of infinite dimension in Sections \ref{SECT:prim} and \ref{SECT:hyperexp}. 

\subsection{Basic notation and terminology} \label{SUBSECT:term}

Besides the usual notation in textbooks, we let $\bN$, $\bN_0$ and $\bN^-$ be the sets of positive, nonnegative and negative integers, respectively.
Set
$[n] := \{1, \ldots, n\}$ for $n \in \bN$, and $[n]_0 := \{0, 1, \ldots, n\}$ for~$n \in \bN_0$.
For an abelian group $(G, +, 0)$, 
we write $G^\times$ for $G \setminus \{0\}$.
Let $R$ be a commutative ring, $S \subset R$ and $x \in R$.  The set $\{s x \mid s \in S\}$ is denoted by  $S \cdot x$ or $x \cdot S$.

All fields appearing in the sequel are of characteristic zero.
Let $F$ be a field and $p \in F[t]^\times$. The degree and leading coefficient of $p$ are denoted
by $\deg_t(p)$ and $\lc_t(p)$, respectively.  In addition, $\deg_t(0) := -\infty$ and $\lc_t(0) := 0$.
We set $F[t]_{<d} := \{ p \in F[t] \mid \deg_t(p)<d\}$ for all $d \in \bN_0$.

The ring of Laurent polynomials in $t$ over $F$ is denoted by $F[t,t^{-1}]$.
Let $q \in F[t, t^{-1}]$ be of the form
$q_k t^k + q_{k-1} t^{k-1} + \cdots + q_{l+1} t^{l+1} + q_l t^l,$
where $k, l \in \bZ$, $k \ge l$, $q_k, q_{k-1}, \ldots,  q_{l+1}, q_l \in F$, and $q_k q_l \neq 0$.
The {\em head degree, head coefficient, tail degree and tail coefficient} of $q$ are
defined to be $k$, $q_k$, $l$ and $q_l$,
denoted by  $\hdeg_t(q)$, $\hc_t(q)$, $\tdeg_t(q)$ and $\tc_t(q)$, respectively.
In addition, $\hdeg_t(0):= -\infty$, $\hc_t(0)=0$, $\tdeg_t(0) := \infty$, and $\tc_t(0):=0$.
Since 
$$F[t,t^{-1}]=F[t] \oplus \left(t^{-1} \cdot F[t^{-1}]\right),$$
we  set $q^+$ and $q^-$ to be  the respective projections of $q$ to $F[t]$ and $t^{-1} \cdot F[t^{-1}]$.
For $Q \subset F[t, t^{-1}]$ and $l \in \bZ$, 
we denote by $Q_{>l}$ 
the set of elements in $Q$ whose tail degrees are greater than $l$.

For an element $f$ of $F(t)$, its numerator and denominator are denoted by $\num_t(f)$ and $\den_t(f)$, respectively.
Moreover, $\den_t(f)$ is set to be monic. We say that $f$ is {\em $t$-proper} if the degree of $\num_t(f)$
is less than that of $\den_t(f)$. For later convenience, the degree of $f$ in $t$ is defined to be the maximum of
the degrees of its numerator and denominator, and is denoted by $\Deg_t(f)$.
Note that $\deg_t(0)=-\infty$ but $\Deg_t(0)=0$.

The subscript $t$ will be omitted for brevity when the indeterminate $t$ is clear from context.

Let $f \in F(t)^\times$, $a=\num(f)$ and $b = \den(f)$. For $p \in F[t]$ with $\deg(p)>0$,
the {\em order} of~$f$ at $p$, denoted by $\nu_p(f)$,  is defined as follows:
if $p$ occurs in $b$ with multiplicity $m>0$ then $\nu_p(f) = -m$; 
otherwise, $\nu_p(f)$ is equal to the multiplicity of $p$ in  $a$. By convention, we set 
$\nu_p(0) = \infty$ for all $p \in F[t]$ with $\deg(p)>0.$
The {\em order of $f$ at infinity}, denoted by $\nu_\infty(f),$ is defined as $\deg(b)- \deg(a)$.

Let $(F, \, ^\prime)$ be a differential field. 
We denote  $\{f^\prime \mid f \in F\}$ by $F^\prime$, which is a linear space over the subfield of constants
in $F$.
An element of $F$ is called a {\em logarithmic derivative}
if it is equal to  $g^\prime/g$ for some $g \in F^\times$.
A differential field $(E, \, \delta)$ is called a {\em differential field extension} of $(F, \, ^\prime)$ if $F$ is a subfield of $E$ and $\delta|_F=\, ^\prime$.
The derivation $\delta$ will still be denoted by $^\prime$ when there is no confusion.

Let $(E, \, ^\prime)$ be a differential field extension of $(F, \, ^\prime)$ and $u \in E$.
Then $u$ is said to be {\em primitive} (resp.\ {\em hyperexponential}) over $F$ if $u^\prime \in F$ (resp.\  $u \neq 0$ and $u^\prime/u \in F$).
An element $u$ of $E$ is said to be {\em logarithmic} (resp.\ {\em exponential}) over $F$ if  $u^\prime$ is a logarithmic derivative in $F$ (resp.\  $u \neq 0$ and $u^\prime/u \in F^\prime$). If $u$ is logarithmic (resp.\ exponential) over $F$,
then it is primitive (resp.\ hyperexponential) over $F$.  The converse is false.
For example, the logarithmic integral $\Li(x)$,  which equals $\int \frac{1}{\log x}$, is primitive but not logarithmic over $\bC(x, \log(x))$, and  $x \exp(x)$ is hyperexponential 
but not exponential over $\bC(x)$.

An element $t$ in a differential field extension of $(F, \, ^\prime)$ is called a {\em monomial over $F$} if $t$ is transcendental over $F$ and $t^\prime \in F[t]$.
A monomial $t$ over $F$ is said to be {\em regular} if $F(t)$ and~$F$ have the same subfield of constants. By Definition \ref{DEF:liouvillian},
a transcendental Liouvillian
extension $C(t_1, \ldots, t_n)$ is a differential field extension of $C$, in which $t_i$ is a regular monomial over $C(t_1, \ldots, t_{i-1})$,
and is either primitive or hyperexponential
over the same field for all $i \in [n]$.

In the rest of this section, we let $(F, \, ^\prime)$ be a differential field, $C$ be its subfield of constants,
and $t$ be a monomial over $F$.
Then $F[t]$ is a differential subring of $F(t)$.
A nonzero polynomial  $p \in F[t]$ is said to be {\em normal} (resp.\ {\em special}) if $\gcd(p, p^\prime)=1$
(resp.\ $\gcd(p, p^\prime)=p$) (see \cite[Definition 3.4.2]{BronsteinBook}).
The set of normal polynomials with positive degrees is denoted by $N_t$.

 Let us recall the notion of $t$-simple elements in \cite[Section 2.1]{CDGL2025}.
\begin{define} \label{DEF:tsimple}
An element of $F(t)$ is
{\em $t$-simple} if it is $t$-proper and has a normal denominator.
\end{define}
All $t$-simple elements in $F(t)$ form an $F$-subspace, which is denoted by $S_t$.
Let
\begin{equation} \label{EQ:reduced}
F\langle  t \rangle = \{ f \in F(t) \mid \text{$\den(f)$ is special} \}.
\end{equation}
Then $F\langle  t \rangle$ is a differential subring (see \cite[Corollary 4.4.1]{BronsteinBook}).
We remark that $t$-simple elements are not required to be $t$-proper in \cite[Definition 3.5.2]{BronsteinBook}. 
Our further requirement for $t$-properness
results in $S_t \cap F \langle t \rangle = \{0\}$, which is necessary for the direct sum  \eqref{EQ:pre1} in Section \ref{SECT:risch}.

In what follows, we use lowercase Greek letters such as $\alpha, \beta, \gamma, \lambda, \mu, \sigma$ and $\tau$ to denote elements in $F$ or those in 
the algebraic
closure of $F$ except coefficients given by polynomials in $F[t]$ and elements of $C$. 
On some occasions, lowercase Greek letters in Alegreya-LF style such as $\bfxi, \bfeta$ and $\bfomega$ indicate elements of $F(t)$, which may possibly lie outside $F$.

Some basic facts about generators in a transcendental Liouvillian extension are collected in the following lemma for  later references. 
\begin{lemma} \label{LM:fact}
Let $t$ be regular over $F$, $d \in \bZ$, $\alpha \in F^\times$ and $p \in F[t]$.
\begin{itemize}
\item[(i)] If $d>0$ and $t$ is primitive over $F$, then $\left( \alpha t^d \right)^\prime = \alpha^\prime t^d + \left(\alpha d t^\prime \right) t^{d-1}$
whose degree is equal to~$d$ if $\alpha$ is not a constant, and $d-1$, otherwise.
\item[(ii)] If $d \neq 0$ and $t$ is hyperexponential over $F$, then $\left(\alpha t^d \right)^\prime = \left(\alpha^\prime + \alpha d (t^\prime/t)\right) t^d$
whose head and tail degrees are both equal to $d$. 
\item[(iii)] If $t$ is  primitive over $F$, then $p$ is normal if and only if $p$ is squarefree.
\item[(iv)] If $t$ is hyperexponential over $F$, then $p$ is normal if and only if $p$ is squarefree and $t \nmid p$.
\end{itemize}
\end{lemma}
\begin{proof} 
(i) Let $m$ be the degree of $\left( \alpha t^d \right)^\prime$. By Leibniz's rule, $\left( \alpha t^d \right)^\prime 
= \alpha^\prime t^d + \left( \alpha d t^\prime \right) t^{d-1}$. 
So $m=d$ if $\alpha^\prime \neq 0$. 
Otherwise, $m=d-1$ since $t$ is both regular and primitive.

(ii) Similarly, Leibniz's rule implies that $\left( \alpha t^d \right)^\prime = \left(\alpha^\prime + \alpha d (t^\prime/t)\right) t^d$,
which is nonzero by \cite[Theorem 5.1.2]{BronsteinBook}. Since $t^\prime/t \in F$, both head and tail degrees of $\left( \alpha t^d \right)^\prime$ are 
equal to $d$.

(iii) and (iv) hold by \cite[Lemma 3.4.4, Theorems 5.1.1 and  5.1.2]{BronsteinBook}.
\end{proof}

We use sequences in algorithmic descriptions in order to avoid clumsy names of list operations.
Let $L: l_1, \ldots, l_k$ be a sequence. Then $l_i$ is denoted by $L[i]$ for all $i \in [k]$.  
The {\em length} of $L$ is defined to be $k$ and denoted by $\len(L)$.
In the description of an algorithm, comments are placed between $(^* {\sl \cdots} ^*)$.
When outlining algorithms, we abbreviate \lq\lq with respect to\rq\rq\ and \lq\lq such that\rq\rq\ 
as \lq\lq w.r.t.\rq\rq\ and \lq\lq s.t.\rq\rq, respectively.

\subsection{Residues and elementary integration} \label{SUBSECT:residue}

First, we recall the notion of residues given in \cite[Definition 4.4.1]{BronsteinBook}. 
\begin{define} \label{DEF:residue}
Let $p \in N_t$ and $V_p = \{f \in F(t) \mid \nu_p(f) \ge -1\}$. 
For an element $f \in V_p$, the {\em residue of $f$ at $p$}, denoted by $\residue_p(f),$ 
is $\pi_p(f p/p^\prime)$, where $\pi_p$ stands for the canonical homomorphism from 
the ring $\{ q \in V_p \mid \nu_p(q) \ge 0\}$ to $F[t]/(p)$.
\end{define}
Note that $\residue_p(f)$ is a congruence class.
We identify  $\residue_p(f)$ with the element of degree less than $\deg(p)$ in $\pi_p(f p/p^\prime)$, because we are mainly concerned with constant residues. 
\begin{remark} \label{RE:residue}
The residue of~$f\in S_t$ is well-defined at every element of $N_t$, as $\nu_p(f) \ge -1$ for all $p \in N_t$.
\end{remark}

Two well-known facts about residues are:  
\begin{itemize}
\item all residues of an element in $F[t]$ are equal to zero;
\item all residues of a logarithmic derivative in $F(t)$ are integers (see \cite[Corollary 4.4.2 (iii)]{BronsteinBook}).
\end{itemize}

The next two lemmas summarize some technical results scattered in \cite{DGGL2023},
which will be useful for proving Theorem \ref{TH:elem}. 
\begin{lemma} \label{LM:logder}
Let $t$ be regular over $F$, $p \in F[t]^\times$ with $d = \deg(p)$ and $p_d = \lc(p)$. 
\begin{itemize}
\item[(i)] If $t$ is primitive over $F$, then $p^\prime/p = s + p_d^\prime/p_d$ for some $s \in S_t$.
\item[(ii)] If $t$ is hyperexponential over $F$, then $p^\prime/p = s+ p_d^\prime/p_d+ d t^\prime/t$
for some $s \in S_t$.
\item[(iii)] All residues of $s$ appearing in (i) or (ii) are integers.
\end{itemize}
\end{lemma}
 \begin{proof}
(i)  Let $q = p_d^{-1} p$. Then $q$ is monic. By \cite[Exercise 3.1]{BronsteinBook}, we have $p^\prime/p=p_d^\prime/p_d+q^\prime/q$. 
All irreducible factors of $q$ are normal by Lemma \ref{LM:fact} (iii).
Applying Lemma \ref{LM:fact} (i) and \cite[Exercise 3.1]{BronsteinBook} again gives $q^\prime/q \in S_t$.
 The conclusion holds by setting $s := q^\prime/q$.

(ii) There exists an integer $m \in \bN_0$ and a monic polynomial $r \in F[t]$ with $t \nmid r$ such that $p = p_d t^m r$.
All irreducible factors of $r$ are normal by Lemma \ref{LM:fact} (iv).
 The logarithmic derivative identity implies that 
\begin{equation} \label{EQ:logder1}
\frac{p^\prime}{p} = \frac{p_d^\prime}{p_d} + m \frac{t^\prime}{t} + \frac{r^\prime}{r}.
\end{equation} 
Set $s := r^\prime/r - (d-m)t^\prime/t$. 
Since $t^\prime/t \in F$, it follows from \cite[Exercise 3.1]{BronsteinBook} and 
Lemma \ref{LM:fact} (ii) that $s \in S_t$. Substituting $d t^\prime/t+s$ for $ m t^\prime/t + r^\prime/r$ in \eqref{EQ:logder1}
yields that  $p^\prime/p = p_d^\prime/p_d + d t^\prime/t + s$.

(iii) The $t$-simple element $s$ in (i) or (ii) is the difference of $p^\prime/p$ and an element in $F$.
Since all residues of a logarithmic derivative are integers, so are the residues of $s$.
\end{proof}

Recall that a differential field extension $F(u_1, \ldots, u_m)$ is called an {\em elementary extension} of $F$
if each $u_j$ is either algebraic or logarithmic or exponential over $F$ and 
the extension contains no new constants other than those in $F$.
We say that an element of $F$ has an {\em elementary integral}
over $F$ if it is a derivative in some elementary extension of $F$ (see \cite[Definition 5.1.4]{BronsteinBook}).  
\begin{lemma} \label{LM:elem}
Let $t$ be regular over $F$ and $f \in S_t^\times$. Assume that $C$ is algebraically closed, and that 
$t$ is either primitive or hyperexponential over $F$. If 
all residues of $f$ belong to $C$, then $f$ has an elementary integral over $F(t)$. 
\end{lemma}
\begin{proof}
Let $\alpha_1, \ldots, \alpha_k$ be the distinct roots of $\den(f)$ in the algebraic closure $\overline{F}$ of $F$.
Since all residues of $f$ in $\overline{F}$ are constants,  it follows from \cite[Lemma 3.1 (i)]{DGGL2023} that 
$$ f =  c_1 \frac{(t-\alpha_1)^\prime}{t-\alpha_1} + \cdots + c_k \frac{(t-\alpha_k)^\prime}{t-\alpha_k} + p$$
for some $c_1, \ldots, c_k \in C$ and $p \in F[t]$.
If $t$ is primitive over $F$, then $p=0$ by $f \in S_t$ and Lemma \ref{LM:fact}~(i).
If $t$ is hyperexponential over~$F$, then $p= c t^\prime/t$ for some $c \in C$ by $f \in S_t$ and Lemma \ref{LM:fact}~(ii).
Thus, $f$ has an elementary integral over $F(t)$. 
\end{proof}

\subsection{A dual presentation for complementary subspaces} \label{SUBSECT:dual}
Let $V$ be a  $C$-linear space. Every subspace of $V$ can be written as the intersection of kernels of some linear functions on $V$.
Such an intersection is understood as a dual presentation.

For a nonempty subset $S$ of $V$,
$\spa_C(S)$ stands for the $C$-subspace spanned by $S$.
The empty set spans the zero subspace $\{0\}$ by convention.
Let $\Theta$ be a basis of $V$. For $\theta \in \Theta$, we denote by $\theta^*$ the linear
function on~$V$ that maps $\theta$ to $1$ and other elements of $\Theta$ to $0$.

\begin{define} \label{DEF:echelon}
Let $V$ be a $C$-linear space, $\Theta$ be a basis of V, $U$ be a subspace of $V$ and $\bI \subset \bN_0$. 
A basis $\{\vu_i \mid i \in \bI\}$ of $U$ is called an {\em echelon basis of $U$ with respect to $\Theta$} if,
for every $i \in \bI$,  there
exists an element $\theta_i \in \Theta$ such that $\vu_i \notin \ker(\theta_i^*)$, and 
$\vu_j \in \ker(\theta_i^*)$ for all $j \in \bI$ with $j<i$. 
In this case, 
we also call $\{ \vu_i \mid i \in \bI \}$ an {\em echelon basis of $U$ with pivots $\{ \theta_i \mid i \in \bI \}$}.
\end{define}
For example, an echelon basis $\{\vu_i \mid i \in \bN_0\}$ of $U$  with pivots $\{ \theta_i \mid i \in \bN_0\}$ is of the 
\lq\lq upper triangular\rq\rq\ form

\[
\begin{array}{l@{\ }c@{\ }l}
\ \vdots & \\ 
\vu_i &=& c_i \theta_i + \displaystyle \sum_{\theta \notin \{\theta_i, \theta_{i+1}, \theta_{i+2}, \ldots \}} c_\theta \theta \\ 
\ \vdots&&\quad\ddots\phantom{\sum_{\{\theta_i, \theta_{i+1}, \theta_{i+2}, \ldots \}} c_\theta \theta}\\[1ex]
\vu_1 &=& \phantom{c_i \theta_i +}c_1 \theta_1 + \displaystyle \sum_{\theta \notin \{\theta_1, \theta_{2}, \theta_{3}, \ldots \}} c_\theta \theta \\[.53cm]
\vu_0 &=& \phantom{c_i \theta_i +c_1\theta_1+}c_0 \theta_0 + \displaystyle \sum_{\theta \notin \{\theta_0, \theta_{1}, \theta_{2}, \ldots\}} c_\theta \theta,  
\end{array}
\]
where $i \in \bN_0$, $c_i \in C^\times$, $\theta_i$ is the pivot of $u_i$, and $c_\theta \in C$, finitely many nonzero.
\begin{lemma} \label{LM:echelon}
With the notation introduced in Definition \ref{DEF:echelon}, we further let
$\{ \vu_i \mid i \in \bI\}$ be an echelon basis of $U$ with pivots $\{ \theta_i \mid i \in \bI\}$.
Then $V = U \oplus \left( \cap_{i \in \bI} \ker(\theta_i^*) \right).$ 
\end{lemma}
\begin{proof} Let $\widetilde{\Theta} = \Theta \setminus \{ \theta_i \mid i \in \bI\}$. 
Then $\cap_{i \in \bI} \ker(\theta_i^*) = \spa_C \widetilde{\Theta}$. The lemma follows from the
observation that $\{ \vu_i \mid i \in \bI\} \cup \widetilde{\Theta}$ is a basis of $V$.
\end{proof}

When developing complete reductions, we face situations, in which a subspace $U$ of $V$ is the image of some linear operator,
and it is difficult to   directly construct its echelon basis. To construct a complement of $U$, 
we introduce an auxiliary subspace $A$ such that $V=U + A$ and~$U \cap A$ has an echelon basis.
Then a complement of $U$ can be obtained from the next lemma. 
\begin{lemma} \label{LM:modulo}
With the notation introduced in Definition \ref{DEF:echelon}, we further let $A$ be a subspace of~$V$ such that $V=U+ A$.
If $U \cap A$ has an echelon basis with pivots  $\{\theta_i \mid i \in \bI\}$, then 
the intersection $$A \cap \big( \mathop{\cap}\limits_{i \in \bI} \ker(\theta_i^*) \big)$$
is a complement of $U$ in $V$. 
\end{lemma}
\begin{proof} 
Let $W=\cap_{i \in \bI} \ker(\theta_i^*)$. Then $V= U \cap A+W$ by Lemma \ref{LM:echelon}.  
Since $U \cap A \subset A$, the modular law of subspaces gives $A = U \cap A + W \cap A$. Thus,
$$V = U + A = U  + U  \cap A + W \cap A = U + W \cap A.$$
The conclusion follows immediately from $U \cap A \cap W = \{0\}$.
\end{proof} 

In this paper,
reductions are typically carried out inside, or modulo, the image of a linear operator.
This motivates us to define
the notion of echelon sequences so as to find pre-images and images simultaneously during reduction.
\begin{define} \label{DEF:eseq}
Let $V$ be a $C$-linear space, $\cL$ be a linear operator on $V$, $\Theta$ be a basis of $V$, and~$U$ be a subspace contained in $\im(\cL)$.
Assume that $\{\cL(\vv_i) \mid i \in \bI\}$ is an echelon basis of $U$ with pivots $\{\theta_i \mid i \in \bI \}$, 
Then we call
$\{ (\vv_i, \cL(\vv_i), \theta_i) \mid i \in \bI \}$ 
an {\em echelon sequence of $U$ with respect to $(V, \cL)$ and $\Theta$}.
\end{define}
\begin{example} \label{EX:dual1}
Let $V=C[x]$ with a $C$-basis $\Theta = \{x^i \mid i \in \bN_0\},$ and $\cL$ be the linear operator that maps $p \in V$ to $(x^2+1)p$.
Then $\im(\cL)$ has an echelon sequence 
$\left\{ \left(x^i, x^{i+2}+x^i,  x^{i+2} \right) \right\}_{i \in \bN_0}.$
By Lemma \ref{LM:echelon}, a complement of $\im(\cL)$ in $V$ is  $\cap_{i \in \bN_0} \ker(\theta_i^*)$ with $\theta_i=x^{i+2}$. This complement is equal to $C[x]_{<2}$.
\end{example} 

In the final part of this section, we make a preparation for computing R-pairs (see  Definition \ref{DEF:Rpair}). 
First, we recall the definition of effective bases in \cite[Section 2.2]{DGLL2025}.
\begin{define} \label{DEF:eff}
Let $V$ be a $C$-linear space. A basis $\Theta$ of $V$ is said to be {\em effective} if
two algorithms are available: one takes every  $\vv \in V^\times$ and computes an element 
$\theta \in \Theta$ such that $\theta^*(\vv) \neq 0$, and the other evaluates $\theta^*(\vv)$ for all 
$\vv \in V$ and $\theta \in \Theta$.
\end{define}
The transcendental Liouvillian extension $C(t_1, \ldots, t_n)$ defined in Definition \ref{DEF:liouvillian}
admits an effective basis via successive irreducible partial fraction decompositions; details can be
found in Algorithms 2.4, 2.6 and Remark 2.7 in \cite[Section 2.2]{DGLL2025}.
\begin{lemma} \label{LM:elim}
With the notation introduced in Definition \ref{DEF:eseq}, we further assume that $\Theta$ is effective. 
If $\{ (\vv_i, \cL(\vv_i), \theta_i) \mid i \in \bI\}$ is an echelon sequence of $U$ with respect to $(V, \cL)$ and $\Theta$,
then, for every $\vx \in V$,  one can compute an element $\vv \in V$ and its image $\cL(\vv) \in U$ (as a byproduct) such that
$\vx-\cL(\vv) \in \cap_{i \in \bI} \ker \left(\theta_i^*\right).$
\end{lemma}
\begin{proof}
Since $\Theta$ is effective, one can express $\vx$ as a linear combination of elements in $\Theta$. Let $\Theta_{\vx}$ be the finite set of all pivots appearing in this linear combination. If $\Theta_{\vx} = \emptyset$, then letting $v = 0$ and thus $\cL(v) = 0$ concludes the proof.

Assume that $\Theta_{\vx} \neq \emptyset$ and consider the maximal index  $j$ with $\theta_j \in \Theta_{\vx}$. Set $c_{\vx} :=\theta_j^*(\vx)$ and $c_j := \theta_j^*(\cL(\vv_j))$, which is nonzero. Define $\vy_j :=  c_{\vx} c_j^{-1} \vv_j$. Then $\cL(\vy_j)=c_{\vx} c_j^{-1} \cL(\vv_j)\in U$, and, by the maximality of $j$, we have $\vx - \cL(\vy_j) \in \ker(\theta_i^*)$ for all $i \ge j$.

Now replace $\vx$ by $\vx - \cL(\vy_j)$ and repeat the argument. 
Then the pivot set of the new $\vx$ does not contain any element $\theta_k$ with $k \ge j$.
So after finitely many iterations we obtain a finite sum $v = \sum_l y_l$ and the corresponding $\cL(v)$ with the required property.
\end{proof}

\section{Normalization and companion operators} \label{SECT:risch}

Throughout this section, 
we let $(F, \, ^\prime)$ be a differential field,
$C$ be the subfield of constants
in~$F$, and $t$ be a monomial over $F$.

We initially sought to construct a complete reduction for $(F(t), \,^\prime)$ under the assumption that there was a complete reduction
for $(F, \, ^\prime)$.  This assumption suffices when $t$ is regular and primitive over $F$. 
However, it becomes insufficient when $t$ is hyperexponential over $F$, as illustrated by the following example.
\begin{example} \label{EX:risch}
Let $t$ be both regular and hyperexponential. Set $\alpha = t^\prime/t$, which belongs to $F$.
For an element $\beta \in F$, we determine whether $\beta t \in F(t)^\prime$.
By an order argument, $\beta t \in F(t)^\prime$ if and only if $\beta t \in F[t, t^{-1}]^\prime$,
which, by Lemma \ref{LM:fact} (ii), is equivalent to $\beta t = (\gamma t)'$ for some $\gamma \in F.$
In other words, $\gamma$ is a solution of the Risch differential equation
$y^\prime + \alpha y = \beta$, or equivalently,
$\beta \in \im(\cR_\alpha)$, where $\cR_\alpha$ is given in Definition \ref{DEF:risch}.
Therefore, complete reductions for $(F(t), \, ^\prime)$ are inherently connected to Risch differential equations.
\end{example}

We are going to construct a complete reduction $\Phi_h$ for $(F(t), \, \cR_h)$ for all $h \in F(t)$
under the assumption that a complete reduction for $(F, \, \cR_\alpha)$ is available  for every element $\alpha \in F$,
where~$\cR_\alpha$ stands for the Risch operator on $F$ associated to $\alpha$.

\begin{example} \label{EX:const}
Let $F=C$ and $c \in C$. The Risch operator $\cR_c:C \rightarrow C$ is given by $y \mapsto y^\prime + c y$
for all $y \in C$. Then $\cR_c(y)=cy$. It follows that  $\im(\cR_c) = C$ if $c \neq 0$, and $\im(\cR_c)=\{0\}$ if $c=0$.
The complete reduction for $(C, \, \cR_c)$ is the zero map if $c \neq 0$, and is the identity map if $c=0$.
\end{example}

Similar to the way of solving Risch equations in \cite[Theorem 6.1.1]{BronsteinBook}, 
it suffices to focus on Risch operators associated to normalized elements defined below (see also \cite[Section 3.2]{CDGL2025}).
\begin{define} \label{DEF:normalized}
An element $h$ of $F(t)$ is said to be {\em $t$-normalized} if, for all $i \in \bZ$,
$$\gcd \left( \num(h)-i \den(h)^\prime, \den(h) \right)=1$$
\end{define}
The gcd-condition is equivalent to the requirement that no nonzero residue of $h$ is 
an integer by \cite[Theorem 4.4.3]{BronsteinBook}.
All elements of $F$ are $t$-normalized.

This normalization will simplify our construction in three aspects:
\begin{itemize}
\item[(i)] Reduce the problem of constructing complete reductions for $(F(t), \, \cR_h)$
 with $h \in F(t)$ to that for $(F(t), \, \cR_h)$  with $t$-normalized $h\in F(t)$ (see Proposition \ref{PROP:ks}). 
\item [(ii)] Further transform the problem in $F(t)$  into an equivalent one on  $F \langle t \rangle$ (see Proposition \ref{PROP:red}).
\item[(iii)] The Risch operator associated to a $t$-normalized element $h$ is injective if $h \in F(t) \setminus F$, $t$ is regular, 
and $t$ is either primitive or hyperexponential over $F$
(see Lemma \ref{LM:inj}).
\end{itemize}

An element of $F(t)$ is said to be {\em weakly $t$-normalized} if it does not have any positive integer residue
(see \cite{Davenport1986} and \cite[Definition 6.1.1]{BronsteinBook}). 
Note that (i) and (ii) remain valid for weakly $t$-normalized elements, while (iii) relies on
the assumption that $h$ has no residue in $\bZ^\times$.

For every element $h \in F(t)$, there exists a pair $(\bfxi, \bfeta) \in F(t) \times F(t)^\times$ such that
$\bfxi$ is $t$-normalized, and $h = \bfxi + \bfeta^\prime/\bfeta$. We call $(\bfxi, \bfeta)$ a {\em normal form} of $h$. 
Normal forms are not unique. 
Among all such normal forms, the canonical one can be computed by
Algorithm {\sc GKS} in \cite[Section 3.2]{CDGL2025}.
\begin{prop} \label{PROP:ks}
Let $h \in F(t)$ with a normal form $(\bfxi, \bfeta)$.
Then the following assertions hold.
\begin{itemize}
\item[(i)] $\cR_h= \bfeta^{-1} \circ \cR_\bfxi \circ \bfeta$, where $\bfeta$ is understood as the map with $f \mapsto \bfeta f$ for all $f \in F(t)$.
\item[(ii)] Assume further that $\Phi_\bfxi$ is a complete reduction for $(F(t), \, \cR_\bfxi)$.
Then $\Phi_h :=\bfeta^{-1} \circ \Phi_\bfxi  \circ \bfeta$ is a complete reduction for $(F(t), \, \cR_h)$.
\item[(iii)]
For all $f \in F(t)$,
$\left(\bfeta^{-1} g, \,  \bfeta^{-1} r \right)$ is an R-pair of $f$ with respect to
 $\Phi_h$ if $\left(g, \,  r \right)$ is an R-pair of $\bfeta f$ with respect to $\Phi_\bfxi$.
\end{itemize}
\end{prop}
\begin{proof}
(i) For all $f \in F(t)$, $\bfeta^{-1} \circ \cR_\bfxi \circ \bfeta (f) = \bfeta^{-1} \left( (\bfeta f)^\prime  + \bfxi (\bfeta f) \right)$.
Then $\bfeta^{-1} \circ \cR_\bfxi \circ \bfeta (f) =  \cR_h(f)$ by $h=\bfxi+\bfeta^\prime/\bfeta$ and a straightforward calculation.

(ii) Since $\Phi_\bfxi$ is an idempotent, so is $\Phi_h$. Hence, $\im(\Phi_h)$ is a complement of its kernel in $F(t)$.
It suffices  to verify that $\ker(\Phi_h)=\im(\cR_h)$.  By (i), 
$\im(\cR_h) = \bfeta^{-1} \cdot \im(\cR_\bfxi)$.
By $\Phi_h =\bfeta^{-1} \circ \Phi_\bfxi  \circ \bfeta$, we have $\ker(\Phi_h) = \bfeta^{-1} \cdot \ker(\Phi_{\bfxi}).$
Thus, $\ker(\Phi_h)=\im(\cR_h)$
by $\ker(\Phi_\bfxi)=\im(\cR_\bfxi)$.

(iii) Let $(g,r)$ be an R-pair of $\bfeta f$ with respect to $\Phi_\bfxi$.
Then $\bfeta f = \cR_\bfxi(g) + r$ and $r \in \im(\Phi_\bfxi)$. It follows that
$f =  \cR_h (\bfeta^{-1} g) + \bfeta^{-1} r,$
which, together with $\bfeta^{-1} r \in \im(\Phi_h)$, implies that $\left(\bfeta^{-1} g, \,  \bfeta^{-1} r \right)$ is an R-pair of $f$
with respect to  $\Phi_h$.
\end{proof}
In the rest of this section, we let $h \in F(t)$ be $t$-normalized, $a=\num(h)$ and $b = \den(h)$. Set
\begin{equation} \label{EQ:subspaceh}
S_{t,h} := \{ s \in S_t \mid \gcd(\den(s),b)=1 \},
\end{equation}
which is also a $C$-subspace in $F(t)$.
For all $f \in F(t)$, Algorithm {\sc GKSR} in  \cite[Section 3.4]{CDGL2025} computes a triple $(g,r,s) \in F(t) \times F\langle  t \rangle \times S_{t,h}$
such that
\begin{equation} \label{EQ:pre2}
f =  \cR_h(g) + \frac{r}{b} + s.
\end{equation}
Moreover, $s$ in the above equality is unique. Therefore, 
\begin{equation} \label{EQ:pre1}
F(t) = \left(\im(\cR_h) + b^{-1} \cdot F\langle  t \rangle  \right) \oplus S_{t,h}.
\end{equation}
This direct sum generalizes the splitting 
$F(t) = \left(F(t)^\prime+ F\langle  t \rangle \right) \oplus S_t$ induced by Algorithm {\sc HermiteReduce} in \cite[Section 5.6]{BronsteinBook}.
It remains to find a $C$-subspace $W \subset F\langle  t \rangle$ with
$$\im(\cR_h) + \left( b^{-1} \cdot F\langle  t \rangle \right)  = \im(\cR_h) \oplus \left( b^{-1} \cdot W \right).$$

In the special case $F(t)=C(x)$, such a subspace can be constructed via 
a polynomial reduction map in \cite[Section 4.1]{BCCLX2013}.
This motivates us to define the notion of companion operators.
\begin{define} \label{DEF:comp}
The map
\[
 \begin{array}{cccc}
\cP_h: & F\langle  t \rangle & \rightarrow & F\langle  t \rangle \\
       &   r                & \mapsto     & b r^\prime + a r.
\end{array}
\]
is called the {\em companion operator} of $\cR_h$.
\end{define}
The following proposition allows us to reduce the problem of constructing a complete reduction for $(F(t), \cR_h)$
to that for $(F\langle  t \rangle, \cP_h)$. 
\begin{prop} \label{PROP:red}
If $W$ is a complement of $\im(\cP_h)$ in $F\langle  t \rangle$, then
\begin{equation} \label{EQ:dsum}
F(t) = \im(\cR_h) \oplus \left(b^{-1} \cdot W \oplus S_{t,h}\right).
\end{equation}
\end{prop}
\begin{proof}
By \eqref{EQ:pre1},
$F(t) = \left(\im(\cR_h) +  b^{-1} \cdot \im(\cP_h) + b^{-1} \cdot W \right) \oplus S_{t,h}.$
Then
$$F(t) = \left(\im(\cR_h) +  b^{-1} \cdot W \right) \oplus S_{t,h}$$
by $b^{-1} \cdot \im(\cP_h) \subset \im(\cR_h)$.
It remains to show that $\im(\cR_h) \cap \left(b^{-1} \cdot W \right) = \{0\}$.
Assume that $v \in \im(\cR_h) \cap \left(b^{-1} \cdot W \right)$. Then $v = w/b \in \im(\cR_h)$ for some $w \in W$, and thus
there exists $g \in F(t)$ such that $b g^\prime + a g = w$.

Since $h$ is $t$-normalized,  we have $g \in F\langle  t \rangle$ by either \cite[Theorem 6.1.2]{BronsteinBook}
or \cite[Theorem 3.12 (iii)]{CDGL2025}.  Hence, $w \in \im(\cP_h)$.
By $\im(\cP_h) \cap W = \{0\}$, $w$ is equal to $0$, and so is $v$.
\end{proof}
\begin{remark} \label{RE:red}
With the notation from the above proposition, we further let
$\Phi_{h, W}$ be the projection from $F(t)$ to $(b^{-1} \cdot W \oplus S_{t,h})$ with respect to \eqref{EQ:dsum},
and $\Psi_{h,W}$ be the projection from $F\langle t \rangle$ to $W$
with respect to $F\langle t \rangle = \im(\cP_h) \oplus W$. Then $\Phi_{h, W}$ and $\Psi_{h,W}$
are complete reductions for $(F(t), \, \cR_h)$ and $(F\langle t \rangle, \cP_h)$, respectively.

Assume that one can compute R-pairs with respect to $\Psi_{h,W}$. Then R-pairs with respect to~$\Phi_{h,W}$ can be derived
as follows. Let  $f \in F(t)$.
\begin{enumerate}
\item  Find $(g, r, s) \in F(t) \times F\langle t \rangle \times S_{t,h}$ satisfying \eqref{EQ:pre2} by Algorithm {\sc GKSR} in  \cite[Section 3.4]{CDGL2025}. 
\item Compute an R-pair $(u,v) \in F\langle t \rangle \times W$ of $r$ with respect to $\Psi_{h,W}$.
\end{enumerate}
Then the pair $(g+u, \, b^{-1} v + s)$ is an R-pair of $f$  with respect to $\Phi_{h, W}$.
 This follows from  the fact $b^{-1} v + s \in \left(b^{-1} \cdot  W \right)+ S_{t,h}$ 
and the calculation:
\begin{align*}
f & = \cR_h(g) + b^{-1} r + s &  \text{(by \eqref{EQ:pre2})}\\
  & = \cR_h(g) + b^{-1} \left(\cP_h(u) + v \right) + s &  \text{(by $r = \cP_h(u)+v$)} \\
  & = \cR_h(g) + \cR_h(u) + b^{-1} v + s & \text{(by $\cR_h(u) = b^{-1} \cP_h(u)$)} \\ 
  & = \cR_h(g+u) + b^{-1} v + s. 
\end{align*}
\end{remark}
\begin{remark} \label{RE:der}
If $h=0$, 
then $b=1$, and \eqref{EQ:dsum} in Proposition \ref{PROP:red} simplifies to
$$F(t) = F(t)^\prime \oplus (W \oplus S_t),$$
which leads to a complete reduction for $(F(t), \, ^\prime)$.
\end{remark}

\section{ Induction hypothesis and notational convention} \label{SECT:ind}

With the problem of constructing complete reductions now reduced to the setting of differential subrings, 
we proceed to formalize the induction framework that will drive our construction.

Let $n \in \bN$ and $C(t_1, \ldots, t_{n-1}, t_n)$ be a transcendental Liouvillian extension of $C$ throughout this section.
For inductive purposes, we set  $F := C(t_1, \ldots, t_{n-1})$ and $t :=t_n$.
By \cite[equation (5.1)]{BronsteinBook}, 
the differential ring $F\langle t\rangle$ takes the form $F[t]$ when $t$ is primitive, and $F[t,t^{-1}]$ when $t$ is hyperexponential. This dichotomy will structure our inductive arguments in Sections \ref{SECT:prim} and \ref{SECT:hyperexp}.

With the field setting fully specified, we now state the induction hypothesis that will be used throughout the remainder of our inductive construction.
\begin{hyp} \label{HYP:ind}
Assume that there is a map that assigns to each element $\alpha$ of $F$ a
complete reduction $\Phi_\alpha$ for $(F, \cR_\alpha)$, and that there is an algorithm for
computing an R-pair for every element of $F$ with respect to $\Phi_\alpha$.
\end{hyp}
By the hypothesis, $\im(\cR_\alpha) \cap \im(\Phi_\alpha) = \{0\}$ for all $\alpha \in F$, and every element $\beta \in F$ has an
R-pair $(\gamma, \Phi_\alpha(\beta)) \in F^2$ with respect to $\Phi_\alpha$, that is, $\beta = \cR_\alpha(\gamma)+ 
\Phi_\alpha(\beta)$.

The base case of our induction is given by the following example.
\begin{example} \label{EX:choice0}
Let $F=C$.
The map in Hypothesis \ref{HYP:ind} assigns to zero the identity operator
and to every element of $C^\times$ the zero operator by Example \ref{EX:const}.
\end{example}

For every $h \in F(t)$, Propositions \ref{PROP:ks} and \ref{PROP:red} convert the problem 
of constructing complete reductions for $(F(t), \cR_h)$ to that for $(F\langle  t \rangle, \cP_h)$ with an additional property that $h$ is  $t$-normalized.

Let $h \in F(t)$ be $t$-normalized. In the next two sections, we will construct a complete reduction
for $(F\langle t \rangle, \cP_h)$ in three steps:
\begin{enumerate}
\item Construct a subspace $A_h \subset F\langle t \rangle $ such that $F\langle t \rangle = \im(\cP_h) + A_h$
by \eqref{EQ:paux_case}, \eqref{EQ:haux_case+} and \eqref{EQ:haux_case-}.
\item Find an echelon sequence $E_h$ of $\im(\cP_h) \cap A_h$ by Propositions \ref{PROP:image}, \ref{PROP:hbasis} and Lemma \ref{LM:elim}.
\item Determine a complement $W$ of $\im(\cP_h) \cap A_h$ in $A_h$ by $E_h$ and Lemma \ref{LM:modulo}.
\end{enumerate}
Once these steps are completed, we obtain the direct sum decomposition
$$F\langle t \rangle = \im(\cP_h) \oplus W,$$ 
and the projection from  $F\langle t \rangle$ to $W$ gives the desired complete reduction.
The above three steps have been carried out in \cite[Section 3]{DGLL2025} when $t$ is primitive and $h=0$.
However, additional technical obstacles emerge when $h \neq 0$ or $t$ is hyperexponential.

The notion of auxiliary subspaces comes  from our attempt to simplify the coefficients of elements in  $F\langle t \rangle$ modulo $\im(\cP_h)$.
Such a simplification relies on an elaborate application of integration by parts in the case $h=0$ (see \cite{CDL2018,DGLW2020}).

A basis of $\im(\cP_h) \cap A_h$ is constructed based on two algorithms presented in \cite[\S 4.3.1]{RaabThesis}: one is to recognize logarithmic derivatives in $F$  (see \cite[\S 5.12]{BronsteinBook}); the other is to find parametric logarithmic derivatives in $F$ (see \cite[\S 7.3]{BronsteinBook} and \cite[Chapter 4]{RaabThesis}). Note that these two algorithms are well developed for transcendental Liouvillian extensions, while they cannot be guaranteed to work for general differential fields.
Although such a basis may not be finite when $t$ is primitive, its regular shape described in Proposition \ref{PROP:image} allows us
to find an echelon sequence of the intersection, which induces a complement of $\im(\cP_h)$ in $F\langle t \rangle$
by Lemma \ref{LM:modulo},
 and leads to an algorithm for computing R-pairs by Lemma \ref{LM:elim}.

We present a lemma
allowing us to transform a $C$-basis of a subspace in $F\langle t \rangle$
to a $C$-basis of its image under $\cP_h$. 
\begin{lemma} \label{LM:inj}
If $h \in F(t) \setminus F$ is $t$-normalized, then both $\cR_h$ and $\cP_h$ are
injective.
\end{lemma}
\begin{proof} Suppose that $u \in \ker(\cR_h)^\times$. Then $h=-u^\prime/u$ by Definition \ref{DEF:risch}.
The logarithmic derivative identity  implies that $h$ is a $\bZ$-linear combination of logarithmic derivatives of polynomials in $F[t]$.
Since $h$ is $t$-normalized,  no logarithmic derivative of any element in $N_t$ appears in
this $\bZ$-linear combination. Since $t$ is either primitive or hyperexponential over $F$,
we derive that $h \in F$, a contradiction.
Thus, $\cR_h$ is injective. Consequently, $\cP_h$ is injective, because  it is equal to $\den(h)  \cR_h$ restricted to $F\langle  t \rangle$.
\end{proof}

As the final preparation for our induction, we make a notational convention that will be valid throughout the next two sections.
\begin{convention} \label{CON:next2}
\begin{itemize}
\item[(i)] Let $C(t_1, \ldots, t_n)$ be a transcendental Liouvillian extension of $C$. 
Set 
$$F := C(t_1, \ldots, t_{n-1}) \quad \text{and} \quad  t :=t_n.$$
\item[(ii)] Assume that $t$ is primitive over $F$ in Section \ref{SECT:prim}, and  hyperexponential over $F$ in Section \ref{SECT:hyperexp}.
\item[(iii)] $h$ is a $t$-normalized element of  $F(t)$.
\item[(iv)] Set $a :=\num(h)$, $b :=\den(h)$ and $m := \Deg(h)$. Write
\[  a =a_m t^m + a_{m-1} t^{m-1} + \cdots + a_0 \quad \text{and} \quad  b = b_m t^m + b_{m-1} t^{m-1} + \cdots + b_0,\]
where $a_i, b_i \in F$ and $b_m=1$ whenever $b_m\neq 0$.
\end{itemize}
\end{convention}

\section{The primitive case} \label{SECT:prim}

By Convention \ref{CON:next2} (ii),  $t$ is primitive throughout this section. 
Then $F\langle  t \rangle=F[t]$. 
The coefficients $a_m, b_m, a_{m-1}$ and $b_{m-1}$ in Convention \ref{CON:next2} (iv) will play a significant role
in reducing polynomials. Note that both $a_{m-1}$ and $b_{m-1}$ are zero if $m=0$.

To describe the leading terms of elements in $\im(\cP_h)$, 
we define a $C$-linear operator
\begin{equation} \label{EQ:lc}
\begin{array}{cccc}
\cL: & F & \rightarrow & F \\
     & \alpha & \mapsto     & b_m \alpha^\prime + a_m \alpha.
\end{array}
\end{equation}
Then,   for an element $f \in F[t]$ with $d = \deg(f)$ and $f_d = \lc(f)$, we have 
\begin{equation} \label{EQ:pred}
\cP_h(f) \equiv \cL(f_d) t^{m+d} \mod F[t]_{<m+d}
\end{equation}
by $\cP_h(f) \equiv (b_m f_d^\prime) t^{m+d}  + (a_m f_d) t^{m+d}  \mod F[t]_{<m+d}$.
\begin{remark} \label{RE:pred}
By Convention \ref{CON:next2} and \eqref{EQ:lc}, the following two assertions hold.
\begin{itemize}
\item[(i)] If $\nu_\infty(h) < 0$, then $a_m \neq 0$ and $b_m=0$. Consequently, $\cL$ is injective.
\item[(ii)] If $\nu_\infty(h) \ge 0$, then $b_m = 1$ and $\cL = \cR_{a_m}$, 
which is the Risch operator on $F$ associated to~$a_m$.
\end{itemize}
\end{remark}

The next lemma  will be used frequently for constructing echelon sequences in this section.
\begin{lemma} \label{LM:lead}
If $\lambda \in \ker(\cL)^\times$, then $\Phi_{a_m}(\lambda t^\prime) \neq 0$, where $\Phi_{a_m}$ is the complete reduction 
for~$(F, \, \cR_{a_m})$ given
in Hypothesis \ref{HYP:ind}.
\end{lemma}
\begin{proof}
Since $\ker(\cL)\neq \{0\}$, we have $\nu_\infty(h)\geq 0$ and then $\cL = \cR_{a_m}$ by 
Remark \ref{RE:pred}. It follows that $\lambda \in \ker(\cR_{a_m})^\times.$ Thus, $\lambda'+a_m\lambda=0.$
Suppose that $\Phi_{a_m}\left( \lambda t^\prime \right) = 0$. Then $\lambda t^\prime \in \im(\cR_{a_m})$.
Hence, there exists $\mu \in F$ such that
$\mu^\prime + a_m \mu = \lambda t^\prime.$
The two equalities involving $a_m$ imply that $t^\prime = (\mu/\lambda)^\prime$, a contradiction to the regularity of $t$.
\end{proof}

The rest of this section has three parts. We first define the auxiliary subspace $U_h$ associated to $(F[t], \, \cP_h)$ and develop a corresponding reduction
from $F[t]$ to $U_h$ modulo $\im(\cP_h)$ in Section~\ref{SUBSECT:paux}.
Section \ref{SUBSECT:pinter} presents a criterion to determine whether $\im(\cP_h) \cap U_h$ is trivial.
An echelon sequence of this intersection, together with a complete reduction for $(F(t), \cR_h)$,
is constructed in Section~\ref{SUBSECT:pbasis}.

\subsection{Auxiliary subspaces} \label{SUBSECT:paux}

The idea of auxiliary subspaces originates from the attempt to decompose an element $f \in F[t]$ as $g^\prime + r$,
where $^\prime = \cR_0$ is the derivation on $F[t]$, $g,r \in F[t]$, $\deg(r) \le \deg(f)$, and all coefficients of~$r$ lie
in $\im(\Phi_0)$. Note that $\Phi_0$ is the complete reduction for $(F, \, ^\prime)$ given in Hypothesis \ref{HYP:ind}.  Such a decomposition
is always possible by \cite[Lemma 3.1]{DGLL2025}. 

We generalize this idea from $^\prime$ to $\cR_h$, where
$h$ is given in Convention \ref{CON:next2} (iii) and (iv).

To this end, we let $f \in F[t]$, $d=\deg(f)$ and $f_d=\lc(f)$, and assume that $d \ge m$.  Then
\begin{equation} \label{EQ:paux_case}
f \equiv 
\begin{cases}
\cP_h\left(a_m^{-1}  f_{d}  t^{d-m}\right)  \mod  F[t]_{<d} & \text{if $\nu_\infty(h)<0$}, \\ \\
\cP_h\left(\alpha t^{d-m}\right) + \beta t^d \mod  F[t]_{<d}  & \text{if $\nu_\infty(h) \ge 0$}, 
\end{cases}
\end{equation}
where $(\alpha,\beta) \in F^2$ is an R-pair of $f_d$ with respect to $\Phi_{a_m}$ given in Hypothesis \ref{HYP:ind}.

The first congruence holds by expanding $\cP_h\left(a_m^{-1}  f_{d}  t^{d-m}\right)$ directly modulo $F[t]_{<d}$ and noticing that $b_m=0$ due to $\nu_\infty(h)<0$.
To show the second congruence, we note that $f_d = \cR_{a_m}(\alpha)+\beta$ by Definition \ref{DEF:Rpair}. 
Then $f_d =  \cL(\alpha)  + \beta $ by Remark \ref{RE:pred} (ii). So $f \equiv  \cL(\alpha) t^d + \beta t^d \mod F[t]_{<d}$.
On the other hand, $\cP_h\left(\alpha t^{d-m}\right) \equiv \cL(\alpha) t^d \mod F[t]_{<d}$ by \eqref{EQ:pred}. 
Combining the last two congruences, we prove the second congruence in \eqref{EQ:paux_case}. 

Iteratively applying the two congruences in \eqref{EQ:paux_case} to elements of $F[t]_{<d}$, 
we arrive at the following definition.
\begin{define} \label{DEF:paux}
The {\em auxiliary subspace associated to $(F[t], \cP_h)$} is defined as
\[
U_h := \left\{ \begin{array}{ll}
 F[t]_{<m} &   \text{if $\nu_\infty(h)<0$},  \\ \\
 F[t]_{<m} + \sum_{i \in \bN_0} \im(\Phi_{a_m}) \cdot t^{m+i} &  \text{if $\nu_\infty(h) \ge 0$}.
\end{array} \right.
\]
\end{define}
The set $U_h$ is a $C$-subspace of $F[t]$ because both $F[t]_{<m}$ and $\im(\Phi_{a_m})$ are $C$-subspaces of $F[t]$.
Indeed, $U_h$ is $C$-linearly isomorphic to $F[t]_{<m} + \im\left(\Phi_{a_m} \right) \otimes_C C[t]$ when $\nu_\infty(h) \ge 0$.
The above  definition generalizes the notion of auxiliary subspaces in \cite{DGLL2025} for the case where $m=0$ and $a_m=0$.
\begin{example} \label{EX:paux}
Let $F=C$, $t=x$ and $^\prime=d/dx$. 
By Convention \ref{CON:next2}, 
\begin{center}
	\begin{tabular}{c|c|c|c|c|c|c} 
		   & $a$ & $b$ & $m$ & $a_m$ & $b_m$ & $\nu_\infty(h)$ \\ \hline 
         $h=0$ &  $0$ & $1$ & $0$ & $0$ & $1$ & $\infty$    \\  \hline 
         $h=x$ &  $x$ & $1$ & $1$ & $1$ & $0$ & $-1$  \\ 
	\end{tabular}
         \label{tab:example5.5}
\end{center}
It follows from Example \ref{EX:choice0} that  the auxiliary subspaces associated to $(F[t], \cP_0)$ and $(F[t], \cP_x)$
are $C[x]$ and $C$, respectively.
\end{example}
\begin{prop} \label{PROP:paux}
Let $U_h$ be the auxiliary subspace associated to $(F[t], \, \cP_h)$, and 
$f \in F[t]^\times$ with $d = \deg(f)$. Then there exists a pair $(g, r) \in F[t] \times U_h$ with $\deg(g) \le d - m$ and $\deg(r) \le d$
such that $f = \cP_h(g)+r$. Consequently,   $F[t] = \im(\cP_h) + U_h.$
\end{prop}
\begin{proof}
If $d<m$, then $f \in U_h$. So it suffices to set $g:=0$ and $r:=f$. 
Assume that $d \ge m$, and that the conclusion holds for polynomials of degrees lower than $d$.
Let $f_d = \lc(f)$.

If $\nu_\infty(h) < 0$, then  $f \equiv \cP_h\left(a_m^{-1}  f_{d}  t^{d-m}\right)  \mod  F[t]_{<d}$ by \eqref{EQ:paux_case}.
Applying the induction hypothesis to $F[t]_{<d}$, we see that the proposition holds.

Otherwise, $\nu_\infty(h) \ge 0$. Let $(\alpha, \beta)$ be an R-pair of $f_d$ with respect to $\Phi_{a_m}$.
Then 
$$f \equiv \cP_h(\alpha t^{d-m}) \mod U_h + F[t]_{<d}$$
by the second congruence in \eqref{EQ:paux_case} and $\beta \in \im\left(\Phi_{a_m}\right)$. 

It follows from the induction hypothesis and $F[t]_{<m} \subset U_h$ that $f \equiv \cP_h(g) \mod U_h$ for some $g \in F[t]$ with $\deg(g) \le d-m$.
Furthermore,  $\deg\left(\cP_h(g)\right) \le d$ by~\eqref{EQ:pred}. 
The conclusion holds by setting $r := f - \cP_h(g)$.
\end{proof}
\begin{define} \label{DEF:ppair}
Let $f$ and $(g,r)$ be given in the above proposition. We call $(g,r)$ an {\em auxiliary pair} of $f$ with respect to $(F[t], \cP_h)$,
or an {\em auxiliary pair} of $f$ if  $(F[t], \cP_h)$ is clear from context.
\end{define}
The proof of Proposition \ref{PROP:paux} leads to  Algorithm \ref{ALG:paux} 
in Section \ref{SUBSECT:palg}.
The next corollary reveals a relation between auxiliary pairs and R-pairs in the special case $h \in F$.
\begin{cor} \label{COR:pinfield}
Let $h \in F$, and $\Phi_h$ be the complete reduction for $(F, \, \cR_h)$ in Hypothesis \ref{HYP:ind}. 
Then, for every element $f \in F$,  an auxiliary pair of $f$ with respect to $(F[t], \cP_h)$ is an R-pair of~$f$
with respect to $\Phi_h$.
\end{cor}
\begin{proof}
Since $f \in F$, its degree is zero. So every auxiliary pair $(\alpha, \beta)$ of $f$ with respect to $(F[t], \cP_h)$ belongs to $F^2$ by Proposition \ref{PROP:paux}.
It remains to show that $(\alpha, \beta)$ is an R-pair of $f$
with respect to $\Phi_h$. 
With Convention \ref{CON:next2}, we have the following table for notation:
\begin{center}
	\begin{tabular}{c|c|c|c|c|c|c}
		  & $a$ & $b$ & $m$ & $a_m$ & $b_m$ & $\nu_\infty(h)$  \\ \hline
         $ h \in F$ &  $h$ & $1$ & $0$ & $h$ & $1$ & $ \ge 0$ \\  
	\end{tabular}
\end{center}
Then the auxiliary subspace $U_h$ associated to $(F[t],\cP_h)$ is $\sum_{i \in \bN_0} \im(\Phi_{h}) \cdot t^{i}$
by Definition \ref{DEF:paux}. So $\beta \in \im(\Phi_h)$ by $\beta \in U_h \cap F$. 
On the other hand,  the restrictions of $\cR_h$ and $\cP_h$ to $F$ are equal by Definitions \ref{DEF:risch} and \ref{DEF:comp}.
Therefore, $f = \cP_h(\alpha) + \beta = \cR_h(\alpha) + \beta$. 
It follows from $F = \im(\cR_h) \oplus \im(\Phi_h)$ that $\beta = \Phi_h(f)$, that is, $(\alpha, \beta)$ is an R-pair with respect to $\Phi_h$.
\end{proof}
\subsection{Intersecting $\im(\cP_h)$ with the associated auxiliary subspace} \label{SUBSECT:pinter}

In this subsection, we let $U_h$ be the auxiliary subspace associated to $(F[t], \, \cP_h)$.
By Proposition~\ref{PROP:paux}, $F[t] = \im(\cP_h) + U_h$. Then $F[t] = \im(\cP_h) \oplus U_h$ if and only if  $\im(\cP_h) \cap U_h {=} \{0\}$.
By setting $I_h := \im(\cP_h) \cap U_h$, 
we aim to derive a necessary and sufficient condition on $I_h = \{0\}.$

The following lemma describes the leading coefficient of every  element in $I_h$.
\begin{lemma} \label{LM:pinter}
If $\cP_h(f) \in U_h$ for some $f \in F[t]$, then $\lc(f) \in \ker(\cL).$
\end{lemma}
\begin{proof}
Let $d = \deg(f)$ and $f_d = \lc(f)$.

If $\nu_\infty(h)<0$, then $U_h = F[t]_{<m}$
by Definition \ref{DEF:paux}. So the degree of $\cP_h(f)$ is less than $m$.
It follows from \eqref{EQ:pred} that $\cL(f_d)=0$.

Assume that $\nu_\infty(h) \ge 0$.  Then $\cL(f_d) \in \im(\Phi_{a_m})$ by \eqref{EQ:pred}, $\cP_h(f) \in U_h$ and Definition \ref{DEF:paux}.
 Note that
$\cL=\cR_{a_m}$ by Remark \ref{RE:pred} (ii). Since $\im\left(\cR_{a_m} \right) \cap  \im\left(\Phi_{a_m}\right)$ is equal to $\{0\}$,
so is $\im\left(\cL \right) \cap  \im\left(\Phi_{a_m}\right)$. 
Accordingly, $\cL(f_d)=0$.
\end{proof}
The next proposition is a criterion for $I_h = \{0\}$.
\begin{prop} \label{PROP:pinter}
$I_h=\{0\}$ if and only if $\ker(\cL)=\{0\}$.  In particular, $I_h=\{0\}$ if $\nu_{\infty}(h)<0$.
\end{prop}
\begin{proof} Assume that $\ker(\cL) = \{0\}.$
Let $\cP_h(f) \in U_h$ for some $f \in F[t]$.
By Lemma \ref{LM:pinter}, $\lc(f)$ is equal to $0$, and so is $f$. Thus, $I_h = \{0\}$.

Conversely, assume that $I_h = \{0\}$.  Let $\lambda \in \ker(\cL)$. We need to show $\lambda=0$.
By \eqref{EQ:pred}, we have that $\cP_h(\lambda) \in F[t]_{<m}$. 
Then $\cP_h(\lambda)  \in U_h$  by Definition \ref{DEF:paux}.
Consequently, $\cP_h(\lambda)=0$ by $I_h=\{0\}$.

If $h \in F(t) \setminus F$, then $\lambda=0$ by Lemma \ref{LM:inj}.

Assume that 
$h \in F$. Then $m=0$, $a=a_m=h$ and $b=b_m=1$ by Convention \ref{CON:next2}. 
So $\cP_h( \lambda t) = \lambda  t^\prime$ by $\lambda \in \ker(\cL)$.
Since $\lambda t^\prime \in F$, it has an auxiliary pair of the form 
$\left(\mu, \Phi_{h}\left( \lambda t^\prime\right) \right)$ for some $\mu \in F$ by Corollary \ref{COR:pinfield},
where $\Phi_h$ is the complete reduction for $(F, \, \cR_h)$ given in Hypothesis~\ref{HYP:ind}.
Thus $\cP_h( \lambda t) = \lambda  t^\prime = \cP_h(\mu) + \Phi_{h}\left( \lambda t^\prime\right).$
Consequently,
$\cP_h( \lambda t- \mu) = \Phi_{h}\left( \lambda t^\prime\right)$,
which, together with $\Phi_{h}\left( \lambda t^\prime\right) \in U_h$, implies that
 $\Phi_{h}\left( \lambda t^\prime\right) \in I_h$.
We conclude   from $I_h = \{0\}$ that $\Phi_{h}\left( \lambda t^\prime\right) = 0$. Then $ \lambda =0$ by Lemma \ref{LM:lead}.

If $\nu_\infty(h)<0$, then $\ker(\cL)=\{0\}$ by Remark \ref{RE:pred} (i). So $I_h=\{0\}$.
\end{proof}

We distinguish whether $I_h$ is trivial or not by types defined below.
\begin{define} \label{DEF:pindex}
The {\em type} of $I_h$ is defined as $0$ if $I_h=\{0\}$ (i.e., $\ker(\cL) = \{0\}$) and otherwise, as an element $\lambda\in \ker(\cL)^\times$.
\end{define}
The type of $I_h$ is unique up to a multiplicative constant in $C^\times$, because $\dim_C\ker(\cL) \le 1$.
It can be found by recognizing logarithmic derivatives in $F$.
\begin{example} \label{EX:type}
Let $h=0$. Then $m=0$, $a_m=0$ and $b_m=1$ by Convention \ref{CON:next2}. Thus, $\cL= \, ^\prime$ by \eqref{EQ:lc}.
Consequently, $\ker(\cL)=C$. The type of $I_0$ can be taken as $1$.
\end{example}

\subsection{Echelon sequences and induced complements} \label{SUBSECT:pbasis}

The goal of this subsection is to construct an echelon sequence of $I_h$,  yielding a complement of~$\im(\cP_h)$ in $F[t]$.
Let $U_h$ be the auxiliary subspace associated to $(F[t], \, \cP_h)$,
and set $I_h := \im(\cP_h) \cap U_h$ as before. 
Assume further that $I_h \neq \{0\}$ and its type is a fixed element $\lambda \in \ker(\cL)^\times.$ 
 We proceed  in the following four steps:
\begin{enumerate}
\item Construct a $C$-basis $B_1$ of $\cP_h^{-1}(I_h)$ in Lemma \ref{LM:preimage}.
\item Describe a $C$-basis $B_2$ of $I_h$ in Proposition \ref{PROP:image} by mapping $B_1$ 
to $I_h$ via~$\cP_h$ with a minor modification.
\item Convert $B_2$ into an echelon sequence of $I_h$ 
in Corollaries \ref{COR:m=0} and \ref{COR:pivot0}.
\item Construct a complement of $\im(\cP_h)$ in $F[t]$ by Lemma \ref{LM:modulo}.
\end{enumerate}
In the rest of this subsection, every basis is a $C$-basis. So we omit the prefix \lq\lq $C$- \rq\rq\ for brevity.
\subsubsection{Standard basis of $\cP_h^{-1}(I_h)$} \label{SUBSUBSECT:standard}

Recall the notation given in Convention \ref{CON:next2}. We have  
$$m=\Deg(h),  \quad \num(h)=\sum_{i=0}^m a_i t^i \quad \text{and} \quad  \den(h) = \sum_{i=0}^m b_i t^i.$$
By Proposition \ref{PROP:pinter}, $\nu_{\infty}(h) \ge 0$. So $b_m=1$ and $\cL=\cR_{a_m}$ by Remark \ref{RE:pred} (ii).

In view of Lemma \ref{LM:pinter} and $\dim_C(\ker(\cL))=1$, all elements of $\cP_h^{-1}(I_h)^\times$ must have leading coefficient $c \lambda$ for 
some $c \in C^\times$. This observation motivates us to compute the leading term of $\cP_h(\lambda  t^i)$ for all $i \in \bN_0$. 
By Definition \ref{DEF:comp} and $\cL( \lambda )=0$,  we have that 
\begin{equation} \label{EQ:image}
\cP_h(\lambda  t^i) \equiv  (i \lambda t^\prime + b_{m-1} \lambda^\prime +  a_{m-1} \lambda) \cdot t^{m+i-1}  \mod  F[t]_{<m+i-1}
\end{equation}
for all $i \in \bN_0$. To express the coefficient $i \lambda t^\prime + b_{m-1} \lambda^\prime +  a_{m-1} \lambda$ of $t^{m+i-1}$
in terms of R-pairs, we introduce the following definition.
\begin{define} \label{DEF:pairs}
Let $\Phi_{a_m}$ be the complete reduction for $(F, \, \cR_{a_m})$ given in Hypothesis \ref{HYP:ind}. 
A {\em first R-pair associated to $(F[t], \cP_h)$} is an R-pair of $\lambda t^\prime$ with respect to  $\Phi_{a_m}$,
and a {\em second one} 
is an R-pair of $b_{m-1} \lambda^\prime  +  a_{m-1} \lambda$ with respect to $\Phi_{a_m}$.
\end{define}
An element may admit R-pairs with different first components. For the rest of this section, 
we fix the first and second associated R-pairs, and denote them by $(\tilde{\sigma},\sigma)$ and~$(\tilde{\tau}, \tau)$, respectively.
Note that $(\tilde{\tau}, \tau):=(0,0)$ if $m=0$.
\begin{remark} \label{RE:nonzero}
The entry $\sigma$ in the first associated R-pair $(\tilde{\sigma}, \sigma)$ is nonzero by Lemma \ref{LM:lead}.
\end{remark}

With $\lambda t^\prime = \cR_{a_m}(\tilde{\sigma})+ \sigma$ and $b_{m-1} \lambda^\prime +  a_{m-1} \lambda = 
\cR_{a_m}(\tilde{\tau})+\tau$, we rewrite \eqref{EQ:image} as 
\begin{equation} \label{EQ:pairs}
\cP_h(\lambda  t^i) \equiv  \cR_{a_m}\left(i \tilde{\sigma} + \tilde{\tau} \right) \cdot t^{m+i-1}  + (i \sigma + \tau) \cdot t^{m+i+1} \mod  F[t]_{<m+i-1}.
\end{equation}

In order to construct a basis for $\cP_h^{-1}(I_h)$, we first define a sequence $\{p_i\}_{i \in \bN_0}$ in $F[t]$ such that $\deg(p_i) =i$, $\lc(p_i)=\lambda$ and
$p_i \in \cP_h^{-1}(I_h)$.

For $i=0$, we let $p_0=\lambda$. Then $\cP_h(p_0) \in F[t]_{<m}$ by \eqref{EQ:pred}.
Since $F[t]_{<m} \subset U_h$, we have $\cP_h(p_0) \in I_h$ 
and, consequently, $p_0 \in \cP_h^{-1}(I_h)$.

Assume that $i>0$. It follows from \eqref{EQ:pairs} and  $\cL=\cR_{a_m}$ that 
$$\cP_h(\lambda t^i) \equiv  \cL(i \tilde{\sigma} + \tilde{\tau}) \cdot t^{m+i-1} + (i \sigma + \tau) \cdot t^{m+i-1}  \mod  F[t]_{<m+i-1},$$
which, together with \eqref{EQ:pred}, implies that there exists a polynomial $g_i \in F[t]_{<m+i-1}$ such that
$$\cP_h(\lambda t^i) =  \cP_h\left((i \tilde{\sigma} + \tilde{\tau}) \cdot t^{i-1} \right) + (i \sigma+  \tau) \cdot  t^{m+i-1} + g_i.$$
Let $(q_i, r_i)$ be the auxiliary pair of $g_i$ obtained from Algorithm \ref{ALG:paux}.
Then $\deg(q_i)<i-1$, $\deg(r_i) < m+i-1$ by Proposition \ref{PROP:paux}. Moreover, 
$$\cP_h(\lambda t^i) =  \cP_h\left((i \tilde{\sigma} + \tilde{\tau}) \cdot t^{i-1} \right) + (i \sigma+ \tau) \cdot t^{m+i-1} + \cP_h(q_i)+r_i.$$
Moving the images under $\cP_h$ in the above equality to the left-hand side, we arrive at  
\[ \cP_h \left(\lambda t^i - (i \tilde{\sigma} + \tilde{\tau}) \cdot t^{i-1} -q_i \right) =  (i \sigma+ \tau) \cdot t^{m+i-1} +r_i. \]
Set
\begin{equation} \label{EQ:prel1}
p_i := \lambda t^i - (i \tilde{\sigma} + \tilde{\tau}) \cdot t^{i-1}  - q_i.
\end{equation}
Then $\deg(p_i)=i$ and
\begin{equation} \label{EQ:prel2}
\cP_h(p_i) = (i \sigma+  \tau) \cdot t^{m+i-1} + r_i.
\end{equation}
Since $\sigma, \tau \in \im(\Phi_{a_m})$ and $r_i \in U_h$, we see that $\cP_h(p_i) \in U_h$, that is,
$p_i \in \cP_h^{-1}(I_h)$ for all $i \in \bN$.
\begin{lemma} \label{LM:preimage}
Let $p_0= \lambda$ and $p_i$ be defined by \eqref{EQ:prel1} for all $i \in \bN$.
Then $\{p_i\}_{i \in \bN_0}$ is a basis of~$\cP_h^{-1}(I_h)$.
\end{lemma}
\begin{proof}
We have seen that $\{p_i\}_{i \in \bN_0} \subset \cP_h^{-1}(I_h)$.
Since $\deg(p_i)=i$ for all $i \in \bN_0$, the set $\{p_i\}_{i \in \bN_0}$ is $C$-linearly independent.
Taking $f \in F[t]$ such that $\cP_h(f) \in U_h$, we see that $\lc(f) \in \ker(\cL)$ by Lemma \ref{LM:pinter}.
Then $\lc(f)=c \lambda$ for some $c \in C$ by  $\lambda \in \ker(\cL)^\times$ and $\dim_C \left(\ker(\cL)\right)=1$. 

Let $d = \deg(f)$.
Then $f - c p_d \in \cP_h^{-1}(I_h)$ with $\deg(f-cp_d) < d$.
Thus,  $f$ is a $C$-linear combination of $p_d,$  \ldots, $p_0$ by a straightforward induction on $d$.
\end{proof}
We call $\{p_i\}_{i \in \bN_0}$ the {\em standard basis} of $\cP_h^{-1}(I_h)$, where $p_0=\lambda$ and $p_i$ is given in \eqref{EQ:prel1} for $i \in \bN$.
This basis is unique, because the type of $I_h$ and the two associated R-pairs are all fixed, and every auxiliary pair
$(q_i, r_i)$ is computed by Algorithm \ref{ALG:paux}.

Let us summarize the notation introduced so far in the following table.
\begin{center}
	\begin{tabular}{c|c|c|c|c} 
		 & Type & 1st R-pair & 2nd R-pair  & Standard basis \\ \hline
      Expression   & $\lambda$       &  $(\tilde{\sigma}, \sigma)$ & $(\tilde{\tau}, \tau)$ &  $p_0 =\lambda$, $p_i$ in \eqref{EQ:prel1}, $i \in \bN$ \\ \hline
      Property & $\lambda \in \ker(\cL)^\times$ & 
      $\lambda t^\prime = \cR_{a_m}(\tilde{\sigma})+ \sigma$ & 
      $ \begin{array}{l}
      b_{m-1} \lambda^\prime +  a_{m-1} \lambda \\ 
      = 
\cR_{a_m}(\tilde{\tau})+\tau 
\end{array}$ &  A basis of $\cP_h^{-1}(I_h)$ \\ 
	\end{tabular}
         \captionof{table}{Notational summary in Section \ref{SUBSUBSECT:standard}} 
         \label{tab:not}
\end{center}

\subsubsection{Basis of $I_h$}

We transform the standard basis derived above into a basis of $I_h$ using the companion operator~$\cP_h$ as 
formulated below.
\begin{prop} \label{PROP:image}
With the notation given in Table \ref{tab:not}, we further let 
$r_i$ be given in \eqref{EQ:prel2}.
\begin{itemize}
\item[(i)] If $h \in F$, then $\{ \cP_h(p_i) \}_{i \in \bN}$ is a basis of $I_h$, and $\cP_h(p_i) = i \sigma t^{i-1}+ r_i$ with degree $i-1$.
\item[(ii)] If $h \in F(t) \setminus F$,  then $\{ \cP_h(p_i) \}_{i \in \bN_0}$ is a basis of $I_h$, $\cP_h(p_0) \in F[t]_{<m}^\times$,
and, for all $i \in \bN$,
$$\cP_h(p_{i}) = (i\sigma + \tau) t^{m+i-1}+ r_{i}.$$
\end{itemize}
\end{prop}
\begin{proof} 
(i) Since $h \in F$, the restriction  of $\cP_h$ to $F$ is equal to $\cL$. Then $\cP_h(p_0)=0$ by $p_0= \lambda$.
Hence, $\{ \cP_h(p_i) \}_{i \in \bN}$ spans $I_h$ over $C$ by Lemma \ref{LM:preimage}.
Moreover, $h \in F$ implies that $m=0$ and $(\tilde{\tau}, \tau)=(0, 0)$.
It follows from \eqref{EQ:prel2} that $\cP_h(p_i) = i \sigma t^{i-1}+ r_i.$
This image is of degree $i-1$ by Remark \ref{RE:nonzero}
and $\deg(r_i)<i-1$.
Thus,
$\{ \cP_h(p_i) \}_{i \in \bN}$ is $C$-linearly independent, and then it is a basis of $I_h$.

(ii) Since $h \in F(t) \setminus F$, the companion operator $\cP_h$ is injective by Lemma \ref{LM:inj}, which, together with  Lemma \ref{LM:preimage}, implies that $\{ \cP_h(p_{i}) \}_{i \in \bN_0}$ is  a basis
of $I_h$. Note that $\cP_h(p_0)=\cP_h(\lambda)$, which is a nonzero polynomial of degree less than $m$ by $\lambda \in \ker(\cL)$ and \eqref{EQ:pred}.
The rest follows by \eqref{EQ:prel2}. 
\end{proof}
The following example rediscovers the complete reduction in
Example \ref{EX:rational}.
\begin{example} \label{EX:cr-rational}
Let $F=C$ and $t=x$ with $x^\prime=1$. Then $F(t)=C(x)$. 
 Set $h=0$. Then the auxiliary subspace $U_0$ is equal to $C[x]$ by Example \ref{EX:paux}.
 The calculations carried out 
for constructing the standard basis of $\cP_0^{-1}(I_0)$ lead to the following table:
\begin{center}
	\begin{tabular}{c|c|c|c|c} 
		 Type $\lambda$ & 1st R-pair $(\tilde{\sigma},\sigma)$ & 2nd R-pair $(\tilde{\tau}, \tau)$ & $p_i$, $i \in \bN_0$ &
$\cP_0(p_i)$, $i \in \bN$ \\ \hline
                $1$       &  $(0,1)$ & $(0,0)$ & $x^i$  & $i x^{i-1}$ \\ 
	\end{tabular}
\end{center}
Then $I_0 = C[x]$ by Proposition \ref{PROP:image} (i). Consequently, 
$C[x] = \im(\cP_0)$. Hence, $\{0\}$ is the only complement of $\im(\cP_0)$ in $C[x]$.
By Proposition \ref{PROP:red}, $C(x) = C(x)^\prime \oplus S_x$.
\end{example}

\subsubsection{Echelon sequences}
We construct echelon sequences of $I_h$ according to whether $h \in F$ or not.
We refer to Definition \ref{DEF:eseq} for the definition of echelon sequences.

Let $\Theta$ be an effective basis of $F$ (see Definition \ref{DEF:eff}). Then 
\begin{equation} \label{EQ:eeff}
\bfOmega := \{ \theta t^i \mid \theta \in \Theta, i \in \bN_0\}
\end{equation} is an effective basis of $F[t]$
by \cite[Remark 2.7]{DGLL2025}.

Recall that $(\tilde{\sigma}, \sigma)$ and $(\tilde{\tau}, \tau)$ are the first and second R-pairs associated 
to $(F[t], \cP_h)$, respectively. 
In the rest of this section, we fix an element $\theta_\sigma \in \Theta$ with $\theta_\sigma^*(\sigma) \neq 0$.
By Remark \ref{RE:nonzero}, such an element exists. It will be used for describing most of the pivots 
in an echelon sequence.

The next corollary  is immediate from Proposition \ref{PROP:image} (i).
\begin{cor} \label{COR:m=0}
With the notation in Table \ref{tab:not},
we further let  $h \in F$. 
Then 
$$\left\{ \left(p_i, \cP_h(p_i), \theta_\sigma t^{i-1} \right) \right\}_{i \in \bN}$$ is an echelon sequence of $I_h$
with respect to $(F[t], \cP_h)$ and $\bfOmega$.
\end{cor}

When $h \in F(t) \setminus F$, constructing an echelon sequence of $I_h$ with respect to $(F[t], \cP_h)$ and $\bfOmega$
requires a careful and tedious case study,
because the elements of  $\{ \cP_h(p_{i}) \}_{i \in \bN_0}$  in Proposition \ref{PROP:image} (ii)
do not necessarily have distinct degrees. 
A detailed case study is carried out in \cite[Section 4.1]{BCCLX2013} for the special case 
in which $F=C$, $t=x$ and $^\prime=d/dx$.
\begin{cor} \label{COR:pivot0}
With the notation given in Table \ref{tab:not},
we further let $h \in F(t) \setminus F$ with degree $m$.
Assume that $d_0$ and $l_0$ are the degree and leading coefficient of $\cP_h(p_0)$, respectively, and set $\theta_0$ to be an element of $\Theta$ such that $\theta_0^*(l_0) \neq 0$.
\begin{itemize}
\item[(i)] If $\theta_\sigma^*(i\sigma+ \tau) \neq 0$ for all $i \in \bN$, then
\[ \left(p_0, \cP_h(p_0), \theta_0 t^{d_0} \right),
\left(p_1, \cP_h(p_1), \theta_\sigma t^{m} \right),
\left(p_2, \cP_h(p_2), \theta_\sigma t^{m+1} \right),
\ldots,
\left(p_i, \cP_h(p_i), \theta_\sigma t^{m+i-1} \right), \ldots \]
is an echelon sequence of $I_h$ with respect to $(F[t], \cP_h)$ and $\bfOmega$,
which is illustrated as
\begin{center}
	\begin{tabular}{crcc} 
preimages &  \hspace{0.5cm}  basis  elements  & degrees of basis elements  &  pivots  \\ \vspace{0.1cm}
  \vdots          & \vdots \hspace{0.8cm}  &  \vdots  &  \vdots\\ \vspace{0.2cm}
$p_i$ &   $\stackrel{\cP_h(p_i)}{\text{\rule{2cm}{0.8pt}}}$ & $m+i-1$ &  $\theta_\sigma t^{m+i-1}$ \\ \vspace{0.2cm}
    \vdots           & \vdots \hspace{0.8cm} &  \vdots       &  \vdots  \\  \vspace{0.1cm}
$p_1$ &  $\stackrel{\cP_h(p_1)}{\text{\rule{1.5cm}{0.8pt}}}$ &  $m$ &  $\theta_\sigma t^{m}$ \\ \vspace{0.1cm}
$p_0$ &  $\stackrel{\cP_h(p_0)}{\text{\rule{1cm}{0.8pt}}}$ &    \hspace{0.8cm} $d_0$ with $d_0<m$ &  $ \theta_0 t^{d_0}$          
	\end{tabular}
\end{center}
where the length of a line segment corresponds to the degree of an element in the basis, and so does it in the next two
pictures.
\item[(ii)]  If $\theta_\sigma^*(j \sigma + \tau)=0$ for some $j \in \bN$ but $j \sigma + \tau \neq 0$, then
there exists an element $\theta \in \Theta$ with $\lc\left(\cP_h(p_j)\right) \notin \ker(\theta^*)$ such that 
\[ \left(p_0, \cP_h(p_0), \theta_0 t^{d_0} \right),
\left(p_1, \cP_h(p_1), \theta_\sigma t^{m} \right),
\ldots,  \left(p_{j-1}, \cP_h(p_{j-1}), \theta_\sigma t^{m+j-2} \right), \]
\[ \left(p_{j}, \cP_h(p_{j}), \theta t^{m+j-1} \right), \left(p_{j+1}, \cP_h(p_{j+1}), \theta_\sigma t^{m+j} \right),
\ldots,
\left(p_i, \cP_h(p_i), \theta_\sigma t^{m+i-1} \right), \ldots \]
is an echelon sequence of $I_h$ with respect to $(F[t], \cP_h)$ and $\bfOmega$, which is illustrated as
\begin{center}
	\begin{tabular}{crcc} 
preimages &  \hspace{0.5cm}  basis elements  \hspace{0.5cm} & degrees of basis elements  & pivots  \\ \vspace{0.1cm}
   \vdots         & \vdots \hspace{1.6cm} & \vdots    & \vdots \\  \vspace{0.2cm}
$p_i$ &   $\stackrel{\cP_h(p_i)}{\text{\rule{3.5cm}{0.8pt}}}$ & $m+i-1$ & $\theta_\sigma t^{m+i-1}$ \\ \vspace{0.2cm}
   \vdots          & \vdots \hspace{1.6cm} & \vdots     &  \vdots \\  \vspace{0.2cm}
$p_{j+1}$ &   $\stackrel{\cP_h(p_{j+1})}{\text{\rule{3cm}{0.8pt}}}$ & $m+j$  & $\theta_\sigma t^{m+j}$ \\ \vspace{0.1cm}
$p_j$ &   $\stackrel{\cP_h(p_j)}{\text{\rule{2.5cm}{0.8pt}}}$ & $m+j-1$  & $\theta t^{m+j-1}$ \\ \vspace{0.1cm}
$p_{j-1}$ &   $\stackrel{\cP_h(p_{j-1})}{\text{\rule{2cm}{0.8pt}}}$ & $m+j-2$ &  $\theta_\sigma t^{m+j-2}$ \\ \vspace{0.2cm}
     \vdots        & \vdots \hspace{1.6cm}       &  \vdots & \vdots \\ \vspace{0.2cm}
$p_1$ &  $\stackrel{\cP_h(p_1)}{\text{\rule{1.5cm}{0.8pt}}}$ & $m$ &  $\theta_\sigma t^{m}$ \\ \vspace{0.1cm}
$p_0$ &  $\stackrel{\cP_h(p_0)}{\text{\rule{1cm}{0.8pt}}}$ &  \hspace{0.8cm} $d_0$ with $d_0<m$  & $\theta_0 t^{d_0}$          
	\end{tabular}
\end{center}
\item[(iii)] If $j \sigma + \tau=0$ for some $j \in \bN$,  then
there exist $q \in F[t]^\times$, $d \in \bN_0$, and $\theta \in \Theta$ with $\lc\left(\cP_h(q)\right) \notin \ker(\theta^*)$ such that
\[ \left(q, \cP_h(q), \theta t^{d} \right), \left(p_0, \cP_h(p_0), \theta_0 t^{d_0} \right),
\left(p_1, \cP_h(p_1), \theta_\sigma t^{m} \right),
\ldots, \left(p_{j-1}, \cP_h(p_{j-1}), \theta_\sigma t^{m+j-2} \right), \]
\[
\left(p_{j+1}, \cP_h(p_{j+1}), \theta_\sigma t^{m+j} \right),
\ldots,
\left(p_i, \cP_h(p_i), \theta_\sigma t^{m+i-1} \right), \ldots \]
is an echelon sequence of $I_h$ with respect to $(F[t], \cP_h)$ and $\bfOmega$, which is illustrated as
\begin{center}
	\begin{tabular}{crcc} 
preimages &  \hspace{0.5cm} basis elements \hspace{0.5cm} &  degrees of basis elements   &  pivots  \\ \vspace{0.1cm}
   \vdots         & \vdots \hspace{1.6cm}  & \vdots   &   \vdots \\ \vspace{0.2cm}
$p_i$ &   $\stackrel{\cP_h(p_i)}{\text{\rule{3.5cm}{0.8pt}}}$ &  $m+i-1$ & $\theta_\sigma t^{m+i-1}$ \\ \vspace{0.2cm}
     \vdots        & \vdots \hspace{1.6cm}      & \vdots & \vdots\\  \vspace{0.2cm}
$p_{j+1}$ &   $\stackrel{\cP_h(p_{j+1})}{\text{\rule{3cm}{0.8pt}}}$ & $m+j$  &$\theta_\sigma t^{m+j}$ \\ \vspace{0.1cm}
$p_{j-1}$ &   $\stackrel{\cP_h(p_{j-1})}{\text{\rule{2cm}{0.8pt}}}$ & $m+j-2$ & $\theta_\sigma t^{m+j-2}$ \\ \vspace{0.1cm}
    \vdots         & \vdots \hspace{1.6cm}       & \vdots & \vdots \\ \vspace{0.1cm}
$p_1$ &  $\stackrel{\cP_h(p_1)}{\text{\rule{1.5cm}{0.8pt}}}$ & $m$ & $\theta_\sigma t^{m}$ \\ \vspace{0.1cm}
$p_0$ &  $\stackrel{\cP_h(p_0)}{\text{\rule{1cm}{0.8pt}}}$ & \hspace{0.8cm} $d_0$ with $d_0<m$ &   $\theta_0 t^{d_0}$ \\ \vspace{0.1cm}
$q$ &   $\stackrel{\cP_h(q)}{\text{\rule{1.8cm}{0.8pt}}}$ &  \hspace{1.6cm} $d$ with $d<m+j-1$  & $\theta t^d$           
	\end{tabular}
\end{center}
\end{itemize}
\end{cor}
\begin{proof}
(i) 
The conclusion is immediate from Proposition \ref{PROP:image} (ii).

(ii) Since $j \sigma + \tau \neq 0$,  we see that $\cP_h(p_j)$ is of degree $m+j-1$ by Proposition \ref{PROP:image} (ii).
Note that $j$ is the unique integer such that $\theta_\sigma^*(j \sigma + \tau)=0$. So
$\theta_\sigma^*(i \sigma + \tau) \neq 0$ for all $i \in \bN$ with $i \neq j$.
Hence, the degree of $\cP_h(p_i)$ is $m+i-1$ for all $i \in \bN$ with $i \neq j$.
Let $\theta$ be an element of $\Theta$ such that the leading coefficient of $\cP_h(p_j)$ does not belong to $\ker(\theta^*)$. We see that 
(ii) holds.

(iii) 
Since $j \sigma+ \tau=0$, the degree of $\cP_h(p_j)$ is less than $m+j-1$ by Proposition \ref{PROP:image} (ii).
Let~$U$ be the subspace spanned by $\cP_h(p_0)$, $\cP_h(p_1),$ \ldots, $\cP_h(p_{j-1})$ over $C$. Then
\[  \left(p_0, \cP_h(p_0), \theta_0 t^{d_0}\right), \,
\left(p_1, \cP_h(p_1), \theta_\sigma t^{m} \right), \,
\ldots, \, \left(p_{j-1}, \cP_h(p_{j-1}), \theta_\sigma t^{m+j-2} \right) \]
is an echelon sequence of $U$ with respect to $(F[t], \, \cP_h)$ and $\Omega$.

Set $\bfomega_0 := \theta_0 t^{d_0}$, $\bfomega_1 := \theta_\sigma t^{m}$, \ldots, $\bfomega_{j-1} := \theta_\sigma t^{m+j-2}$.
It follows from Lemma \ref{LM:elim} that there exists 
a polynomial $r \in F[t]$ such that 
$\cP_h(p_j) - \cP_h(r) \in \cap_{k=0}^{j-1}\ker\left(\bfomega_k^*\right)$ and $\cP_h(r) \in U.$

Set $q := p_j-r$. Then $\cP_h(q)$ lies in both $I_h$ and $\cap_{k=0}^{j-1}\ker\left(\bfomega_k^*\right)$.
Hence,  
$$\cP_h(q), \cP_h(p_0), \cP_h(p_1), \ldots, \cP_h(p_{j-1}), \cP_h(p_{j+1}), \cP_h(p_{j+2}), \ldots$$
form a basis of $I_h$, because $\{ \cP_h(p_0), \cP_h(p_1), \ldots, \cP_h(p_{j-1}), \cP_h(p_j), \cP_h(p_{j+1}), \cP_h(p_{j+2}), \ldots \}$
is a basis of $I_h$.
Let $d$ be the degree of $\cP_h(q)$, and choose an element $\theta \in \Theta$ such that 
the leading coefficient of $\cP_h(q)$ does not belong to $\ker(\theta^*)$. We conclude that (iii) holds. 
\end{proof}

\subsubsection{Induced complements}
Let $E := \left\{  (f_i, \cP_h(f_i), \bfomega_i) \right\}_{i \in \bN}$ be an echelon sequence of $I_h$
given in Corollaries \ref{COR:m=0} or \ref{COR:pivot0}. It follows from  Lemma \ref{LM:modulo} that the subspace
\begin{equation} \label{EQ:complement}
U_h \cap \big( \mathop{\cap}\limits_{i \in \bN} \ker(\bfomega_i^*) \big)
\end{equation} 
is a complement of $\im(\cP_h)$ in $F[t]$, which is
called the {\em complement of $\im(\cP_h)$ induced by $E$}.

Every complement of $\im(\cP_h)$ in $F[t]$ is finite-dimensional over $C$ if $F=C$ 
by Example \ref{EX:cr-rational} and \cite[Lemma 8]{BCCLX2013}.
However, such a subspace is of  infinite dimension if $\dim_C(F)=\infty$.
\begin{example} \label{EX:prim}
Let $F=C(x)$, $t=\log(x)$ and $d/dx$ be the derivation on $F(t)$. Let
$$h = \dfrac{2x^2-2t}{xt^2+x},$$
which is $t$-normalized by a direct verification.

We construct an echelon sequence of $I_h$ with respect to $(F[t], \cP_h)$ and $\bfOmega$ in \eqref{EQ:eeff}. 
By  Convention \ref{CON:next2}, we have the following table:
\begin{center}
	\begin{tabular}{c|c|c|c|c|c|c}
$a$ &$b$ & $m$ & $a_m$ & $b_m$ & $a_{m-1}$ & $b_{m-1}$  \\ \hline
$- 2 x^{-1} t + 2 x $ & $t^2+1$ & $2$ & $0$ & $1$& $-2 x^{-1}$ & $0$  
	\end{tabular}
\end{center}
Since $a_m=0$, 
we may take the Hermite-Ostrogradsky reduction to be the complete reduction $\Phi_0$ for~$(F, \, ^\prime)$ 
in Hypothesis \ref{HYP:ind}.
The auxiliary subspace $U_h$ associated to~$(F[t], \cP_h)$ is equal to
\begin{equation} \label{EQ:pauxex}
F[t]_{<2} + \sum_{i \in \bN_0}  S_x \cdot  t^{i+2}
\end{equation} by 
Definition \ref{DEF:paux} and $\im(\Phi_0)=S_x$.

Note that $\cL$ in \eqref{EQ:lc} is equal to $d/dx$ by $a_2=0$ and $b_2=1$. Then we can set the type $\lambda$ of $I_h$ to be $1$. 
By Algorithm \ref{ALG:paux},  the first associated R-pair 
is $(\tilde{\sigma}, \sigma)=\left(0,x^{-1} \right)$  and 
the second R-pair 
is $(\tilde{\tau}, \tau)=\left(0, -2x^{-1} \right)$. 
Since $\sigma=x^{-1}$,  we choose $\theta_\sigma$ given before Corollary \ref{COR:m=0} to be $x^{-1}$. 
Corollary \ref{COR:pivot0} (iii) is needed to construct an echelon sequence $E$, because $2 \sigma+ \tau=0$.

We first compute $p_0$ and $p_1$ in the standard basis $\{p_i\}_{i \in \bN_0}$ of $\cP_h^{-1}(I_h)$.
Since the type $\lambda$ is equal to $1$, 
we have $p_0=1$, and then  $\cP_h(p_0) = -2x^{-1} t+2x$. Hence, $d_0=1$. Choose $\theta_0=x^{-1}$.
Then 
the second member of $E$ is
\begin{equation} \label{EQ:p0}
(p_0, \, \cP_h(p_0), \, \theta_0 t^{d_0}) = \left(1, \, -2x^{-1} t+2x, \,  x^{-1} t \right).
\end{equation}
Note that $\theta_0=\theta_\sigma$ is just a coincidence.
By \eqref{EQ:prel1} and $\tilde{\sigma}=\tilde{\tau}=0$, we see that $p_1 = t$. Thus, 
the third member of $E$ is 
\begin{equation} \label{EQ:p1}
(p_1, \, \cP_h(p_1), \, \theta_\sigma t^2) =\left(t, \, -x^{-1}t^2+2xt+x^{-1}, \, x^{-1} t^2 \right).
\end{equation}

Since $2 \sigma+ \tau=0$, we compute the first member of $E$ by $p_2$. 
By \eqref{EQ:prel1} and Algorithm \ref{ALG:paux}, we find  $p_2 = t^2 - x^2$. Then
$\cP_h(p_2) = (2x+2x^{-1})t-2x^3-2x.$
Eliminating $x^{-1}t$ from $\cP_h(p_2)$ by $\cP_h(p_0)$ yields  
$\left(q, \cP_h(q)\right) = \left(t^2-x^2+1, 2xt - 2x^3 \right).$
Choose $\theta=x$. Then the first member is 
\begin{equation} \label{EQ:q}
(q, \, \cP_h(q), \, \theta t ) =\left(t^2-x^2+1, \, 2xt - 2x^3, \, x t \right).
\end{equation}
For all integers $k>3$, the $k$th member of $E$ can be found by \eqref{EQ:prel1}, \eqref{EQ:prel2} and Algorithm \ref{ALG:paux}.
Its pivot is equal to $x^{-1} t^{k}$ by Corollary \ref{COR:pivot0} (iii). 
We summarize the results in the following table:

\vspace{-0.1cm}
\begin{center}
	\begin{tabular}{c|c|c|c|c|c|c|c} 
$\lambda$ &$(\tilde{\sigma}, \sigma)$ & $(\tilde{\tau}, \tau)$ & $\theta_\sigma$ & $j$ & $(q,  \cP_h(q),  \theta t )$ & $(p_0, \cP_h(p_0),  \theta_0 t^{d_0}) $ 
&  $(p_1,  \cP_h(p_1),  \theta_\sigma t^2)$  \\ \hline
$1$ & $\left(0,x^{-1}\right)$ & $\left(0, -2x^{-1}\right)$ & $x^{-1}$ & $2$ & r.h.s.\ of \eqref{EQ:q} & r.h.s.\ of \eqref{EQ:p0} & r.h.s.\ of \eqref{EQ:p1} 
	\end{tabular}
\end{center}
where $j$ stands for the positive integer with $j\sigma+ \tau=0$. 

By \eqref{EQ:complement} and \eqref{EQ:pauxex}, the complement $W$ of $\im(\cP_h)$ induced by $E$ is
equal to 
$$\big(C(x)[t]_{<2} + \sum_{i \in \bN_0}  S_x \cdot  t^{i+2} \big) \cap  \big( \mathop{\cap}\limits_{i \in \bN} \ker(\bfomega_i^*) \big),$$
where $\bfomega_1 = xt$, $\bfomega_2 = x^{-1}t$, $\bfomega_3 = x^{-1} t^2$ as listed in the table above,
and $\bfomega_k = x^{-1} t^{k}$ for $k>3$ following our conclusion about the pivot of the $k$th member in~$E$. Any element $r \in W$
can be written in  the form 
$$r = r_0 + r_1 t + r_2 t^2 + r_3 t^3 + r_4 t^4 + \cdots + r_k t^k$$
with $r_i \in C(x)$ satisfying the following conditions:
\begin{itemize}
\item in the irreducible partial fraction decomposition of~$r_1$, the polynomial part has no term of degree~$1$, and
the proper part contains no term of the form $c/x$ for any $c \in C^\times$;
\item $r_3$ is an arbitrary element of $S_x$;
\item $r_2, r_4, \ldots, r_k$ all belong to $S_x$ and have no pole at $x=0$.  
\end{itemize} 
\end{example}
We have made three choices for pivots in the above example.  Together with \eqref{EQ:prel1} and Algorithm \ref{ALG:paux},
these choices uniquely determine an echelon sequence of $I_h$. 
To systematically study such constructions, we introduce the notion of initial sequences in Definition \ref{DEF:iseq}
of the appendix. Every initial sequence is finite and 
carries all the information required for the 
unique construction of an echelon sequence. 
For instance, the final table in the above example encodes an initial sequence.
Moreover, Algorithm \ref{ALG:iseq} computes
initial sequences, and Algorithm \ref{ALG:pseq} generates arbitrary many members of the echelon sequence 
determined by a given initial sequence. 

With Remark \ref{RE:algcr}, we are ready to present the main result of this section.
\begin{thm} \label{TH:pcr}
Assume that Hypothesis \ref{HYP:ind} holds and $t$ is primitive. Then  
there exists an algorithm that, for every $t$-normalized element $h$ in $F(t)$,  
constructs a complete reduction $\widetilde{\Phi}_h$ for~$(F(t), \, \cR_h)$.
\end{thm}
\begin{proof}
Recall that $U_h$ is the auxiliary subspace  associated to $(F[t], \, \cP_h)$, and that $\lambda$ is the type of $I_h$.
If $\lambda=0$, then we can set $W := U_h$ by Proposition \ref{PROP:paux}.

Otherwise, let 
$E = \{(q_i, Q_i, \bfomega_i)\}_{i \in \bN}$
be the echelon sequence of~$I_h$ obtained from Corollary \ref{COR:m=0} or Corollary \ref{COR:pivot0}.
Set $W$ to be the complement of $\im(\cP_h)$ induced by $E$. Then $F[t] = \im(\cP_h) \oplus W$.
It follows from Proposition \ref{PROP:red} that
\begin{equation} \label{EQ:directh}
F(t) = \im(\cR_h) \oplus \left(b^{-1} \cdot W \oplus S_{t, h}\right),
\end{equation}
where $b=\den(h)$ by Convention \ref{CON:next2}, and $S_{t,h}$ is given in \eqref{EQ:subspaceh}.
Therefore, the projection from $F(t)$ to $\left(b^{-1} \cdot W \oplus S_{t, h}\right)$ with respect to \eqref{EQ:directh}
gives a complete reduction $\widetilde{\Phi}_h$ for $(F(t), \cR_h)$.

It remains to describe an algorithm for computing R-pairs with respect to $\widetilde{\Phi}_h$.
Let $\widetilde{\Psi}_h$ be the projection from $F[t]$ to $W$ with respect to $F[t] = \im(\cP_h) \oplus W$.
By Remark \ref{RE:red}, it suffices to develop an algorithm for computing R-pairs with respect to $\widetilde{\Psi}_h$.

Take any $p\in F[t]$. An auxiliary pair $(q, r)$ of $p$ can be found by Algorithm \ref{ALG:paux}.
If $\lambda=0$, then $(q,r)$ is an R-pair with respect to $\widetilde{\Psi}_h$.
Otherwise, one can compute $\tilde{q} \in F[t]$ and $\tilde{r} \in I_h$ such that
$\tilde{r} = \cP_h(\tilde{q})$ and $r-\tilde{r} \in \cap_{ i \in \bN} \ker(\bfomega_i^*)$ by Lemma \ref{LM:elim}.
On the other hand, $r-\tilde{r} \in U_h$ by $r, \tilde{r} \in U_h$.
Thus, $r-\tilde{r} \in W$ by \eqref{EQ:complement}.
Consequently,
$\left(q + \tilde{q},  \, r - \tilde{r}\right)$ is an R-pair of $p$ with respect to $\widetilde{\Psi}_h$.
\end{proof}

The next corollary 
relates the complete reduction $\widetilde{\Phi}_h$ for $(F(t),  \cR_h)$ from the preceding theorem to 
the complete reduction $\Phi_h$ for $(F,\cR_h)$ specified in Hypothesis \ref{HYP:ind} whenever $h\in F$.
It is useful 
for elementary integration in Section \ref{SECT:app}.

\begin{cor} \label{COR:prestrict}
Let $h \in F$, $\Phi_h$ be the complete reduction for $(F, \, \cR_h)$ in
Hypothesis~\ref{HYP:ind}, and~$\widetilde{\Phi}_h$ be the complete reduction for $(F(t), \, \cR_h)$ in 
Theorem~\ref{TH:pcr}. The following two assertions hold.
\begin{itemize}
\item[(i)] For all $f \in F$,
$\widetilde{\Phi}_h(f) = \Phi_h(f)$ if $I_h=\{0\}$, and, otherwise,  there exists a constant $c \in C$ such that $\widetilde{\Phi}_h(f)=\Phi_h(f)+c\Phi_h\left(\lambda t^\prime\right),$
where $\lambda$ is the type of $I_h$.
\item[(ii)] $\im(\widetilde{\Phi}_0 ) \subset \im(\Phi_0) \oplus \left( t \cdot F[t] \right) \oplus  S_t.$
\end{itemize}
\end{cor}
\begin{proof} 
 (i) 
Since $h \in F$, we see that $b=1$ by Convention \ref{CON:next2}. Then $S_{t, h}$ in \eqref{EQ:directh} is equal to $S_t$. 
Hence, equation \eqref{EQ:directh} becomes
\begin{equation} \label{EQ:direct0}
F(t) = \im(\cR_h) \oplus \left( W \oplus S_t\right).
\end{equation} 
Moreover, $\widetilde{\Phi}_h$ is the projection from $F(t)$ to $W \oplus S_t$ with respect to this direct sum.

For every $f\in F$, it follows from $h \in F$ and Corollary  \ref{COR:pinfield} 
that
\begin{equation} \label{EQ:Fdecomp}
f=\cR_h(\alpha)+\Phi_h(f)
\end{equation}
for some $\alpha \in F$, and that $(\alpha, \Phi_h(f))$ is an auxiliary pair of $f$ with respect to $(F[t], \, \cP_h)$.

If $I_h = \{0\}$, then $W$ in \eqref{EQ:direct0} is equal to $U_h$. Since $\cR_h(\alpha) \in \im(\cR_h)$ and $\Phi_h(f) \in U_h$,
we conclude by \eqref{EQ:Fdecomp} 
that $\widetilde{\Phi}_h(f) = \Phi_h(f)$. 

Assume that $I_h \neq \{0\}$. Then $W$ in \eqref{EQ:direct0} is the complement induced by the echelon sequence~$E$ of $I_h$ given by  Corollary \ref{COR:m=0}.
It follows from $\Phi_h(f) \in U_h$ and   \eqref{EQ:Fdecomp} 
that $\widetilde{\Phi}_h(f)$ is the projection of~$\Phi_h(f)$ on $W$ with respect to 
$U_h = I_h \oplus W$.  

Let $(\tilde{\sigma}, \sigma)$ be the first R-pair associated to $(F[t], \cP_h)$.
Then  $\sigma=\Phi_h( \lambda t^\prime)$ by Definition \ref{DEF:pairs}.
Note that the first member of $E$ is $( \lambda t-\tilde{\sigma}, \Phi_h( \lambda t^\prime), \theta_\sigma)$. 
Set $c = -\theta_\sigma^*\left( \Phi_h( \lambda t^\prime) \right)^{-1} \cdot \theta_\sigma^* \left(  \Phi_h(f) \right)$.
We find that $\Phi_h(f) + c \Phi_h( \lambda t^\prime)$ belongs to $W$, because other pivots in $E$ are of positive degree in~$t$. Therefore,
$\widetilde{\Phi}_h(f)=\Phi_h(f)+c\Phi_h\left(\lambda t^\prime\right).$ 

(ii) By Definition \ref{DEF:paux}, the auxiliary subspace associated to $(F[t], \cP_0)$ is 
$$  U_0 = \sum_{i \in \bN_0} \im(\Phi_0) \cdot t^i = \im(\Phi_0) +   \sum_{i \in \bN} \im(\Phi_0) \cdot t^i,$$ 
which is contained in $\im(\Phi_0) + \left(t \cdot F[t]\right)$. 
 Since $W \subset U_0$, we have $W \subset \im(\Phi_0) + \left(t \cdot F[t]\right)$. 
 Thus, $\im(\widetilde{\Phi}_0 ) \subset \im(\Phi_0) + \left(t \cdot F[t] \right)  + S_t$.
 The sum of the three subspaces is  evidently direct.  
\end{proof}

\section{The hyperexponential  case} \label{SECT:hyperexp}

This section is organized in the same way as in Section \ref{SECT:prim}.
In terms of content, there are two substantive differences. 
First, the definition of auxiliary subspaces
becomes more involved due to the presence of negative powers. Second, the construction of echelon sequences 
turns out to be considerably  simpler,
since the intersection of $\im(\cP_h)$ with the auxiliary subspace associated to $(F[t,t^{-1}], \, \cP_h)$ is of dimension at most two over $C$
(see Proposition \ref{PROP:hbasis}).

The results in this section generalize those in \cite[Chapter 4]{GaoThesis} for the exponential case and those in \cite[Section 4.2]{CDGL2025} for the rationally hyperexponential case (see \cite[Definition 4.2]{CDGL2025}).

We keep the assumptions and notation in 
Convention \ref{CON:next2}.  Then $t$ is hyperexponential and~$F\langle  t \rangle=F[t, t^{-1}]$.
The tail coefficients $a_0$ and $b_0$ will play a significant role in reducing negative powers of $t$.

To describe the head and tail terms of an element in $\im(\cP_h)$,
we define two $C$-linear operators as follows.
For every $d \in \bZ$, let
\begin{equation} \label{EQ:coeff}
 \begin{array}{cccc}
     \cH_d: & F & \rightarrow & F \\
              & \alpha &  \mapsto & b_m  \alpha^\prime+\left(a_m + d  b_m   \dfrac{ t^\prime}{t} \right) \alpha 
   \end{array}  
\, \text{and} \,
   \begin{array}{cccc}
     \cT_d: & F & \rightarrow & F \\
              & \alpha &  \mapsto & b_0 \alpha^\prime+\left(a_0 + d  b_0   \dfrac{ t^\prime}{t} \right) \alpha.
   \end{array}
\end{equation}
For every $f \in F[t, t^{-1}]^\times $ with $k=\hdeg(f),$ $f_k = \hc(f)$, $l=\tdeg(f)$ and $f_l = \tc(f)$, we have 
\begin{equation} \label{EQ:hred}
\cP_h( f ) =   \cH_k(f_k) t^{m+k} + \cdots + \cT_l(f_l)t^{l}
\end{equation}
by Lemma \ref{LM:fact} (ii).
It follows that $\cP_h(f^+) \in F[t]$ and $\cP_h(f^-) \in F[t]_{<m} + t^{-1} F[t^{-1}]$.

\begin{remark} \label{RE:hred}
By Convention \ref{CON:next2} and \eqref{EQ:coeff}, the following assertions hold for all $k,l \in \bZ$.
\begin{itemize}
\item[(i)] If $\nu_\infty(h)<0$, then $a_m \neq 0$, $b_m = 0$ and $\ker(\cH_k) = \{0\}$.
\item[(ii)] If $\nu_\infty(h) \ge 0$, then $b_m=1$ and $\cH_k = \cR_{\lambda_k}$, where
$\lambda_k = a_m + k \dfrac{t^\prime}{t}.$
 \item[(iii)] If $\nu_t(h)<0$, then $a_0 \neq 0$, $b_0 = 0$, and $\ker(\cT_l) = \{0\}$.
\item[(iv)] If $\nu_t(h) \ge 0$, then $b_0 \neq 0$ and $\cT_l = b_0 \cR_{\mu_l}$, where
$\mu_l = \dfrac{a_0}{b_0} + l \dfrac{t^\prime}{t}.$
\end{itemize}
The $\lambda_k$'s and $\mu_l$'s defined above will be used throughout  the rest of this section.
\end{remark}
In this section, we will mainly prove conclusions concerning tail degrees. 
For the head-degree case, the reasoning is analogous to the arguments developed in the previous section.
Actually,  the head-degree case is easier to handle than the tail-degree case because either $b_m=1$ or $b_m=0$.

\subsection{Auxiliary subspaces} \label{SUBSECT:haux}

As in Section \ref{SUBSECT:paux}, we aim to reduce all coefficients of elements in $F[t, t^{-1}]$ modulo $\im(\cP_h)$.

Let $f \in F[t,t^{-1}]$ with $k=\hdeg(f)$, $f_k=\hc(f)$, $l=\tdeg(f)$ and $f_l = \tc(f)$. Assume further that $k \ge m$ and $l<0$.
By arguments analogous to those at the beginning of Section \ref{SUBSECT:paux}, we obtain the following four congruences.
\begin{equation} \label{EQ:haux_case+}
f^+  \equiv \begin{cases}
 \cP_h\left(a_m^{-1} f_k t^{k-m}\right) \mod F[t]_{<k}  & \text{if $\nu_\infty(h)<0$},  \\ \\
\cP_h(\alpha t^{k-m}) + \beta t^k \mod F[t]_{<k} &  \text{if $\nu_\infty(h) \ge 0$}, 
\end{cases}
\end{equation}
where $(\alpha, \beta) \in F^2$ is an R-pair of $f_k$ with respect to $\Phi_{\lambda_{k-m}}$ given in Hypothesis \ref{HYP:ind}, 
and 
\begin{equation} \label{EQ:haux_case-}
f^-  \equiv  
\begin{cases}
\cP_h\left(a_0^{-1} f_l t^{l}\right)  \mod   F[t]_{<m} + \left(t^{-1} \cdot F[t^{-1}]\right)_{>l}  & \text{if $\nu_t(h)<0$}, \\ \\
\cP_h( \alpha t^l) + (b_0 \beta)  t^l   \mod F[t]_{<m} +  \left(t^{-1} \cdot F[t^{-1}]\right)_{>l} & \text{if $\nu_t(h) \ge 0$}, 
\end{cases}
\end{equation}
where $(\alpha, \beta) \in F^2$ is an R-pair of $b_0^{-1} f_l$ with respect to $\Phi_{\mu_l}$ given in Hypothesis \ref{HYP:ind}. 
Note that $\nu_t(h) < 0$ implies $a_0 \neq 0$ by Remark \ref{RE:hred} (iii) and  $\nu_t(h) \ge  0$ implies $b_0 \neq 0$ by Remark \ref{RE:hred} (iv).

The two congruences in \eqref{EQ:haux_case+} can be verified by the same reasoning used to establish \eqref{EQ:paux_case}. The first congruence in \eqref{EQ:haux_case-} holds
by \eqref{EQ:hred} and Remark \ref{RE:hred} (iii). To show the second congruence in \eqref{EQ:haux_case-}, 
we note that $b_0^{-1} f_l = \cR_{\mu_l}(\alpha)+ \beta$, which, together with $\cT_l= b_0 \cR_{\mu_l}$ given in Remark \ref{RE:hred} (iv),
implies that $f_l = \cT_l(\alpha) + b_0 \beta$.
So 
$f^- \equiv  \cT_l(\alpha)  t^l  + \left(b_0 \beta \right)  t^l \mod F[t]_{<m}+ \left(t^{-1} \cdot F[t^{-1}]\right)_{>l}$. 
The congruence follows from \eqref{EQ:hred}. 

The above deduction motivates the notion of auxiliary subspaces below.
\begin{define} \label{DEF:haux}
The auxiliary subspace associated to $(F[t,t^{-1}], \cP_h)$ is defined as 
$$V_h = V_h^{+} + F[t]_{<m} + V_h^-,$$
where
\[
V_h^{+}=
\left\{ \begin{array}{ll}
\{0\} & \text{if $\nu_\infty(h)<0$}, \\ \\
\displaystyle \sum_{i \ge m} \im  \left(\Phi_{\lambda_{i-m}}\right) \cdot t^i & \text{if $\nu_\infty(h) \ge 0$,}
\end{array} \right.\, {\text and}~V_h^{-} = \left\{
\begin{array}{ll}
\{0\} & \text{ if $\nu_t(h)<0$}, \\ \\
  \displaystyle \sum_{j < 0}   b_0 \cdot \im  \left(\Phi_{\mu_{j}}\right) \cdot  t^j
 & \text{if $\nu_t(h) \ge 0$}.
\end{array} \right.
\]
\end{define}

\begin{example} \label{EX:haux}
Let $F=C$ and $t^\prime/t=1$, that is, $t$ models $\exp(x)$. Let $h=0$. Convention \ref{CON:next2} and Remark \ref{RE:hred} lead to the following table:
\begin{center}
	\begin{tabular}{c|c|c|c|c|c|c} 
		 $m$ & $a_m$ & $a_0$ & $b_m$ & $b_0$ & $\lambda_k$, $k \ge 0$  &  $\mu_l$, $l<0$ \\ \hline
                $0$       &  $0$ & $0$ & $1$  & $1$ & $k$ & $l$ 
	\end{tabular}
\end{center}
From Example \ref{EX:choice0}, we obtain $V_h^{+}=C$ and $V_h^{-}=\{0\}$.
So $C$ is the auxiliary subspace~$V_h$ associated to~$(F[t,t^{-1}], \cP_h)$.
\end{example}
Analogous to Proposition \ref{PROP:paux} and Definition \ref{DEF:ppair}, we have
\begin{prop} \label{PROP:haux}
Let $f \in  F[t,t^{-1}]$,
$k=\hdeg(f)$ and $l = \tdeg(f)$. Then
\begin{itemize}
\item[(i)] $f^+ = \cP_h(p) + r$ for some $p \in F[t]$ with $\hdeg(p) \le k-m$ and $r \in V_h$,
\item[(ii)]  $f^- = \cP_h(q) + s$ for some $q \in t^{-1} \cdot F[t^{-1}]$ with $\tdeg(q) \ge l$ and $s \in V_h$.
\end{itemize}
Consequently, $F[t, t^{-1}] = \im(\cP_h) + V_h$.
\end{prop}
\begin{proof}
(i) The proof is analogous to that of Proposition \ref{PROP:paux} (i).

(ii) If $f^-=0$, then $l \ge 0$.  It suffices to set $q:=0$ and $s:=0$. 
Assume that $l<0$ and the conclusion holds for all elements of  $\left(t^{-1} \cdot F[t^{-1}]\right)_{>l}$.

If $\nu_t(h)<0$, then  $f^- \equiv \cP_h\left(a_0^{-1} f_l t^{l}\right) + g \mod   F[t]_{<m}$ for some $g \in t^{-1} F[t^{-1}]_{>l}$
by the first congruence in \eqref{EQ:haux_case-}. The conclusion holds by an application of the induction hypothesis to $g$.

Otherwise, $\nu_t(h) \ge 0$.
Let $\left(\alpha, \beta \right)$ be an R-pair of $b_0^{-1} f_l$ with respect to $\Phi_{\mu_{l}}$.
It follows from the second congruence in \eqref{EQ:haux_case-}, $F[t]_{<m} \subset V_h$ and $(b_0 \beta) \cdot  t^l \in V_h^-$ 
that $f^- \equiv \cP_h(\alpha t^{l}) +g \mod  V_h$
for some $g  \in \left( t^{-1} \cdot F[t^{-1}]\right)_{>l}$.
The conclusion holds by the induction hypothesis.

By (i) and (ii), $f = \cP_h(p+q)+r+s$ and $r+s \in V_h$. 
Hence, $F[t, t^{-1}] = \im(\cP_h) + V_h$.
\end{proof}
\begin{define} \label{DEF:hpair}
Let $f, p, q, r$ and $s$ be as in the above proposition. We call $(p+q,r+s)$ an {\em auxiliary pair} of $f$ with respect to $(F[t,t^{-1}], \cP_h)$ or
an {\em auxiliary pair} of $f$ if $(F[t,t^{-1}], \cP_h)$ is clear from context.
\end{define}
Based on the congruences in \eqref{EQ:haux_case+}, \eqref{EQ:haux_case-}, 
and the proof of Proposition \ref{PROP:haux}, an algorithm is developed
for computing auxiliary pairs (see Algorithm \ref{ALG:haux} in the appendix).
\begin{cor} \label{COR:hinfield}
Let $h \in F$, and $\Phi_h$ be the complete reduction for $(F, \, \cR_h)$ in Hypothesis \ref{HYP:ind}. 
Then, for every element $f \in F$, an auxiliary pair of $f$ with respect to $\left(F[t,t^{-1}], \, \cP_h\right)$ is an R-pair of $f$
with respect to $\Phi_h$.
\end{cor}
\begin{proof} Note that $\lambda_0$ defined in Remark \ref{RE:hred} (ii) is equal to $h$ under the assumption $h\in F$.
The rest is similar to the proof of Corollary \ref{COR:pinfield}.
\end{proof}
\subsection{Intersecting $\im(\cP_h)$ with the associated auxiliary subspace} \label{SUBSECT:hinter}

In this subsection, $V_h$ stands for the auxiliary subspace associated to $(F[t,t^{-1}], \cP_h)$, and $J_h$ denotes 
$\im(\cP_h) \cap V_h$.
We will derive a necessary and sufficient condition on $J_h = \{0\}$.

The next two lemmas connect $J_h$ with the kernels of $\cH_d$ and $\cT_d$ in \eqref{EQ:coeff}.
We only prove the second assertion in each lemma, and the first assertion can be shown likewise.
\begin{lemma} \label{LM:nonzero2}
Let $f \in F[t,t^{-1}]^\times$, $k=\hdeg(f)$ and $l = \tdeg(f)$. Assume that $\cP_h(f) \in V_h$. Then
\begin{itemize}
\item[(i)] $\hc(f) \in \ker(\cH_k)$ if $k \ge 0$, and
\item[(ii)] $\tc(f) \in \ker(\cT_l)$  if $l < 0$.
\end{itemize}
\end{lemma}
\begin{proof}
(ii) Assume that $l<0$ and let $f_l = \tc(f)$. Suppose that $\cT_l(f_l) \neq 0.$ Then
$\cT_l(f_l) t^l \in V_h^-$ by~$\cP_h(f) \in V_h$ and $l<0$.
In particular, $V_h^{-} \neq \{0\}$. So $\nu_t(h) \ge 0$ and $ \cT_l(f_l) \in b_0 \cdot \im(\Phi_{\mu_l})$ by Definition~\ref{DEF:haux}.
By Remark \ref{RE:hred} (iv), $\cT_l = b_0 \cR_{\mu_l}$. Thus,  
$\cR_{\mu_l} (f_l) \in  \im(\Phi_{\mu_l})$.
Since $\im(\cR_{\mu_l}) \cap \im(\Phi_{\mu_l})$ is equal to $\{0\},$  we derive that $\cT_{l}(f_l)=0$, a contradiction.
\end{proof}
\begin{lemma} \label{LM:nonzero1}
\begin{itemize}
\item[(i)] If $\ker(\cH_k) \neq \{0\}$ for some $k \in \bN_0$, then for all $\sigma \in \ker(\cH_k)^\times$, there exists $p \in F[t]$
with $\hdeg(p)=k$ and $\hc(p)=\sigma$  such that
$\cP_h(p) \in V_h.$
Moreover, for every $f \in F[t, t^{-1}]$
with $\cP_h(f) \in V_h$, there exists a  constant $c \in C$ such that $f - c p \in t^{-1} \cdot F[t^{-1} ]$. 
\item[(ii)] If $\ker(\cT_l) \neq \{0\}$ for some $l \in \bN^-$, then for all $\tau \in \ker(\cT_l)^\times$, there exists $q \in t^{-1} \cdot F[t^{-1}]$
with $\tdeg(q)=l$ and $\tc(q)=\tau$ such that
$\cP_h(q) \in V_h.$
Moreover, for every $f \in F[t, t^{-1}]$
with $\cP_h(f) \in V_h$, there exists a constant $c \in C$ such that $f - c q \in F[t]$.
\end{itemize}
\end{lemma}
\begin{proof}
%

(ii) By Remark \ref{RE:hred} (iii), $\nu_t(h) \ge 0$.   
By \eqref{EQ:hred} and $\tau \in \ker(\cT_l)$, the image $\cP_h(\tau t^l)$ belongs to~$F[t]_{<m} + \left(t^{-1} \cdot F[t^{-1}]\right)_{>l}$. 
Let $(g,r)$ be an auxiliary pair of $\cP_h\left(\tau t^l\right)$.
Then 
$$\cP_h\left(\tau t^l \right) = \cP_h(g) +r, \quad g \in \left(t^{-1} \cdot F[t^{-1}]\right)_{>l}, \quad \text{and} \quad r \in V_h$$
by Proposition \ref{PROP:haux} (ii). Set $q :=\tau t^l - g$. Then
$\tdeg(q)=l,$ $\tc(q) =\tau $ and $\cP_h(q) \in V_h.$

Assume that $f \in F[t,t^{-1}]$ with $d = \tdeg(f)$ and $\cP_h(f) \in V_h$. If $d \ge  0$, then we set $c:=0$.
Otherwise, $\tc(f) \in \ker(\cT_d)$ by Lemma \ref{LM:nonzero2} (ii). It follows from \cite[Lemma 6.2.2]{BronsteinBook} that $d=l$.
Since $\dim_C(\ker(\cT_l))=1$, there exists a constant $c \in C$ such that $\tc(f)=c \tau$. Then $\cP_h(f - c q) \in V_h$ and $\tdeg(f - c q) > l$.
So $\tdeg(f - c q) \ge 0$ by Lemma \ref{LM:nonzero2} (ii) and \cite[Lemma 6.2.2]{BronsteinBook}.
\end{proof}

We now present a criterion for $J_h=\{0\}$.
\begin{prop} \label{PROP:hinter}
$J_h = \{0\}$ if and only if one of the following two conditions holds.
\begin{itemize}
\item[(i)] $h \in F$,
\item[(ii)] $\ker(\cH_k)=\{0\}$ for all $k \in \bN_0$ and $\ker(\cT_l)=\{0\}$ for all $l \in \bN^-$.
\end{itemize}
In particular, $J_h=\{0\}$ if $\nu_{\infty}(h)<0$ and $\nu_t(h)<0$.
\end{prop}
\begin{proof}
Assume that $J_h=\{0\}$ and $h \notin F$. Suppose that $\ker(\cH_k)$ is nontrivial for some $k \in \bN_0$.
Then there exists a polynomial $p \in F[t]^\times$
such that $\cP_h(p) \in V_h$ by Lemma \ref{LM:nonzero1} (i).
Moreover, $\cP_h(p) \neq 0$ by Lemma \ref{LM:inj}. So $J_h$ is nontrivial, a contradiction.
The same contradiction is reached by  Lemma \ref{LM:nonzero1} (ii) and Lemma \ref{LM:inj} if  $\ker(\cT_l)$ is nontrivial for some $l \in \bN^-$.

To show the converse, we let $f  \in F[t,t^{-1}]$ with $\cP_h(f) \in V_h$.
It suffices to show that $\cP_h(f) = 0$.

First, we consider the case in which $h \in F$. Then $m=0$, $b_0=1$,  $\nu_{\infty}(h) \ge 0$ and $\nu_{t}(h) \ge 0$.
It follows from Definition \ref{DEF:haux} that
$$V_h = \sum_{i \in \bN_0} \im(\Phi_{\lambda_i}) \cdot t^i + \sum_{j \in \bN^-}  \im(\Phi_{\mu_j}) \cdot t^j,$$
where $\lambda_i = it^\prime/t+a_0$ and $\mu_j = j t^\prime/t + a_0$ by Remark \ref{RE:hred} (ii) and (iv).

Write $f = \sum_{i=l}^k f_i t^i$ with $f_i \in F$. 
Since $\cP_h(f)=f^\prime+a_0f$, we have 
$$\cP_h(f) = \sum_{i=0}^k \cR_{\lambda_i}(f_i) t^i + \sum_{j=l}^{-1} \cR_{\mu_j}(f_j) t^j$$
by Lemma \ref{LM:fact} (ii). 
Hence, $\cR_{\lambda_i}(f_i) \in \im(\Phi_{\lambda_i})$ for all $i \in \bN_0$
and $\cR_{\mu_j}(f_j) \in \im(\Phi_{\mu_j})$ for all $j \in \bN^{-}$.
Thus, $\cP_h(f)=0$ by $\im(\cR_{\lambda_i}) \cap  \im(\Phi_{\lambda_i})= \{0\}$ and $\im(\cR_{\mu_j}) \cap  \im(\Phi_{\mu_j}) = \{0\}$.

Next, we assume that $h \notin F$, $\ker(\cH_k)=\{0\}$ for all $k \in \bN_0$ and $\ker(\cT_l)=\{0\}$ for all $l \in \bN^-$.
Suppose that $\cP_h(f) \neq 0$. If $k=\hdeg(f) \ge 0$, then $f_k \in \ker(\cH_k)^\times$ by Lemma \ref{LM:nonzero2} (i),
a contradiction. So $f^+ = 0$. Thus, we may further assume that $\tdeg(f)=l<0$.
It follows from Lemma \ref{LM:nonzero2} (ii) that $f_l \in \ker(\cT_l)^\times$. We have reached   another contradiction.

In particular, if $\nu_{\infty}(h)<0$ and $\nu_t(h)<0$, then both $\cH_k$ and $\cT_l$ are injective by Remark \ref{RE:hred} (i) and~(iii), respectively.
So $J_h = \{0\}$ by the second condition.
\end{proof}
 We introduce the notion of types to characterize $J_h$.
\begin{define} \label{DEF:htype}
The {\em type} of $J_h$ is defined to be
\begin{itemize}
\item[(i)] $0$ if either $h \in F$, or  $\ker(\cH_k)=\{0\}$ for all $k \in \bN_0$ and $\ker(\cT_l)=\{0\}$ for all $l \in \bN^-$,
\item[(ii)] $(k, \sigma)$ if $h \notin F$, $\sigma  \in \ker(\cH_k)^\times$ for some $k \in \bN_0$ and
$\ker(\cT_l)=\{0\}$ for all $l \in \bN^-$,
\item[(iii)] $(l,\tau)$ if $h \notin F$, $\ker(\cH_k) = \{0\}$ for all $k \in \bN_0$ and
$\tau \in \ker(\cT_l)^\times$ for some $l \in \bN^-$,
\item[(iv)] $(k,\sigma),(l,\tau)$ if $h \notin F$, $\sigma  \in \ker(\cH_k)^\times$ for some $k \in \bN_0$ and
$\tau \in \ker(\cT_l)^\times$ for some $l \in \bN^-$.
\end{itemize}
\end{define}
The type of $J_h$ can be computed by an algorithm for solving the parametric logarithmic derivative problem in $F$. 
The type is zero if and only  if $J_h = \{0\}$ by Proposition \ref{PROP:hinter}.
If it is nonzero, then $k, l$ are unique by \cite[Lemma 6.2.2]{BronsteinBook}, and $\sigma, \tau$ are unique up to
multiplication by a  constant in $C^\times$. The type
is subsequently fixed once and for all. 
\begin{example} \label{EX:hindex}
Let $F=C$, $t=\exp(x)$ and $h=-t^2/(t^2+t+1)$. Convention \ref{CON:next2} and Remark~\ref{RE:hred} lead to the following table:
\begin{center}
	\begin{tabular}{c|c|c|c|c|c|c} 
		 $m$ & $a_m$ & $a_0$ & $b_m$ & $b_0$ & $\cH_k(z)$, $k \ge 0$  &  $\cT_l(z)$, $l<0$ \\ \hline
        $2$  &  $-1$ & $0$ & $1$  & $1$ & $(k-1)z$ & $lz$ 
	\end{tabular}
\end{center}
Hence, the type of $J_h$ is equal to $(1,1)$.
\end{example}
\subsection{Echelon sequences and induced complements}

In contrast to the primitive setting, we have $\dim_C(J_h) \le 2$, as stated in the following proposition.
\begin{prop} \label{PROP:hbasis}
\begin{itemize}
\item[(i)] $\dim_C(J_h)=1$ if the type of $J_h$ is equal to $(i,\kappa)$ for some $i \in \bZ$.
\item[(ii)] $\dim_C(J_h)=2$ if the type of $J_h$ is equal to $(k,\sigma), (l,\tau)$ for some $k \in \bN_0$ and $l \in \bN^-$.
\end{itemize}
\end{prop}
\begin{proof} Since the type of $J_h$ is nonzero in  both (i) and (ii), we see that $h \in F(t) \setminus F$ by
Proposition~\ref{PROP:hinter}. Consequently, $\cP_h$ is injective by Lemma \ref{LM:inj}.

(i) Let us consider the case in which $i<0$.
By Lemma \ref{LM:nonzero1} (ii), there exists an element $q \in t^{-1} \cdot F[t^{-1}]$ with $\tdeg(q)=i$ such that $\cP_h(q) \in V_h$ and
for every $f \in F[t,t^{-1}]$ with $\cP_h(f) \in V_h$,  $g := f - c q \in F[t]$ for some $c \in C$.
Then $\cP_h(g) \in V_h$ by  $\cP_h(f),  \cP_h(q)\in V_h$.
Suppose that $g \neq 0$. Let $k = \hdeg(g)$. Then $k \ge 0$ by $g \in F[t]$. It follows from Lemma \ref{LM:nonzero2} (i) that
$\ker(\cH_k) \neq \{0\}$, a contradiction to the type of~$J_h$.
So $g=0$. Consequently, $\cP_h(f) = c \cP_h(q)$. 
Note that $\cP_h(q) \neq 0$ by the injectivity of $\cP_h$.    So $\dim_C(J_h)=1$.

Likewise, one shows that $\dim_C(J_h)=1$ if $i \ge 0$.

(ii) By Lemma \ref{LM:nonzero1}, there exist two elements $p \in F[t]$ with $\hdeg(p)=k$
and $q \in t^{-1} \cdot F[t^{-1}]$ with $\tdeg(q)=l$ such that $\cP_h(p), \cP_h(q) \in V_h$. 
Let $f \in F[t,t^{-1}]$ with $\cP_h(f) \in V_h$. 
Then $g:=f - c_p p \in t^{-1} \cdot F[t^{-1}]$ for some $c_p \in C$ by Lemma \ref{LM:nonzero1} (i).
Moreover, $\cP_h(g) \in V_h$ by $ \cP_h(f), \cP_h(p) \in V_h$.
It follows from  Lemma \ref{LM:nonzero1} (ii) that $g - c_q q \in F[t]$ for some $c_q \in C$.
Then $g - c_q q =0$ because both  $g$ and $q$ belong to $t^{-1} \cdot F[t^{-1}]$. Thus, 
$f = c_p p + c_q q$. Consequently, $\cP_h(f) = c_p \cP_h(p) + c_q \cP_h(q)$.
Since $p \in F[t]$ and $q \in t^{-1}F[t^{-1}]$, 
$\cP_h(p)$ and $\cP_h(q)$ are $C$-linearly independent by the injectivity of $\cP_h$. We conclude that $\dim_C(J_h)=2$.
\end{proof}

As in Section \ref{SECT:prim}, we let $\Theta$ be an effective basis of $F$ and set $\bfOmega := \{\theta t^i \mid \theta \in \Theta, i \in \bZ\}$,
which is an effective  basis of $F[t, t^{-1}]$. We are going to construct an echelon sequence of $J_h$.
For the sake of completeness,
the echelon sequence is set to be $\emptyset$ if the type of $J_h$ is equal to $0$. 
\begin{prop} \label{PROP:hechelon}
Assume that $J_h \neq \{0\}$.
\begin{itemize}
\item[(i)] If the type of $J_h$ is equal to $(i,\kappa)$ for some $i \in \bZ$, then
$J_h$ has an echelon sequence: $(p, \cP_h(p), \bfomega)$, and $V_h \cap \ker(\bfomega^*)$ is a complement of $\im(\cP_h)$ in $F[t,t^{-1}]$.
\item[(ii)] If the type of $J_h$ is equal to $(k,\sigma), (l,\tau)$ for some $k \in \bN_0$ and $l \in \bN^-$,
then $J_h$ has an echelon sequence:
$\left(p_1, \cP_h(p_1), \bfomega_1 \right),  \, \left(p_2, \cP_h(p_2), \bfomega_2\right),$ 
and $V_h \cap \ker(\bfomega_1^*) \cap \ker(\bfomega_2^*)$ is a complement of $\im(\cP_h)$ in $F[t, t^{-1}]$.  
\end{itemize}
\end{prop}
\begin{proof}
(i) By Proposition \ref{PROP:hbasis} (i), we have $\dim_C(J_h)=1$. Then $J_h = \spa_C\{\cP_h(p)\}$ for some $p \in F[t, t^{-1}]$.
Let $d$ and~$\gamma$ be the head degree and head coefficient of $\cP_h(p)$, respectively.
We choose an element $\theta \in \Theta$ such that $\theta^*(\gamma) \neq 0$.
Then $\left(p, \cP_h(p), \theta t^d\right)$ is an echelon sequence of $J_h$. 
It follows from Lemma \ref{LM:modulo} that $V_h \cap \ker(\bfomega^*)$ is a complement of $\im(\cP_h)$ in $F[t,t^{-1}]$.

(ii) By Proposition \ref{PROP:hbasis} (ii), we have $\dim_C(J_h)=2$. 
Then $J_h = \spa_C\{\cP_h(p), \cP_h(q)\}$ for some $p, q \in F[t, t^{-1}]$.
Let $d$ and $\gamma$ be the head degree and head coefficient of $\cP_h(p)$, respectively,
and let $\theta_p \in \Theta$ with $\theta_p^*(\gamma) \neq 0$.
Set $\bfomega_p := \theta_p t^{d}$, $c_p :=  \bfomega_p^* \left(\cP_h(p)\right)$, $c_q := \bfomega_p^*  \left( \cP_h(q) \right)$, and
$r := q - c_p^{-1} c_q p.$
Then $\cP_h(r)$ belongs to $\ker(\bfomega_p^*)$ and $\{\cP_h(r), \, \cP_h(p)\}$ is also a basis of $J_h$.
Let $e$ and  $\rho$ be the head degree and head coefficient of $ \cP_h(r)$, respectively, and choose $\theta_r \in \Theta$ with $\theta_r^*(\rho) \neq 0$.
Set $\bfomega_r = \theta_r t^{e}$.
Then $\left(r, \cP_h(r), \bfomega_r \right)$, $\left(p, \cP_h(p), \bfomega_p\right)$ is an echelon sequence of $J_h$.
By Lemma \ref{LM:modulo},  $V_h \cap \ker(\bfomega_r^*) \cap \ker(\bfomega_p^*)$ is a complement of $\im(\cP_h)$ in $F[t, t^{-1}]$.
The proof is completed by setting $p_1:=r, \bfomega_1:=\bfomega_r, p_2 :=p$ and $\bfomega_2:=\bfomega_p$. 
\end{proof}

An algorithm for constructing an echelon sequence of $J_h$
can be developed according to the proofs of Propositions \ref{PROP:hbasis} and \ref{PROP:hechelon} (see Algorithm \ref{ALG:hseq} in the appendix). 
\begin{define} \label{DEF:ecomp}
Let $E$ be the echelon sequence in Proposition \ref{PROP:hechelon} (i) or (ii).
The complement of $\im(\cP_h)$ in $F[t, t^{-1}]$ in the same proposition is called the {\em complement  induced by $E$}.
\end{define}

In analogy to Theorem \ref{TH:pcr} and Corollary \ref{COR:prestrict}, we have the following.
\begin{thm} \label{TH:hcr}
Assume that Hypothesis \ref{HYP:ind} holds and  $t$ is hyperexponential. Then 
there exists an algorithm that, for every $t$-normalized element $h$ of $F(t)$,  constructs a complete reduction~$\widetilde{\Phi}_h$  for $(F(t), \, \cR_h)$.
\end{thm}
\begin{proof}
It suffices to construct a complement $W$ of $\im(\cP_h)$ in $F[t, t^{-1}]$ by Proposition \ref{PROP:red}.

Recall that $V_h$ denotes the auxiliary subspace  associated to $(F[t, t^{-1}], \, \cP_h)$,  and that $J_h$ denotes $\im(\cP_h) \cap V_h$.
If the type of $J_h$ is equal to zero, 
then $J_h = \{0\}$ and we simply set $W := V_h$ by Proposition \ref{PROP:haux}.
Assume that the type of $J_h$ is equal to $(i,\kappa)$ for some $i \in \bZ$.
By Proposition \ref{PROP:hechelon} (i), $J_h$
has an echelon sequence $(p, \cP_h(p), \bfomega)$, and we take $W := V_h \cap \ker(\bfomega^*)$ by Lemma \ref{LM:modulo}.  
Otherwise, $J_h$ has an echelon sequence:
$\left(p_1, \cP_h(p_1), \bfomega_1 \right),  \, \left(p_2, \cP_h(p_2), \bfomega_2\right)$ by Proposition \ref{PROP:hechelon} (ii)
and we set $W := V_h \cap \ker(\bfomega_1^*) \cap \ker(\bfomega_2^*)$ by Lemma \ref{LM:modulo}.

An algorithm for computing R-pairs with respect to $\widetilde{\Phi}_h$ can be developed 
along the same lines as the construction outlined in the last two paragraphs in the proof of Theorem \ref{TH:pcr}.
\end{proof}
\begin{cor} \label{COR:hrestrict}
Let $h \in F$, $\widetilde{\Phi}_h$ be the complete reduction for $(F(t), \cR_h)$ in Theorem \ref{TH:hcr},
and~$\Phi_h$ be the complete reduction for $(F, \cR_h)$ in Hypothesis \ref{HYP:ind}.
Then 
\begin{itemize}
\item[(i)] $\widetilde{\Phi}_h(f) = \Phi_h(f)$ for all $f \in F$, and 
\item[(ii)] $\im(\widetilde{\Phi}_0) \subset\im(\Phi_0) \oplus  \left(t^{-1} \cdot F[t^{-1}]\right) \oplus  
\left(t \cdot F[t] \right) \oplus  S_t.$
\end{itemize}
\end{cor}
\begin{proof}
(i) By Proposition \ref{PROP:hinter} and $h \in F$, we have $J_h=\{0\}$. 
The corollary follows from 
Corollary \ref{COR:hinfield}, via the same reasoning used in the proof of Corollary \ref{COR:prestrict} (i) to verify
$\widetilde{\Phi}_h(f) = \Phi_h(f)$ for all $f \in F$ when $I_h = \{0\}$.

(ii) The auxiliary subspace associated to $(F[t, t^{-1}], \cP_0)$ is 
$$\sum_{k \in \bN} \im\left(\Phi_{\lambda_k}\right) \cdot t^k +  \im\left(\Phi_{ \lambda_0}\right) + \sum_{l \in \bN^-} \im \left(\Phi_{ \mu_l}\right) \cdot t^l$$ 
by Convention \ref{CON:next2} and Definition \ref{DEF:haux}. Moreover, $\lambda_0=0$ by $h=0$.
The rest of the proof proceeds in a fashion analogous to that of Corollary \ref{COR:prestrict} (ii).
\end{proof}

\section{Completing the induction} \label{SECT:main}

In this section, we let $F_0=C$ and $F_n=C(t_1, \ldots, t_n)$ be a transcendental Liouvillian extension introduced 
in Definition \ref{DEF:liouvillian}. 
Then $F_i = C(t_1, \ldots, t_i)$ is also a transcendental Liouvillian extension of $C$ for all $i \in [n]$. 
With Remark \ref{RE:algcr}, we state the main result of this paper as follows. 
\begin{thm} \label{TH:main}
There exists an algorithm that, for every element $h \in F_n$, constructs
a complete reduction $\Phi_{n,h}$ for $(F_n, \, \cR_h)$.
\end{thm}
\begin{proof}
We proceed by induction on $n$. If $n=0$, then the conclusion holds by Example \ref{EX:const}.

Let $n>0$. Assume that, for every $\alpha \in F_{n-1}$,
there exists an algorithm that constructs
a complete reduction $\Phi_{{n-1},\alpha}$  for $(F_{n-1}, \cR_\alpha)$, and assume that
the map in Hypothesis \ref{HYP:ind} is defined by $\alpha \mapsto \Phi_{{n-1},\alpha}$.
For $h \in F_n$, let $(\bfxi, \bfeta)$ be the canonical form of $h$.
Then, with the induction hypothesis, one can construct a complete reduction for $(F_n, \cR_\bfxi)$ by either Theorem \ref{TH:pcr} or Theorem \ref{TH:hcr},
depending on whether $t_n$ is primitive or hyperexponential over $F_{n-1}$.
Thus,  
a complete reduction $\Phi_{n,h}$ for $(F_n, \cR_h)$ is obtained from Proposition \ref{PROP:ks} (ii).
The map in Hypothesis \ref{HYP:ind} is then defined by $h \mapsto \Phi_{n,h}$ for all $h \in F_{n}$.
\end{proof}

In practice, we need to maintain some initial data in order to avoid choosing elements from a set
in different ways (see \cite[Remark 2.5]{DGLL2025}).
Let $\Theta_0=\{1\}$.
For each $i \in [n]$, we have an effective basis $\Theta_i$ of $F_i$ given in \cite[Section 2.2]{DGLL2025}.
We set up a table $\bT$ indexed by $(i, \alpha)$, where $i \in [n]$ and $\alpha \in F_i$ is $t_i$-normalized.
The value of $(i, \alpha)$ is an initial (resp.\ echelon) sequence, where $i \ge 1$ and $t_i$ is primitive (resp.\ hyperexponential).
See Definition \ref{DEF:iseq} for the notion of initial sequences.
The table can be understood as the map in Hypothesis \ref{HYP:ind} with $F=F_i$ for each $i \in [n]$.
It is global and will be updated as long as a complete reduction for $(F_i, \cR_\alpha)$ is constructed.
These considerations lead to an algorithm for computing R-pairs with respect to~$\Phi_{n,h}$ given in Theorem \ref{TH:main}.
Below is an outline of the algorithm, in which we initialize the table $\bT$ to be empty. 
\begin{out} \label{OUT:Rpair}
Given $f, h \in F_n$, compute an R-pair of $f$ w.r.t.\ $\Phi_{n,h}$. 
\begin{enumerate}
\item $(^*${\sl Base case}$^*)$
If $n=0$, then an R-pair of $f$ w.r.t.\ $\Phi_{n,h}$ is $(f/h, 0)$ if $h \neq 0$,  and $(0,f)$ if $h=0$.
The algorithm terminates. \hfill 
 $(^*${\sl See Examples \ref{EX:const} and \ref{EX:choice0}}$^*)$ 
\item  $(^*${\sl Preprocessing}$^*)$ 
\begin{itemize}
\item[(2.1)] Find the canonical form $(\bfxi, \bfeta)$ of $h$ w.r.t.\ $t_n$ by Algorithm {\sc GKS} in \cite[Section 3.2]{CDGL2025}.
\item[(2.2)] Compute $(g, r, s) \in F_n \times F_{n-1}\langle t_n \rangle \times S_{t_n, \bfxi}$ s.t.\
$\bfeta f = \cR_\bfxi(g) + r/b + s$
by Algorithm {\sc GKSR} in \cite[Section 3.4]{CDGL2025}, where $b=\den(\bfxi)$  and $S_{t_n, \bfxi}$ is defined by \eqref{EQ:subspaceh}. \hfill  $(^*${\sl See \eqref{EQ:pre2}}$^*$) 
\item[(2.3)] If $r=0$, then an R-pair of $f$ w.r.t.\ $\Phi_{n,h}$ is $(\bfeta^{-1} g, \bfeta^{-1} s)$.
The algorithm terminates.

\hfill  $(^*${\sl See Propositions \ref{PROP:red} and \ref{PROP:ks}}$^*)$
\end{itemize}

\item $(^*${\sl Recursion}$^*)$
Assume that $t_n$ is primitive (resp.\ hyperexponential) over $F_{n-1}$.
\begin{enumerate}
\item[(3.1)] Compute an auxiliary pair $(p,q)$ of $r$ w.r.t.\ $(F_{n-1} \langle t_n \rangle, \cP_\bfxi)$ by Algorithm \ref{ALG:paux} (resp.\ Algorithm \ref{ALG:haux}).
\item[(3.2)] If $q =0$, then $(\bfeta^{-1} (g+p), \bfeta^{-1} s)$ is a required R-pair.
The algorithm terminates. 

\hfill $(^*${\sl See Propositions \ref{PROP:paux}, \ref{PROP:haux} and Remark \ref{RE:red}}$^*).$
\item[(3.3)] $(^*${\sl Updating $\bT$}$^*)$
If $(n,\bfxi)$ is an index of $\bT$, then set $L$ to be the value of $(n,\bfxi)$.
Otherwise, compute an initial sequence $L$ of $I_h$ by Algorithm \ref{ALG:iseq} (resp.\ an echelon sequence~$E$
of~$J_h$ by Algorithm \ref{ALG:hseq}), and add $(n,\bfxi)=L$ (resp.\ $(n,\bfxi)=E$) to $\bT$.

\item[(3.4)] $(^*${\sl Trivial intersection}$^*)$
If $L= \emptyset$ (resp.\ $E= \emptyset$), then $\left(\bfeta^{-1} (g+p), \bfeta^{-1} \left(q/b+s\right) \right)$
is a required R-pair of $f$ w.r.t.\ $\Phi_{n,h}$. The algorithm terminates.

\hfill $(^*${\sl See Propositions \ref{PROP:pinter}, \ref{PROP:hinter} and Remark \ref{RE:red}}$^*)$

\item[(3.5)] $(^*${\sl Nontrivial intersection}$^*)$
By Algorithm \ref{ALG:pproj} (resp.\ Algorithm \ref{ALG:hproj}), compute an element $\tilde{q} \in F_{n-1}\langle t_n \rangle$ and
an element $\tilde{r} \in W$
such that $q = \cP_\bfxi( \tilde{q}) +  \tilde{r}$, where $W$ is the complement of $\im(\cP_\bfxi)$ in $F_{n-1}\langle t_n \rangle$ induced
by $L$ (resp.\ $E$) (see Definition \ref{DEF:icomp} (resp.\ Definition \ref{DEF:ecomp})).
A required  R-pair is
$\left(\bfeta^{-1} (g+p+ \tilde{q}), \bfeta^{-1} \left(\tilde{r}/b+s\right) \right).$

\hfill  $(^*${\sl See Theorems \ref{TH:pcr}, \ref{TH:hcr} and Remark \ref{RE:red}}$^*)$
\end{enumerate}
\end{enumerate}
\end{out}
\begin{remark} \label{RE:plog}
The complete reduction outlined above depends on algorithms for solving logarithmic derivative recognition problem
and parametric logarithmic derivative problem. 
In practice, such algorithms (see \cite[Section 4.3.1]{RaabThesis}) can be implemented provided that a $\bQ$-basis of $C$ is available. 
\end{remark}

The two examples below demonstrate the procedure with all steps referenced to Outline \ref{OUT:Rpair}. 
The first example provides details of Example \ref{EX:exp}.
\begin{example} \label{EX:exptower}
	Determine whether
	$
	f=\dfrac{x}{1+\exp(x)} \cdot \exp\left(\dfrac{x}{1+\exp(x)} \right)
	$
	is a derivative in the transcendental elementary extension $\left(\bC(x,t,y), d/dx\right)$,  where $t=\exp(x)$ and $y=\exp\left(\dfrac{x}{1+\exp(x)} \right)$.
	
	Note that $f=gy$, where $g=\dfrac{x}{1+\exp(x)}$. To avoid excessive recursion, we set $h:=y^\prime/y$, which
    is equal to $g^\prime$.
	By the same argument as in Example~\ref{EX:risch}, it suffices to compute an R-pair of $g$ with respect to the complete reduction
$\Phi_{2,h}$ for $(\bC(x,t),\cR_h)$.
	
	Initially, we set $F_0=\bC$, $F_1=F_0(x)$, $F_2=F_1(t)$, and the table $\bT$ to be empty.

Since $h$ is a derivative in $F_2$, it is $t$-normalized. So $(\bfxi,\bfeta)=(h,1)$ in step 2.1.  With Convention \ref{CON:next2} and $h=\bfxi$, we have
	\begin{center}
		\begin{tabular}{c|c|c|c|c|c|c} 
			$a$ & $b$ & $m$ & $a_m$ & $b_m$ & $a_0$ & $b_0$ \\ \hline
			$1+(1-x)t$ & $(1+t)^2$ & $2$ & $0$ & $1$ & $1 $ & $1$ 
		\end{tabular}
	\end{center}

	In step 2.2, we have $g = \mathcal{R}_h(0) + (r/b) + 0$, where $r = x+xt$. Step 2.3 is skipped since $r \neq 0$.

	In step 3.1, Algorithm \ref{ALG:haux} computes an auxiliary pair $(p,q)=  \left(0, x+xt  \right)$ of $r$,
where $r$ belongs to the auxiliary subspace $V_h$ associated to $(F_1[t, t^{-1}], \cP_h)$. 

	
	Step 3.2 is skipped since $q$ is nonzero. 
	In step 3.3, we find that $(2,h)$ is not an index of $\mathbb{T}$. 
    Set $J_h := \im(\cP_h) \cap V_h$. An echelon sequence of $J_h$ is constructed as follows.
	
	First, we obtain that the type of $J_h$ is $(0,1),(-1,1)$ by solving two parametric logarithmic derivative problems. 
    Second, Algorithm \ref{ALG:haux} computes the auxiliary pairs of $\mathcal{P}_h(1)$ and $\mathcal{P}_h(t^{-1})$,
    which are $(0,-xt+t+1)$ and $(0,-t-1-x)$, respectively. 
Third, selecting   $xt$ as a pivot for $\mathcal{P}_h(1)$ and  $t$ as a pivot for $\mathcal{P}_h(t^{-1})$, we obtain the echelon sequence
	$$
	E: \, (t^{-1},-t-1-x,t), \, (1,-xt+t+1,xt)
	$$
    of $J_h$. Finally, the entry $(2,h)=E$ is added to $\mathbb{T}.$
	
	Step 3.4 is skipped since $E \neq \emptyset$.
	
In step 3.5, Algorithm \ref{ALG:hproj} projects $q$ to $\im(\cP_h)$ and
the complement of $\im(\cP_h)$ induced by~$E$. The respective projections are $\mathcal{P}_h(-1-t^{-1})$ and $0$.
Then an R-pair of $g$ with respect to $\Phi_{2,h}$ is $(-1-t^{-1},0)$, that is,
$g=\mathcal{R}_h(-1-t^{-1}).$   
It follows from $h=y^\prime/y$ and $f = g y$ that 
$$f = \bigl( \left(-1-t^{-1}\right) y \bigr)^\prime.$$
In functional notation, we have
	$
	\displaystyle \int f=-\bigl(1+\exp(-x)\bigr)\cdot \exp\left(\dfrac{x}{1+\exp(x)}\right).
	$
\end{example}
In the next example, $F_n$ is neither a primitive tower nor a hyperexponential one. 
\begin{example} \label{EX:infield2}
Let $t=\log(x)$. Determine whether
$$f=\dfrac{(2x^3+2x^2-1)t-t^3-t^2-2x^5+1}{x^2(t^2+1)} \cdot \exp\Big(\displaystyle\int\!\dfrac{2x^2-2t}{xt^2+x}\Big)$$
is a derivative in $\bQ(x, t, f)$, where the derivation is $d/dx$.

We rewrite $f$ in another form so as to make intermediate expressions more compact.

Let $y  = \exp\left(\displaystyle\int \dfrac{2x^2-2t}{xt^2+x}\right)$.
Then
$f  = g y,$
where
$g = \dfrac{(2x^3+2x^2-1)t-t^3-t^2-2x^5+1}{x^2(t^2+1)}.$
Moreover, $\bQ(x, t, f) = \bQ(x, t, y)$.
By the same reasoning in Example \ref{EX:risch}, $f$ belongs to $\bQ(x,t,y)^\prime$ if and only if
there exists an element $w \in \bQ(x,t)$ such that $\cR_h(w)=g$, where
$$h := \frac{y^\prime}{y}= \frac{2x^2-2t}{xt^2+x}.$$
So it suffices to compute an R-pair of $g$ with respect to the complete reduction $\Phi_{2,h}$ for $(\bQ(x,t), \cR_h)$. 

Let $F_0=\bQ$, $F_1=F_0(x)$ and $F_2=F_1(t)$. Then $h$ is $t$-normalized by a straightforward verification. So $(\bfxi, \bfeta)=(h,1)$ in step 2.1.
With Convention \ref{CON:next2} and $\bfxi=h$, we have 
\begin{center}
	\begin{tabular}{c|c|c|c|c|c|c} 
$a$ & $b$ & $m$ & $a_m$ & $b_m$ & $a_0$ & $b_0$ \\ \hline
$-2 x^{-1} t + 2 x $ & $t^2 + 1$ & $2$ & $0$ & $1$ & $2 x $ & $1$ 
	\end{tabular}
\end{center}

In step 2.2, we have $g = \mathcal{R}_h(0) + (r/b)$, where $r = - x^{-2} t^3 - x^{-2} t^2 + (2x+2-x^{-2}) t - 2 x^3 + x^{-2}$.

In step 3.1, we compute an auxiliary pair 
\begin{equation} \label{EQ:apair}
(p,q)=  \left(x^{-1}t, 2xt-2x^3  \right)
\end{equation}
of $r$ with respect to $(F_1[t], \cP_h)$.
During the process of the auxiliary reduction, an R-pair of $\lc_t(r)$ with respect to $\Phi_{1,0}$ is computed. So
the entry $(1,0) = 1, (0,1),$ $(0,0)$, $1$ is added to $\bT$.

Note that $h$ is the same as that in Example \ref{EX:prim}, in which the last table presents an initial sequence $L$ with
$L[1] = 1, L[2]=(0, x^{-1}), L[3]=  (0, 2x^{-1}),  L[4] = x^{-1},  L[5]=2,$ 
$$L[6] = \left(t^2-x^2+1, 2xt - 2x^3, xt \right), \,\, L[7] = \left(1, -2x^{-1} t+2x, x^{-1}t \right),$$
and $L[8] = \left(t,-x^{-1}t^2+2xt+x^{-1} , x^{-1} t^2 \right).$
Step 3.3 is then completed by adding $(2, h) = L$ to~$\bT$.

Step 3.4 is skipped because $L  \neq \emptyset$. 

Let $E$ be the echelon sequence determined by $L$.
In step 3.5, we project $q$ in \eqref{EQ:apair} to $\im(\cP_h)$ and $W$, respectively, where $W$ is the complement of $\im(\cP_h)$ induced by $E$.
Since $\deg_t(q)=1$, 
no further members in $E$ need to be computed. In fact, $q$ is equal to the second component of~$L[6]$.
Then the respective projections of $q$ are $\cP_h(t^2-x^2+1)$ and $0$.
It follows that an R-pair of $g$ with respect to $\Phi_{2,h}$ is  $(x^{-1}t + t^2 -x^2+1,0)$. In other words, 
$g = \cR_h(x^{-1}t + t^2 -x^2+1)$, 
which implies that $f = (g y)^\prime$. In functional notation, we have
\[
 \displaystyle \int f = \left(\frac{\log(x)}{x}+\log(x)^2-x^2+1  \right) \cdot \exp\left(\displaystyle \int \dfrac{2x^2-2\log(x)}{x\log^2(x)+x}\right).\]
The same integral may also be evaluated using the {\tt int()} command with  option {\tt method=parallelrisch} in {\sc maple} 2026, 
and  the {\tt Integrate[]} command in {\sc mathematica} 14.3.
\end{example}

\section{Applications} \label{SECT:app}

We present two applications of the complete reduction in Theorem \ref{TH:main} for the case $h=0$, in which the Risch
operator $\cR_0$ is the derivation on a transcendental Liouvillian extension.
The first application concerns elementary integration, and the second one addresses reduction-based creative telescoping.
To this end, we make the following notational convention throughout this section, because we only consider derivations.
\begin{convention} \label{CON:elem}
Let $F_0 = C$ and $F_n=C(t_1, \ldots, t_n)$ be  a transcendental Liouvillian extension.
For all $i \in [n]_0$, $\Psi_i$ stands for  the complete reduction $\Phi_{i,0}$ for $(F_i, \, ^\prime)$ given in Theorem \ref{TH:main}.
\end{convention} 
With the above convention, we see that $\Psi_0$ is the identity map of $C$ and $\Psi_n$ is a complete reduction for $(F_n, \, ^\prime)$.
For elementary integration, we need to distinguish primitive generators of~$F_n$ from hyperexponential ones.

Set
$\bP:=\{i \in [n] \mid t_i^\prime \in F_{i-1}\}$ and $\bH := [n] \setminus \bP.$
Moreover, define 
$$ P := \spa_C\{ \Psi_{i-1}(t_i^\prime) \mid i \in \bP \}
\quad \text{and} \quad  H :=  \spa_C\{ \Psi_{j-1}(t_j^\prime/t_j) \mid j \in \bH \}.$$
\begin{remark} \label{RE:type}
If $i \in \bP$, the type of $\im(\cP_0) \cap U_0$ in $F_{i-1}[t_i]$ can be taken as $1$ by Example \ref{EX:type}.
\end{remark}

Some useful properties of remainders are given in the following two technical lemmas.
\begin{lemma} \label{LM:restrict}
For  $i \in [n]_0$ and $f \in F_i$,  $\Psi_n(f)-\Psi_i(f) \in P.$
\end{lemma}
\begin{proof}
Let $k \ge i$. By Remark \ref{RE:type} and Corollary  \ref{COR:prestrict}, we see that 
$\Psi_{k+1}(f) - \Psi_{k}(f) \in P$ if $k+1 \in \bP$.
By Corollary \ref{COR:hrestrict},  $\Psi_{k+1}(f) - \Psi_k(f) = 0$ if $k+1 \in \bH$.  
Then $\Psi_n(f)-\Psi_i(f) \in P$
by the identity $\Psi_n(f)-\Psi_i(f) = \sum_{k=i}^{n-1} \left( \Psi_{k+1}(f) - \Psi_k(f) \right).$
\end{proof}

Recall that $S_{t_i}$ stands for the set of $t_i$-simple elements in $F_{i-1}(t_i)$ for all $i \in [n]$.
Set $S:=S_{t_1} + \cdots +S_{t_n}$, which is a direct sum.
The next definition extends \cite[Definition 3.4]{DGLW2020}. 
\begin{define} \label{DEF:simple}
An element of $F_n$ is said to be {\em simple} if it belongs to $S$.
\end{define}
\begin{lemma} \label{LM:rem}
For every simple element $s$,  we have $s - \Psi_n(s) \in P$.
\end{lemma}
\begin{proof}
Let $s = \sum_{i \in [n]} s_i$, where $s_i \in  S_{t_i}$.   By Proposition \ref{PROP:red} and Remark \ref{RE:red}, $\Psi_i(s_i)=s_i$ for all $i \in [n]$. Then
$s - \Psi_n(s) = \sum_{i \in [n]} \left( \Psi_i(s_i) - \Psi_n(s_i) \right)$
by $\Psi_n(s) = \sum_{i \in [n]} \Psi_n(s_i)$.
The lemma follows from Lemma \ref{LM:restrict}.
\end{proof}

Next, we define residues of elements in $S$ without reference to any particular generator of $F_n$.
\begin{define} \label{DEF:sresidue}
Let $s$ be a simple element of the form $\sum_{i \in [n]} s_i$ with $s_i \in  S_{t_i}$.
A {\em residue of $s$} is a residue of some~$s_j$ with $j \in [n]$ 
at a normal polynomial of positive degree in $F_{j-1}[t_j]$.
\end{define}
Residues of simple elements are well-defined by $S = \oplus_{i \in [n]}  S_{t_i}$ and Remark \ref{RE:residue}.

We are ready to generalize \cite[Theorem 5.3]{DGLL2025} from primitive towers to transcendental Liouvillian extensions.
\begin{thm} \label{TH:elem}
With Convention \ref{CON:elem}, we further assume that $C$ is algebraically closed. Then
an element~$f$ of $F_n$ has an elementary integral over $F_n$ if and only if
\begin{itemize}
  \item [(i)] there exists a simple element $s$ such that
  $\Psi_n(f) - s \in  H + P,$ and
  \item [(ii)] all residues of $s$ belong to $C$.
\end{itemize}
\end{thm}
\begin{proof} Assume that both (i) and (ii) hold.  By (ii) and Lemma \ref{LM:elem},
$s$ has an elementary integral over $F_n$.  Combining with (i), it suffices to show that every element of $H + P$ has an elementary integral over~$F_n$.
For $i \in \bH$, we have $t_i^\prime/t_i = u_i^\prime + \Psi_{i-1}(t_i^\prime/t_i)$ for some $u_i \in F_{i-1}$, which implies that
$\Psi_{i-1}(t_i^\prime/t_i)$ is the derivative of $\log(t_i)-u_i$.
For $i \in \bP$, we have $t_i^\prime =  v_i^\prime + \Psi_{i-1}(t_i^\prime)$ for some $v_i \in F_{i-1}$,
yielding that $\Psi_{i-1}(t_i^\prime) \in F_n^\prime$.
Accordingly, $\Psi_n(f)$ has an elementary integral over $F_n$, and so does $f$.

Conversely, assume that $f$ has an elementary integral over $F_n$.
Then there exists a $C$-linear combination $g$ of logarithmic derivatives in $F_n$ such that
  $f \equiv g \mod F_n^\prime$  by \cite[Theorem 5.5.2]{BronsteinBook}.
  Thus, $\Psi_n(f)=\Psi_n(g)$ by $\Psi_n(F_n^\prime)=\{0\}$.  It remains to show that  (i) and (ii) hold for $\Psi_n(g)$.

By a repeated use of Lemma \ref{LM:logder} and the logarithmic derivative identity, there exists an element $s \in S$
with merely constant residues,
such that $g - s$ is a $C$-linear combination of $t_i^\prime/t_i$ for $i \in \bH$.
Then 
$\Psi_n(g)-\Psi_n(s) \in H+P$ 
by Lemma \ref{LM:restrict}. So $\Psi_n(g) - s \in H+P$ by Lemma \ref{LM:rem}. Both (i) and (ii) hold.
\end{proof}

To compute elementary integrals by the above theorem, we set
$R_i := t_i \cdot F_{i-1}[t_i]$ if $i \in \bP$,  and
$R_i := \left(t_i \cdot F_{i-1}[t_i]\right) + \left( t_i^{-1} \cdot F_{i-1}[t_i^{-1}] \right)$ if $i \in \bH$.
Moreover, let $R=C+R_1+\cdots+R_n$. This sum is  evidently direct, and so is $R+S$.
\begin{prop} \label{PROP:decomp}
With Convention \ref{CON:elem} and the notation just introduced, we obtain that
$$\im(\Psi_n) \subset R \oplus S.$$
\end{prop}
\begin{proof}
We proceed by induction on $n$.   For $n=0$, $\im(\Psi_0)=C$ by Example \ref{EX:const}.
So $\im(\Psi_0) \subset R$.
The conclusion holds.
Assume that $n>0$ and that the conclusion holds for $n-1$.
We need to consider the cases in which $n$ belongs to either $\bP$ or $\bH$ separately. 
Assume that $n \in \bP$. Then $\im(\Psi_n) \subset \im(\Psi_{n-1}) + R_n + S_n$ by Corollary \ref{COR:prestrict} (ii). Otherwise, $\im(\Psi_n) \subset \im(\Psi_{n-1}) + R_n+S_n$ by Corollary  \ref{COR:hrestrict} (ii). In either case,  $\im(\Psi_n) \subset R + S$
holds by the induction hypothesis.
%
\end{proof}

Next, we outline an algorithm for computing elementary integrals over $F_n$.
\begin{out} \label{OUT:elem}
Given $f\in F_n$, determine whether $f$ has an elementary integral over $F_n$, and computes such an integral if there exists one. 
\begin{enumerate}
\item[1.] Compute an R-pair $(g, \Psi_n(f))$.
If $\Psi_n(f)=0$, then $\int f = g$ and the algorithm terminates.
\item[2.] By Proposition \ref{PROP:decomp}, we  decompose
\begin{itemize}
\item 
$\Psi_n(f) = r_f + s_f$,   
\item $\Psi_{i-1}(t_i^\prime/t_i) = r_i + s_i$ if $i \in \bH$, and $\Psi_{i-1}(t_i^\prime) = r_i + s_i$ if $i \in \bP$,
\end{itemize}
where $r_f, r_i \in R$ and $s_f, s_i \in S$.
\item[3.] Let $z_1, \ldots, z_n$ be constant indeterminates.
\begin{itemize}
\item[3.1.] Compute the  augmented matrix $(M|v) \in C^{k\times (n+1)}$ of a  linear system in $z_1, \ldots ,z_n$
 induced by the equation $r_f - \sum_{i \in [n]} z_i r_i =0.$


\item[3.2.] Using \cite[Algorithm 2.8]{DGLL2025},
compute the augmented matrix $(N|w) \in C^{l\times (n+1)}$ of a linear system in $z_1, \ldots ,z_n$ induced by 
the condition that all residues of $s_f - \sum_{i\in [n]}z_i s_i$ are constants.

\item[3.3.] Solve the combined linear system given by
$\left(\begin{array}{c|c}
M & \vv \\
N & \vw
\end{array}
\right)
$.
If the system has no solution, then $f$ has no elementary integral over $F_n$ by Theorem \ref{TH:elem} and the fact that $R+S$ is direct.
The algorithm terminates.
\end{itemize}
\item[4.] Let $c_1, \ldots, c_n$ be a solution of the combined linear system. Set
$s := s_f - \sum_{i \in [n]} c_i s_i.$ Then
\begin{align*}
\Psi_n(f) & =  r_f + s_f \\
    & =  \left(r_f - \sum_{i=1}^n c_i r_i \right) + \left(s_f - \sum_{i=1}^n c_i s_i \right) + \sum_{j \in \bH} c_j \Psi_{j-1}\left(
    \frac{t_j^\prime}{t_j}\right)
      + \sum_{k \in \bP} c_k \Psi_{k-1} \left(t_k^\prime\right) \\
    & =  s + \sum_{j \in \bH} c_j \Psi_{j-1}\left(\frac{t_j^\prime}{t_j}\right)
    + \sum_{k \in \bP} c_k \Psi_{k-1}(t_k^\prime).
\end{align*}
The integral of $s$ is elementary over $F_n$ by Lemma \ref{LM:elem}.
Let $(q_j, \Psi_{j-1}(t_j^\prime/t_j))$ be an R-pair of $t_j^\prime/t_j$ for all $j \in \bH$, and
$(q_k, \Psi_{k-1}(t_k^\prime))$ be an R-pair of $t_k^\prime$ for all $k \in \bP$.
It follows  that
\begin{equation} \label{EQ:int}
\int f = g + \int s +  \sum_{j \in \bH} c_j (\log(t_j)-q_j) + \sum_{k \in \bP} c_k (t_k-q_k),
\end{equation}
where $s$ is integrated by determining its residues (see \cite[Theorem 4.4.3]{BronsteinBook} and \cite{DGGL2023}).
\end{enumerate}
\end{out}

Note that the integral of $s$ in \eqref{EQ:int} may involve elements in the algebraic closure of $C$. 
\begin{remark} \label{RE:log}
\cite[Algorithm 2.8]{DGLL2025} used in step 3.2 originates from \cite[Theorem 3.9]{RaabThesis}.
In practice, such algorithms can be implemented provided that a $\bQ$-basis of $C$ is available. 
\end{remark}
\begin{example} \label{EX:elem1}
Let us integrate
$$f:=\dfrac{(\log x+1)\, \Li(x^c)}{x^{c+1}},$$
where $c$ is a constant indeterminate,
$\Li(x)$ is the logarithmic integral, i.e., $\Li(x)'=1/\log x.$

Let $C$ be the algebraic closure of $\bQ(c)$,  
$t_1=x,$ $t_2=x^c,$ $t_3=\log (x),$ and $t_4={\rm Li}(x^{c}).$
Then $t_1'=1$, $t_2'/t_2 =c/t_1,$  $t_3' = 1/t_1$, $t_4' = t_2 /(t_1t_3),$ and 
$F_4:=C(t_1,t_2,t_3,t_4)$ is a transcendental Liouvillian extension. Moreover, $f=(t_3 t_4 +t_4)/(t_1t_2).$
\begin{enumerate}
\item By the algorithm in Outline \ref{OUT:Rpair}, an R-pair of $f$ with respect to $\Psi_4$ is $(g, \Psi_4(f))$, where
$$g=\dfrac{ c t_2t_3- c t_3t_4-(c +1)t_4}{c^2t_2} \quad \text{and} \quad \Psi_4(f)=\dfrac{c+1}{c^2t_1t_3}.$$ 
\item Projecting  remainders $\Psi_4(f)$, $\Psi_{0}(t_1')$, $\Psi_{1}(t_2'/t_2)$, $\Psi_2(t_3^\prime)$
and $\Psi_3(t_4^\prime)$ with respect to $R \oplus S$, we  obtain  the following table

\begin{center}
	\begin{tabular}{c|c|c|c|c|c}
		    & $\Psi_0(t_1')$ & $\Psi_1(t_2'/t_2)$ & $\Psi_2(t_3')$ & $\Psi_3(t_4^\prime)$ & $\Psi_4(f)$   \\ \hline
		{\rm Proj.\ in $R$}  & $ r_1 = 1$       & $r_2=0$  & $r_3=0$   & $r_4 = 0$  & $r_f = 0$ \\ \hline
		{\rm Proj.\ in $S$} & $s_1=0$      & $s_2=c/t_1$  & $s_3=1/t_1$ & $s_4 = t_2/(t_1t_3)$  & $s_f = \Psi_4(f)$ 
	\end{tabular}
\end{center}
\item Let $z_1, z_2, z_3$ and $z_4$ be constant indeterminates. By making two ansatzes:
\begin{itemize}
\item[(i)] $r_f = z_1 r_1 + z_2 r_2 + z_3 r_3 + z_4 r_4$, and
\item[(ii)] $s_f - z_1 s_1 - z_2 s_2 - z_3 s_3 - z_4 s_4$ has merely constant residues,
\end{itemize}
we obtain a linear system in $z_1$, $z_2$, $z_3$  and $z_4$ with augmented matrix 
$$\left(\begin{array}{cccc|c}
1&0&0&0& 0\\
0&0&0&-c & 0
\end{array}
\right). $$
The system has a solution $(0,0,0,0).$ Hence, $f$ has an elementary integral over $F_4.$
\item Computing the residues of $s_f$ yields that
\[ \int f=\dfrac{1}{c}\log (x)+\dfrac{c+1}{c^2}\log(\log (x)) -\dfrac{ c \log x+( c +1)}{c^2x^{c}} \Li(x^{c}). \]
\end{enumerate}
\end{example}
\begin{example} \label{EX:elem2}
Let $F_0=\mathbb{C}$ and $F_3=F_0(t_1, t_2, t_3)$, where
$$\dfrac{t_1'}{t_1}= \sqrt{-1}, \quad t_2'=\dfrac{\sqrt{-1}(t_1^2+1)}{t_1^2-1} \quad \text{and} \quad \frac{t_3'}{t_3}=\dfrac{t_2^2-1}{2 \sqrt{-1}t_2}. $$
Then $F_3$ is a transcendental Liouvillian extension, where $t_1$, $t_2$ and $t_3$  model
$$\exp(\sqrt{-1}x),  \quad \log( \sin x) \quad \text{and} \quad  \exp\left (\int \frac{\log(\sin x)^2-1}{2 \sqrt{-1}\log(\sin x)} \right),$$
respectively. Note that no nonzero constant is  a derivative in $F_3$.

 Let us try to integrate
$$f:=\dfrac{(\sqrt{-1}t_2^2-\sqrt{-1}t_1^2t_2^2-2)t_3+\sqrt{-1}t_1^2t_2+3\sqrt{-1}t_2-2t_2}{2(t_1^2-1)t_2(t_3+t_2)}$$
over $F_3$ by the algorithm in Outline~\ref{OUT:elem}.
\begin{itemize}
\item[1.] An R-pair of $f$ with respect to $\Psi_3$ is $(0, f)$. So $\Psi_3(f)=f$ and $f \notin F_3^\prime$.
\item[2.] Projecting the corresponding remainders with respect to $R \oplus S$ yields that 
\begin{center}
	\begin{tabular}{c|c|c|c|c}
		  & $\Psi_0(t_1'/t_1)$ & $\Psi_1(t_2')$ & $\Psi_2(t_3'/t_3)$ & $\Psi_3(f)$    \\ \hline
		{\rm Proj.\ in $R$} & $ r_1 = \sqrt{-1}$       & $r_2=\sqrt{-1}$  & $r_3=t_2/\left(2\sqrt{-1}\right)$  & $r_f$    \\ \hline
		{\rm Proj.\ in $S$} & $s_1=0$      & $s_2=2\sqrt{-1}/(t_1^2-1)$  & $s_3=-1/\left(2\sqrt{-1}t_2\right)$   & $s_f$
	\end{tabular}
\end{center}
with
$r_f=\dfrac{t_2}{2\sqrt{-1}}$ and $s_f=\dfrac{1}{(1-t_1^2)t_2}+\dfrac{\sqrt{-1}(t_1^2t_2^2+t_1^2-t_2^2+3)}{2(t_1^2-1)(t_2+t_3)}.$

\item[3.] Let $z_1, z_2, z_3$ be  three constant indeterminates. Making ansatzes analogous to those in the third step of the above example leads to
the linear system in $z_1$, $z_2$ and $z_3$ whose augmented matrix is
$$\left(\begin{array}{ccc|c}
	\sqrt{-1}&\sqrt{-1}&0 & 0\\
	0&0&1 & 1 \\
	0&0& \sqrt{-1} & 1
\end{array}
\right).
$$
Since the system is inconsistent, we conclude that $f$ has no elementary integral over $F_3.$
\end{itemize}
\end{example}

 Last but not least, we introduce one more application of our complete reduction to creative telescoping. 
We consider the case in which the integrand contains a shift variable $k$. More precisely, let $C=\bC(k)$, where $k$ is a constant indeterminate.
Assume further that the shift operator $\cS_k: k \mapsto k+1$ can be extended to $F_n$, and that the extended operator commutes with the derivation
on $F_n$. For every element $f \in F_n$ and every $i \in \bN_0$, 
one can compute an R-pair $(g_i, r_i)$ of $\cS_k^i(f)$ with respect to $\Psi_n$. Then
there exists a recurrence operator $\cL \in C[\cS_k]^\times$ of order no more than $i$ such that $\cL(f) \in F_n^\prime$ if and only if
$r_0, r_1, \ldots, r_i$ are $C$-linearly dependent. This allows us to construct a telescoper for $f$ up to a given order.
Unfortunately, we have not found any criterion for the existence of telescopers in such extensions.  Discussions related to the construction
of telescopers in primitive towers are given in \cite[Section 5.2]{DGLL2025}.

We are now ready to supplement the details for Example \ref{EX:tele0}.
\begin{example} \label{EX:tele}
For all $k \in \bN$, compute
$$\int_{0}^{\frac{\pi}{2}}\cos(2kx)\log(\sin(x))~dx \quad \text{and} \quad  \int_{0}^{\frac{\pi}{2}}\sin(2kx)\log(\sin(x))~dx.$$

Let $f(k,x)=\exp(2k\sqrt{-1}x)\log(\sin(x))$ and $A(k) = \int_{0}^{\frac{\pi}{2}}f(k,x)~dx$.
The integrand $f(k,x)$ is not a D-finite function over $\bC(k,x)$.
The two integrals are the real and imaginary parts of $A(k)$, which are denoted by $R(k)$ and $I(k)$, respectively.

Let $F_0=\bC(k)$ and $F_4=F_0(t_1, t_2, t_3,t_4)$ with
$$t_1'=1, \quad \frac{t_2'}{t_2} = \sqrt{-1}, \quad t_3'=\frac{\sqrt{-1}(t_2^2+1)}{t_2^2-1},\quad \text{and} \quad  \frac{t_4'}{t_4}=2k\sqrt{-1}.$$
The generators $t_1$ $t_2$, $t_3$ and $t_4$ model $x$, $\exp(\sqrt{-1}x)$, $\log(\sin(x))$ and $\exp(2k\sqrt{-1}x)$, respectively.
Then $f(k,x) = t_3t_4$ and $f(k+1,x) = t_2^2t_3t_4$.
Using $\Psi_4$, we find an R-pair $(g_0, r_0)$ of $f(k,x)$ and an R-pair $(g_1, r_1)$ of $f(k+1,x)$, where
$$ (g_0, r_0) = \left( - \frac{\sqrt{-1}(2kt_3-1)t_4}{4 k^2}, \,  -\frac{t_4}{k(t_2^2-1)}\right)$$
and
$$(g_1, r_1) = \left(-\frac{\sqrt{-1}(2k^2t_2^2 t_3 + 2k t_2^2 t_3 -k t_2^2 - 2k -2)t_4}{4(k+1)^2k}, \, -\frac{t_4}{(k+1)(t_2^2-1)}\right).$$
Since $(k+1) r_1 -k r_0 = 0$, we have
\begin{equation} \label{EQ:ct}
(k+1) f(k+1, x) - k f(k,x) = (k+1) g_1^\prime - k g_0^\prime.
\end{equation}
Although neither $\lim_{x \rightarrow 0^+} g_0$ nor $\lim_{x \rightarrow 0^+} g_1$ exists,
$$\lim_{x \rightarrow 0^+}\left((k+1) g_1 - k g_0\right) = \frac{\sqrt{-1}(2k+1)}{4 k(k+1)}.$$
Integrating both sides of \eqref{EQ:ct} from $0$ to $\pi/2$, we have
\begin{equation} \label{EQ:rec}
A(k+1)-\frac{k}{k+1}A(k)= \frac{\sqrt{-1}((-1)^k-2k-1)}{4k(k+1)^2}.
\end{equation}
To compute $A(k)$ for all $k \in \bN$, we first evaluate $A(1) = \int_{0}^{\frac{\pi}{2}} \exp(2\sqrt{-1}x)\log(\sin(x)) \, dx$.
Its integrand is represented by $t_2^2t_3$ in $F_3$. Using $\Psi_3$, the integrand has an R-pair
$$\left(-\frac{\sqrt{-1}}{2}t_2^2t_3+\frac{\sqrt{-1}}{4}t_2^2+\frac{\sqrt{-1}}{2}t_3-\frac{t_1}{2}, \,  0\right).$$
So the integrand belongs to $F_3^\prime$.
It follows that $A(1) =-\pi/4-\sqrt{-1}/2$. 

Taking respective real and imaginary parts of both sides in \eqref{EQ:rec} and $A(1)$, we have
\[
\left\{
\begin{array}{l}
R(k+1) - \displaystyle \frac{k}{k+1} R(k) = 0 \\ \\
R(1) = \displaystyle  - \frac{\pi}{4}
\end{array} \right.
\quad \text{and} \quad
 \left\{
 \begin{array}{l}
I(k+1) - \displaystyle  \frac{k}{k+1} I(k) = \displaystyle  \frac{(-1)^k-2k-1}{4k(k+1)^2} \\ \\
I(1) = \displaystyle  -\frac{1}{2}.
\end{array}
\right.
\]
The recurrence for $R(k)$ together with its initial condition implies $R(k) = - \pi/(4k)$.
Using the {\sc mathematica} package {\tt Sigma} (see \cite{Schn2006}),  we get
$$I(k) =\frac{1-(-1)^k}{4k^2}-\frac{1}{2k}\left(\sum_{i=1}^{k} \frac{1-(-1)^i}{i}\right)+\frac{c}{k}$$
for some $c\in \bC$. With $I(1)=-1/2$,
we find that
\[
I(k)=\frac{1-(-1)^k}{4k^2}-\frac{1}{2k}\left(\sum_{i=1}^{k} \frac{1-(-1)^i}{i}\right).
\]
\end{example}
\begin{remark}
The result $R(k)=-\pi/(4k)$ is given by Formula 4.384.7 in \cite{IntBook}.
Integrals in analogy to $I(k)$
are evaluated in \cite{LPR2002} and \cite[Chapter 5]{RaabThesis}.
\end{remark}

\section{Experiments} \label{SECT:expr}

We present empirical results about in-field integration obtained from
the complete reduction ({\tt CR}) in Theorem \ref{TH:main}  and the {\tt int()} command   in {\sc maple} without any option.
Experiments were carried out with {\sc maple 2026} on a computer with iMac CPU 3.6GHz, Intel Core i9, 16GB memory.
Our {\sc maple} scripts for {\tt CR}, the experimental data and timings are contained in  a zipped file 
under the name {\tt ReductionPair.zip} at \url{https://risc.jku.at/sw/reductionpair/}

Derivations acting on the differential fields in this section are all equal to $d/dx$.
An integrand is the derivative of a rational fraction in a transcendental Liouvillian extension
$\bQ(t_1, \ldots, t_n)$. 
A dense (resp.\ sparse) fraction has numerator and denominator generated by 
$$
\text{{\tt randpoly}{\tt ([t1, t2, \ldots, tn], degree=d)}}
$$
with option {\tt dense} (resp.\ the default value, which is sparse).

All integrands in our experiments are derivatives, because {\tt int()}
would try to integrate a non-derivative in other closed forms.
For each degree $d$, five derivatives were  integrated. Outputs of {\tt CR} were normalized to be rational fractions 
with coprime numerators and denominators.
 
Timings were measured by
{\sc maple} CPU time in seconds.
Computation would be aborted if one example took more than 3600 seconds. 
We mark an entry with \lq\lq $\int$\rq\rq\ whenever {\tt int} returned an unevaluated integral.
The correctness was verified for each output of {\tt CR}.

In the first suite, we let $n=3$, $t_1=x$, $t_2= \exp(x)$ and $t_3 = \log(\exp(x)+x)$. Fractions were dense in $t_1, t_2$ and $t_3$.
The average timings are summarized in Table~\ref{tab:dEL}.
\begin{center}
\begin{tabular}{c|c|c|c|c|c|c|c|c|c|c} 
    $d$          & $1$  & $2$ & $3$ &  $4$   & $5$ & $6$    & $7$ & $8$   & $9$ & $10$  \\ \hline
    {\tt  CR}    & 0.01 & 0.03 & 0.02 & 0.16 & 0.70  & 2.20 &  8.01 & 22.73 & 61.27 &   $>$3600 \\ \hline
    {\tt int}    & 0.04 & 0.08 & 0.16 & 0.48 & 1.44  & 4.59 & 27.84 & 44.41 & 111.00 &  $>$3600
\end{tabular}
         \captionof{table}{Dense fractions in $x$, $\exp(x)$ and $\log(\exp(x)+x)$} 
         \label{tab:dEL}
\end{center}
In the second suite, we let $n=3$, $t_1=x$, $t_2= \log(x^2+1)$ and $t_3 = \exp(x^2/2)$.
Integrands were generated in the same way for the first suite. The average timings are given in Table~\ref{tab:dLE}.
\begin{center}
\begin{tabular}{c|c|c|c|c|c|c|c|c|c|c} 
    $d$          & $1$  & $2$ & $3$ &  $4$    & $5$   & $6$  & $7$   & $8$   & $9$   & $10$  \\ \hline
    {\tt  CR}    & 0.02 & 0.01 & 0.05 &  0.11 & 0.44  & 1.61 &  4.83 & 12.29 & 39.43&  96.84 \\ \hline
    {\tt int}    & 0.04 & 0.09 & 0.23 &  0.74 & 2.04  & 6.83 & 19.78 & 66.33 & 176.75 & $>$3600
\end{tabular}
         \captionof{table}{Dense fractions in $x$, $\log(x^2+1)$ and $\exp(x^2/2)$} 
         \label{tab:dLE}
\end{center}

We let $n=4$, $t_1 = x$, $t_2 = \log(x^2+1)$, $t_3 =  \exp(x^2/2)$
and $t_4 = \exp(\exp(x^2/2))$ in the third suite. Rational fractions  were sparse, because dense fractions were too huge when $d>3$.
The average timings are given in Table~\ref{tab:sLEE}.
\begin{center}
\begin{tabular}{c|c|c|c|c|c|c|c|c|c|c}
    $d$          & $1$  & $2$  & $3$  &  $4$  & $5$   & $6$  & $7$   & $8$    & $9$  & $10$  \\ \hline
    {\tt  CR}    & 0.01 & 0.02 & 0.04 &  0.07 & 0.15  & 0.97 &  2.42 & 4.40 & 30.53 &  104.30\\ \hline
    {\tt int}    & 0.07 & $\int$ & 0.50 &  $\int$ & 1.06  & 2.84 &  21.11 & 14.62 & $\int$ & 219.41
\end{tabular}
         \captionof{table}{Sparse fractions in $x$, $\log(x^2+1)$, $\exp(x^2/2)$ and $\exp(\exp(x^2/2))$} 
         \label{tab:sLEE}
\end{center}

We integrated the derivatives of dense polynomials in $t_1, t_2, t_3$ and  $t_4$.
The average timings are given in Table~\ref{tab:dLEE}.
\begin{center}
\begin{tabular}{c|c|c|c|c|c|c|c|c|c|c} 
    $d$            & $11$ &  $12$ & $13$  & $14$  & $15$   & $16$  & $17$  & $18$  & $19$   & $20$  \\ \hline
    {\tt  CR}      &  0.24 & 0.31  & 0.39  &  0.51  & 0.62  & 0.76  &  0.95 & 1.17   & 1.38 & 1.77 \\ \hline
    {\tt int}      & 8.82 &  13.32 & 19.64  & 30.13 &  53.06 & 83.45 & 141.59&  
    257.31 & 386.16  & 1242.43
\end{tabular}
         \captionof{table}{Dense polynomials in $x$, $\log(x^2+1)$, $\exp(x^2/2)$ and $\exp(\exp(x^2/2))$} 
         \label{tab:dLEE}
\end{center}

Lastly, we let $n=4$, $t_1=x$, $t_2 = \log(x^2+1)$, $t_3=\exp(x^2/2)$ and $t_4= \exp\left(x \log(x^2+1) \right)$.
The average timings for integrating the derivatives of sparse fractions and dense polynomials 
are summarized in Tables~\ref{tab:sLEE2} and~\ref{tab:dLEE2}, respectively.
\begin{center}
\begin{tabular}{c|c|c|c|c|c|c|c|c|c|c}
    $d$          & $1$  & $2$ & $3$   &  $4$    & $5$   & $6$  &  $7$   & $8$     & $9$   & $10$  \\ \hline
    {\tt  CR}    & 0.01 & 0.02 & 0.05 &  0.28 & 0.41  & 0.19 &  20.45  & 0.14    & 3.43 & 518.15\\ \hline
    {\tt int}    & 0.08 & 0.21 & 0.36 &  $\int$ & $\int$  & 0.88 &  41.39  & 1.44    & 7.67 &  1449.17
\end{tabular}
         \captionof{table}{Sparse fractions in $x$, $\log(x^2+1)$, $\exp(x^2/2)$ and $\exp\left(x \log(x^2+1) \right)$} 
         \label{tab:sLEE2}
\end{center}

\begin{center}
\begin{tabular}{c|c|c|c|c|c|c|c|c|c|c}
    $d$            & $11$ & $12$ & $13$   & $14$  & $15$   & $16$  & $17$  & $18$   & $19$  & $20$  \\ \hline
    {\tt  CR}      & 0.40 &  0.54 & 0.67  & 0.81  &  1.04  & 1.27  & 1.54  &  1.88  & 2.22  & 2.68\\ \hline
    {\tt int}      & 10.06 &  14.66 & 21.34 & 31.84 & 51.05 & 71.08 & 92.80 & 126.09 & 174.35 & 279.45
\end{tabular}
         \captionof{table}{Dense polynomials in $x$, $\log(x^2+1)$, $\exp(x^2/2)$ and $\exp\left(x \log(x^2+1) \right)$} 
         \label{tab:dLEE2}
\end{center}

We did not compare {\tt CR} against {\tt int} for computing elementary integrals, because
such a comparison would need an  algorithm for determining  constant residues. 
Our {\sc maple} scripts use an evaluation-based algorithm described in \cite{DGGL2023}, which outperforms resultant-based algorithms 
developed in 1970's (see \cite[Section 4.4]{BronsteinBook}).

\section{Concluding remarks} \label{SECT:conc}

In this paper, we construct a complete reduction for $(F, \, \cR_h)$, where $F$ is a transcendental Liouvillian extension, $h \in F$
and $\cR_h$ is the
Risch operator associated to $h$. The complete reduction 
directly furnishes an algorithm for in-field integration,
and
leads to a new algorithm for computing elementary integrals over $F$.

Apart from generalizing the results presented in this paper from transcendental Liouvillian extensions to 
arbitrary Liouvillian extensions, in which some generators
may be algebraic, there remain at least two challenging problems. The first
 is to develop a reduction-based
algorithm for computing telescopers for elements in $F$ without any a priori order bound. 
Solving this would require a criterion
on the existence of telescopers for elements in $F$. 
The second problem concerns adapting remainders of the complete reduction for algebraic functions in \cite{CKK2016}
to enable more  efficient computation of elementary integrals over algebraic-function fields.

\section*{Acknowledgments}

We are grateful to the anonymous referees for friendly and
careful reviews. Their comments have encouraged and guided us to revise this submission 
substantially.

We thank Manuel Kauers and Clemens Raab for supportive discussions and valuable suggestions. 
Yiman Gao thanks the Computer Algebra and Applications group at RISC, led by Carsten Schneider, for their stimulating discussions.

Special thanks go to Junlin Xu from Maplesoft and Ralf Hemmecke from RISC
for helping us carry out
experiments with {\sc Maple} 2026 and {\sc FriCAS}, respectively.

\appendix

\section{Algorithm descriptions} \label{APP:alg}
This appendix is devoted to algorithmic descriptions of some constructive proofs in
Sections \ref{SECT:prim} and \ref{SECT:hyperexp}. 
Let $F$ be a differential field of characteristic zero, $C$ be the constant subfield of $F$, and $t$ be a regular monomial over $F$.
For brevity, we use the phrase \lq\lq the data in Convention \ref{CON:next2}\rq\rq\
as shorthand for  $F(t)$, $h$, $m$, $a$, $b$, $a_m$, $b_m$, $a_0$ and $b_0$ within 
all algorithm specifications. 
For $\alpha \in F$, $\Phi_\alpha$ stands for the complete reduction for $(F, \, \cR_\alpha)$ 
in Hypothesis \ref{HYP:ind}. Moreover, $\Theta$ is an effective basis of $F$.

\subsection{Algorithms in the primitive case} \label{SUBSECT:palg}

In this subsection, $t$ is  primitive  over $F$, $h \in F(t)$ is $t$-normalized, 
$U_h$ stands for the auxiliary subspace given in Definition \ref{DEF:paux}, and
$I_h := \im(\cP_h) \cap U_h$.

The first algorithm describes an auxiliary reduction in $F[t]$. It is based on \eqref{EQ:paux_case}
and the proof of Proposition \ref{PROP:paux}.
\begin{alg}
{\sc PrimAuxRed} \label{ALG:paux}

\smallskip \noindent
{\sc Input:}  $f \in F[t]$ and the data in Convention \ref{CON:next2}

\smallskip \noindent
{\sc Output:}  $(g, r)$, an auxiliary pair of $f$ w.r.t.\ $(F[t], \cP_h)$
\begin{enumerate}
\item[1.] $p \leftarrow f$, \,  $(g, r) \leftarrow 0, 0$, \, $(d, p_d) \leftarrow \deg(p), \lc(p)$

\vspace{-0.1cm}
\item[2.]
{\sc while} $d \ge m$ {\sc do} 

\vspace{-0.2cm}
\begin{itemize}
\item[]  {\sc if} $\nu_\infty(h) < 0$ {\sc then} $(u,v) \leftarrow a_m^{-1}  p_d  t^{d-m}, \, 0$ 
\item[]   {\sc else} $(\alpha, \beta) \leftarrow$ an R-pair  of $p_d$ w.r.t.\ $\Phi_{a_m}$, \, $(u, v) \leftarrow \alpha t^{d-m},  \beta t^d$ 
\item[] 
$ (g, r) \leftarrow g  + u,  r+v$, \, 
$p \leftarrow p - \cP_h(u)-v$, \,
$(d, p_d) \leftarrow \deg(p), \lc(p)$  \hfill ($^*${\sl by \eqref{EQ:paux_case}}$^*$)
\end{itemize}


\item[3.] {\sc return} $(g, r+p)$
\end{enumerate}
\end{alg}

Next, we define a finite sequence $L$ that uniquely determines an echelon sequence of $I_h$.

\begin{define} \label{DEF:iseq}
Let $m=\Deg(h)$ and $\lambda$ be the type of $I_h$. 
We define an {\em initial sequence} $L$ of~$I_h$ as follows.
\begin{enumerate}
\item If $\lambda=0$, then set $ L:=\emptyset$.

\item Assume that $\lambda \neq 0$. Fix $(\tilde{\sigma}, \sigma)$ and $(\tilde{\tau}, \tau)$ to 
be the first and second R-pairs associated to $(F[t], \cP_h)$, respectively,
and fix an element $\theta_\sigma \in \Theta$ such that $\theta_\sigma^*( \sigma) \neq 0$.
\begin{enumerate}
\item[2.1.] If $h \in F$, then set $L$ to be  the four-term sequence: 
\begin{equation} \label{EQ:four}
\lambda, \, (\tilde{\sigma}, \sigma), \, (\tilde{\tau}, \tau), \, \theta_\sigma.
\end{equation} 
\item[2.2.] Assume further that $h \in F(t) \setminus F$.
\begin{enumerate}
\item[2.2.1.] 
If $\theta_\sigma^*(i \sigma + \tau ) \neq 0$ for all $i \in \bN$, then set $L$ to be the five-term sequence:
\begin{equation} \label{EQ:five}
\lambda, \, (\tilde{\sigma}, \sigma), \, (\tilde{\tau}, \tau), \, \theta_\sigma, \,  \left(p_0, \cP_h(p_0), \, \theta_0 t^{d_0} \right),
\end{equation}
where the last member is given by Corollary \ref{COR:pivot0} (i).

\item[2.2.2.] If $\theta_\sigma^*(j \sigma + \tau) = 0$ for some $j \in \bN$ but $j \sigma + \tau \neq 0$, then set $L$ to be the 
seven-term sequence:
\begin{equation} \label{EQ:seven}
\lambda, \, (\tilde{\sigma}, \sigma), \, (\tilde{\tau}, \tau), \, \theta_\sigma, \, j, \,
\left(p_0, \cP_h(p_0), \theta_0 t^{d_0} \right), \, \left(p_{j}, \cP_h(p_{j}), \theta t^{m+j-1} \right),
\end{equation}
where the last two members are given by Corollary \ref{COR:pivot0} (ii).

\item[2.2.3.] If $j \sigma + \tau =0$ for some $j \in \bN$, then set $L$ to be the $(j+6)$-term 
sequence: 
\begin{align} 
& \,\, \lambda, \, (\tilde{\sigma}, \sigma), \, (\tilde{\tau}, \tau), \, \theta_\sigma, \, j,  \nonumber \\
& \left(q, \cP_h(q), \theta t^{d} \right), \left(p_0, \cP_h(p_0), 
\theta_0 t^{d_0} \right), \left\{\left(p_{i}, \cP_h(p_{i}), \theta_\sigma t^{m+i-1} \right)\right\}_{i \in [j-1]}, \label{EQ:last}
\end{align}
where the last $j+1$ members are given by Corollary \ref{COR:pivot0} (iii).
\end{enumerate}
\end{enumerate}
\end{enumerate}
\end{define}
An initial sequence $L$ with $L \neq \emptyset$  uniquely  determines an echelon sequence $E$ of $I_h$ by \eqref{EQ:prel1} and Algorithm \ref{ALG:paux}. 
Moreover, we do not recompute $\cP_h(p_i)$ according to \eqref{EQ:prel2}.
\begin{define} \label{DEF:icomp}
Let $L$ and $E$ be given above. 
The complement of $\im(\cP_h)$
induced by $E$ is also called the {\em complement induced by $L$}.
\end{define}

The next algorithm computes initial sequences.

\begin{alg}
{\sc InitSeq} \label{ALG:iseq}

\smallskip \noindent
{\sc Input:}  the data in Convention \ref{CON:next2}

\smallskip \noindent
{\sc Output:}  an initial sequence of $I_h$

\begin{enumerate}
\item[1.] $\lambda \leftarrow$ the type of $I_h$, {\sc if} $ \lambda=0$ {\sc then} {\sc return} $\emptyset$ \hfill  ($^*$$ I_h = \{0\}$$^*$)

\vspace{-0.2cm}
\item[2.] $(\tilde{\sigma}, \sigma) \leftarrow$ a first R-pair associated to $(F[t], \cP_h)$ 

$(\tilde{\tau}, \tau) \leftarrow$ a second R-pair associated to $(F[t], \cP_h)$

$\theta_\sigma \leftarrow$ an element of $\Theta$ s.t.\ $ \theta^*_\sigma(v) \neq 0$, \, $L_0 \leftarrow \lambda, \, (\tilde{\sigma}, \sigma), \, (\tilde{\tau}, \tau), \, \theta_\sigma$

\vspace{-0.2cm}
\item[3.] {\sc if} $h \in F$ {\sc then} {\sc return} $L_0$ \hfill  ($^*${\sl by Corollary \ref{COR:m=0} and \eqref{EQ:four}}$^*$)

\vspace{-0.2cm}
\item[4.]
$p_0 \leftarrow \lambda$, \, $(d_0, l_0) \leftarrow \deg\left(\cP_h(p_0) \right),   \lc\left(\cP_h(p_0)\right)$

$\theta_0 \leftarrow$ an element of $\Theta$ s.t.\ $ \theta^*_0(l_0) \neq 0$, \, $\bfomega_0 \leftarrow \theta_0 t^{d_0}$, \, 
$j \leftarrow - \theta_\sigma^*(\tau)/\theta^*_\sigma(\sigma)$, \, $k \leftarrow - \tau/ \sigma$

\vspace{-0.2cm}
\item[5.] {\sc if} $j \notin \bN$ {\sc then} {\sc return} $L_0, \left(p_0, \cP_h(p_0), \bfomega_0 \right)$  
\hfill  ($^*${\sl by Corollary \ref{COR:pivot0} (i)} and \eqref{EQ:five}$^*$)

\vspace{-0.2cm}
\item[6.] {\sc if} $k \notin \bN$ {\sc then} compute $p_j$ and $\cP_h(p_j)$ by \eqref{EQ:prel1}, \eqref{EQ:prel2} and Algorithm \ref{ALG:paux}
\begin{itemize}

\vspace{-0,2cm}
\item[] $l \leftarrow$ $\lc \left(\cP_h(p_j)\right)$, \, $\theta$ $\leftarrow$  an element of $\Theta$ s.t.\ $\theta^*(l) \neq 0$
\item[] {\sc return} $L_0, j,  \left(p_0, \cP_h(p_0), \bfomega_0 \right),    
\left(p_j, \cP_h(p_j), \theta t^{m+j-1} \right)$

\hfill ($^*${\sl by Corollary \ref{COR:pivot0} (ii) and \eqref{EQ:seven}}$^*$)
\end{itemize}

\vspace{-0.2cm}
\item[7.]
\begin{itemize}
\item[7.1.] compute $p_1, \cP_h(p_1),  \ldots, p_{j-1}, \cP_h(p_{j-1}), p_j, \cP(p_j)$ by \eqref{EQ:prel1}, \eqref{EQ:prel2} and Algorithm \ref{ALG:paux}

\item[7.2.] $(q, Q)  \leftarrow p_j, \cP_h(p_j)$

\item[7.3.] {\sc for} $i$ {\sc from} $j-1$ {\sc to} $0$ {\sc by} $-1$ {\sc do}
\begin{itemize}
\item[] {\sc if} $i>0$ {\sc then} $\bfomega_i \leftarrow  \theta_\sigma t^{m+i-1}$
\item[] $c \leftarrow \bfomega_i^*\left(\cP_h(p_i) \right)^{-1} \cdot \bfomega_i^*(Q),$ \, $(q, Q) \leftarrow q - cp_i, Q - c \cP_h(p_i)$
\end{itemize}
\item[7.4.] $(d, l) \leftarrow \deg(Q), \lc(Q)$, $\theta$ $\leftarrow$  an element of $\Theta$ s.t.\ $\theta^*(l) \neq 0$
\item[7.5.] {\sc return} $L_0, j, \left(q, Q, \theta t^d \right), \left(p_0, \cP_h(p_0), \bfomega_0 \right),
\left(p_1, \cP_h(p_1), \bfomega_1 \right), \ldots, \left(p_{j-1}, \cP_h(p_{j-1}), \bfomega_{j-1} \right)$

\hfill  ($^*${\sl by Corollary \ref{COR:pivot0} (iii) and \eqref{EQ:last}}$^*$)
\end{itemize}
\end{enumerate}
\end{alg}
The correctness of this algorithm follows directly from Definition \ref{DEF:iseq},  Corollaries \ref{COR:m=0}
and \ref{COR:pivot0}.
Note that $j$ and $k$ in step 4 are equal when both of them are positive integers, justifying the use of $j$
in step 7.

Given an initial sequence $L$ with $L \neq \emptyset$ and a positive integer $l$, we can compute the first~$l$ members of
the echelon sequence $E$ by the following algorithm.
\begin{alg}
{\sc PrimEchSeq} \label{ALG:pseq}

\smallskip \noindent
{\sc Input:}   $l \in \bN$, the data in Convention \ref{CON:next2}, and an initial sequence $L$ of $I_h$ with $L \neq \emptyset$

\smallskip \noindent
{\sc Output:} the first $l$ members of the echelon sequence induced by $L$

\smallskip \noindent
{\sc Remark:} In the following pseudo-code, $p_0=L[1]$, $p_i$ and $\cP_h(p_i)$ are computed by \eqref{EQ:prel1}, \eqref{EQ:prel2} and Algorithm \ref{ALG:paux} for $i>0$. 
\end{alg}
\begin{enumerate}
\item[1.] $\lambda \leftarrow L[1]$, $(\tilde{\sigma}, \sigma) \leftarrow L[2]$, $(\tilde{\tau}, \tau) \leftarrow L[3]$, $\theta_\sigma \leftarrow  L[4]$

\vspace{-0.2cm}
\item[2.] {\sc if} $\len(L)=4$ {\sc then} {\sc return} $\left\{ \left(p_i, \cP_h(p_i), \theta_\sigma t^{i-1} \right) \right\}_{i \in [l]}$ 
\hfill ($^*${\sl by Corollary \ref{COR:m=0}} and \eqref{EQ:four}$^*$)

\vspace{-0.2cm}
\item[3.] {\sc if} $\len(L)=5$ {\sc then} 
{\sc return} $L[5]$, $\left\{ \left(p_i, \cP_h(p_i), \theta_\sigma t^{m+i-1} \right) \right\}_{i \in [l-1]}$

\hfill ($^*${\sl by Corollary \ref{COR:pivot0} (i) and \eqref{EQ:five}}$^*$)

\vspace{-0.2cm}
\item[4.] $j \leftarrow L[5]$
\item[5.] {\sc if} $\len(L)=7$ {\sc and} $j \sigma + \tau \neq 0$ {\sc then}  $(q_1, Q_1, \bfomega_1) \leftarrow L[6]$

\vspace{-0.2cm}
\begin{itemize}
\item[] {\sc for} $i$ {\sc from} $2$ {\sc to} $l$ {\sc do}
\item[] \quad {\sc if} $i \neq j+1$ {\sc then}
$(q_i, Q_i, \bfomega_i) \leftarrow p_{i-1}, \cP_h(p_{i-1}), \theta_\sigma t^{m+i-2}$
{\sc else} $(q_i, Q_i, \bfomega_i) \leftarrow L[7]$ 

{\sc return} $(q_1, Q_1, \bfomega_1), (q_2, Q_2, \bfomega_2), \ldots, (q_l, Q_l, \bfomega_l)$  \hfill ($^*${\sl by Corollary \ref{COR:pivot0} (ii) and \eqref{EQ:seven}}$^*$)

\end{itemize}


\vspace{-0,2cm}
\item[6.]
\begin{itemize}
\item[6.1.] {\sc for} $i$ {\sc from} $1$ {\sc to} $j+1$ {\sc do} $(q_i, Q_i, \bfomega_i) \leftarrow L[5+i]$ 
\item[6.2.] {\sc for} $i$ {\sc from} $j+2$ {\sc to} $l$ {\sc do}
$(q_i, Q_i, \bfomega_i) \leftarrow p_{i-1}, \cP_h(p_{i-1}), \theta_\sigma t^{m+i-2}$
\item[6.3.] {\sc return} $(q_1, Q_1, \bfomega_1), (q_2, Q_2, \bfomega_2), \ldots, (q_l, Q_l, \bfomega_l)$  \quad ($^*${\sl by Corollary \ref{COR:pivot0} (iii) and \eqref{EQ:last}$^*$)}
\end{itemize}

\end{enumerate}

The following algorithm computes the respective projections of a given element in $F[t]$ to $\im(\cP_h)$ and its complement induced
by an initial sequence.
\begin{alg}
{\sc PrimProj} \label{ALG:pproj}

\smallskip \noindent
{\sc Input:}  $f \in F[t]$, the data in Convention \ref{CON:next2}, and an initial sequence $L$ of $I_h$ with $L \neq \emptyset$

\smallskip \noindent
{\sc output:} $(g, r) \in F[t] \times W$ s.t.\ $f = \cP_h(g) + r$, where $W$ is the complement induced by $L$
\begin{enumerate}
\item[1.] $\lambda  \leftarrow L[1]$ \hfill  ($^*${\sl retrieve the type}$^*$)

\vspace{-0.2cm}
\item[2.]  $(g, r) \leftarrow$ an auxiliary pair of $f$ computed by Algorithm \ref{ALG:paux}  \hfill ($^*${\sl auxiliary reduction}$^*$)

\vspace{-0.2cm}
\item[3.] {\sc if} $\lambda=0$ {\sc or} $r=0$ {\sc then} {\sc return} $(g, r)$   \hfill ($^*$$I_h = \{0\}$ {\sl or the trivial projection}$^*$)

\vspace{-0.2cm}
\item[4.] use Algorithm \ref{ALG:pseq} to compute the first $k$ members
$(q_1, Q_1, \bfomega_1),  \ldots, (q_k, Q_k, \bfomega_k)$
in the echelon sequence $E$ induced by $L$ s.t.\ all remaining  members in
$E$ are of degrees $>\deg(r)$

\vspace{-0.2cm}
\item[5.] {\sc for} $i$ {\sc from} $k$ {\sc to} $1$ {\sc by} $-1$ {\sc do}
$c\leftarrow  \bfomega_i^*(Q_i)^{-1} \bfomega_i^*(r)$, \, $(g, r)  \leftarrow g + c q_i, r - c Q_i$ \hfill ($^*${\sl elimination}$^*$)

\vspace{-0.2cm}
\item[6.] {\sc return} $(g,r)$
\end{enumerate}
\end{alg}
The first four steps of the above algorithm are evidently correct. Step 5 is justified by Lemma \ref{LM:elim}
and its proof. In step 6, $r$ belongs to the complement of $\im(\cP_h)$
induced by $L$, because $\bfomega_i^*(r)=0$ for all $i \in [k]$ and the pivot of the $l$th
member does not appear in $r$ for all $l>k$, as guaranteed by our choice of $k$ in step 4.

\subsection{Algorithms in the hyperexponential case} \label{SUBSECT:ALG:halg}
In this subsection, $t$ is hyperexponential over $F$, $h \in F(t)$ is $t$-normalized, $V_h$ is the
auxiliary space given in Definition \ref{DEF:haux}, and $J_h = \im(\cP_h) \cap V_h$.
The elements $\lambda_k$ and $\mu_l$ of $F$ are defined in Remark \ref{RE:hred} (ii) and (iv), respectively.

The first algorithm implements the auxiliary reduction in $F[t, t^{-1}]$. It is based on 
the congruences in \eqref{EQ:haux_case+} and \eqref{EQ:haux_case-}, and the proof of Proposition \ref{PROP:haux}.
\begin{alg}
{\sc HyperexpAuxRed} \label{ALG:haux}

\smallskip \noindent
{\sc Input:} $f \in F[t, t^{-1}]$ and the data in Convention \ref{CON:next2}

\smallskip \noindent
{\sc Output:}  $(g, r)$, an auxiliary pair of $f$ w.r.t.\ $(F[t,t^{-1}], \cP_h)$
\begin{enumerate}
\item[1.] $p \leftarrow f^+$, $(k, p_k) \leftarrow \hdeg(p), \hc(p),$ \, $q  \leftarrow f^-$, \,
 $(l, q_l) \leftarrow \tdeg(q), \tc(q)$, \, $(g, r) \leftarrow 0,0$

\vspace{-0.2cm}
\item[2.] {\sc while} $k \ge m$ {\sc do}
\begin{itemize}
\item[]  {\sc if} $\nu_\infty(h) < 0$ {\sc then} 
$(u,v)  \leftarrow a_m^{-1} p_k t^{k-m}, 0$, \, 
$g \leftarrow g + u$  
\item[] {\sc else} 
\item[] \quad $(\alpha, \beta) \leftarrow$ an R-pair  of $p_k$ w.r.t.\ $\Phi_{\lambda_{k-m}}$,  \, 
$(u,v) \leftarrow \alpha t^{k-m},  \beta  t^k$, \, $(g,r) \leftarrow g + u, \, r + v$ 
\item[] $p \leftarrow p - \cP_h(u) -v$, \, $(k, p_k) \leftarrow \hdeg(p), \hc(p)$ \hfill ($^*${\sl by \eqref{EQ:haux_case+}}$^*$)
\end{itemize}

\vspace{-0.2cm}
\item[3.] {\sc while} $l < 0$ {\sc do}
\begin{itemize}
\item[]  {\sc if} $\nu_t(h) < 0$ {\sc then} $(u,v) \leftarrow a_0^{-1} q_l t^{l}, 0$, \, 
$g  \leftarrow g + u$  
\item[] {\sc else} 
\item[] \quad $(\alpha, \beta) \leftarrow $  an R-pair of $b_0^{-1} q_l$ w.r.t.\ $\Phi_{\mu_l}$, \, 
$(u,v) \leftarrow  \alpha t^{l}, b_0 \beta t^l$, \, $(g, r) \leftarrow g + u, r + v$ 
\item[] $q \leftarrow q - \cP_h(u) -v$,  $(l, q_l) \leftarrow \tdeg(q), \tc(q)$ \hfill  ($^*${\sl by \eqref{EQ:haux_case-}} $^*$)
\end{itemize}

\vspace{-0.2cm}  
\item[4.] {\sc return} $(g, r+p+q)$
\end{enumerate}
\end{alg}

Next, we determine whether $J_h = \{0\}$, and find an echelon sequence of $J_h$ with respect to~$(F[t, t^{-1}], \cP_h)$ if $J_h \neq \{0\}$.
This algorithm is based on Propositions \ref{PROP:hbasis} and \ref{PROP:hechelon}.
\begin{alg}
{\sc HyperexpEchSeq} \label{ALG:hseq}

\smallskip \noindent
{\sc Input:} the data in Convention \ref{CON:next2}

\smallskip \noindent
{\sc Output:} $\emptyset$ if $J_h=\{0\}$, and an echelon sequence of $J_h$, otherwise 
\begin{enumerate}
\item[1.] $M \leftarrow$ the type of $J_h$

\vspace{-0.2cm}
\item[2.] {\sc if} $M = 0$ {\sc then} {\sc return} $\emptyset$ \hfill ($^*$$J_h=\{0\}$$^*$)

\vspace{-0.2cm}
\item[3.] ($^*$$\dim(J_h)=1$$^*$)

{\sc if} $\len(M)=1$ {\sc then} $(k,\sigma) \leftarrow M[1]$  \hfill ($^*${\sl find the type}$^*$)

\vspace{-0.2cm}  
\begin{itemize}
\item[]    $(g,r) \leftarrow$ the auxiliary pair of $\cP_h(\sigma t^k)$ computed by Algorithm \ref{ALG:haux}
\item[]   $(p, P) \leftarrow ut^k-g, r$, \, $\theta \leftarrow$ an element of $\Theta$ s.t. $\hc(P) \notin \ker(\theta^*)$
\item[]   {\sc return} $\left(p, P, \theta t^{\hdeg(P)} \right)$ \hfill ($^*${\sl by Proposition \ref{PROP:hechelon} (i)}$^*$)
\end{itemize}

\vspace{-0.2cm}
\item[4.] ($^*$$\dim(J_h)=2$$^*$) $(k, \sigma) \leftarrow M[1]$, \,  $(l, \tau) \leftarrow M[2]$ 
\hfill ($^*${\sl find the type}$^*$)

\vspace{-0.2cm}
\begin{itemize}
\item[4.1.] $(g_k,r_k) \leftarrow$ the auxiliary pair of $\mathcal{P}_h(\sigma  t^k)$ computed by Algorithm \ref{ALG:haux}

$(g_l, r_l) \leftarrow$ the auxiliary pair of $\mathcal{P}_h(\tau t^l)$ computed by Algorithm \ref{ALG:haux}

$(p_k, P_k) \leftarrow  \sigma t^k-g_k, r_k$,  \, $(p_l, P_l) \leftarrow \tau t^l-g_l, r_l$ \hfill
($^*${\sl by Proposition \ref{PROP:hbasis} (ii)}$^*$)
\item[4.2.] $d_k \leftarrow \hdeg(P_k)$, \, $\theta_k \leftarrow$ an element of $\Theta$ s.t.\ $\hc(P_k) \notin \ker(\theta_k^*)$, \,
$\bfomega_k \leftarrow \theta_k t^{d_k}$

$c \leftarrow \bfomega_k^*(P_k)^{-1} \bfomega_k^*(P_l)$, \, $(q, Q) \leftarrow p_l - c p_k, P_l - c P_k$, \,
$d \leftarrow \hdeg(Q)$

$\theta \leftarrow$ an element of $\Theta$ s.t.\ $\hc(Q) \notin \ker(\theta^*)$,
$\bfomega \leftarrow \theta t^{d}$, {\sc return} $(q, Q, \bfomega), (p_k, P_k, \bfomega_k)$

\hfill ($^*${\sl by Proposition \ref{PROP:hechelon} (ii)}$^*$)
\end{itemize}
\end{enumerate}
\end{alg}

The last algorithm computes the respective projections of an element in $F[t, t^{-1}]$ to $\im(\cP_h)$ and the complement induced
by an echelon sequence of $J_h$. It is simpler than Algorithm \ref{ALG:pproj}, because an echelon sequence of $J_h$ has at most
two members. We therefore present only its formal specification.
\begin{alg}
{\sc HyperexpProj} \label{ALG:hproj}

\smallskip \noindent
{\sc Input:} $f \in F[t, t^{-1}]$, the data in Convention \ref{CON:next2}, and an echelon sequence $E$ of $J_h$

\smallskip \noindent
{\sc Output:}  $(g, r) \in F[t, t^{-1}] \times W$ s.t.\ $f = \cP_h(g) + r$, where $W$ is the complement induced by $E$

\end{alg}

\end{document}